\documentclass[cleveref,english]{lipics-v2021}
\usepackage[T1]{fontenc}
\usepackage[utf8]{inputenc}
\usepackage{amsmath,amsfonts,amssymb,amsthm}
\usepackage[dvipsnames,svgnames, table]{xcolor}
\usepackage{graphicx}
\usepackage{enumerate}
\usepackage{enumitem}
\usepackage{todonotes}
\hypersetup{hidelinks,breaklinks}

\usepackage{doi}

\usepackage{diagbox}
\usepackage{array} 
\usepackage{arydshln} 
\usepackage{makecell}

\usepackage{mathcommand}
\usepackage[notion, quotation, electronic]{knowledge}

\usepackage{algorithm, algorithmicx, algpseudocode,setspace}
\usepackage{bbm}

\usepackage{tikz}
\usetikzlibrary{automata, snakes, positioning, arrows,arrows.meta,calc,decorations,decorations.markings,math,shapes.geometric}

\tikzset{
	>=stealth',
	-={stealth',ultra tdeck,scale=3} 
	node distance=1cm, 
	every state/.style={thick}, 
	initial text=$ $, 
}

\title{The Complexity of Simplifying $\omega$-Automata through the Alternating Cycle Decomposition}
\titlerunning{The Complexity of Simplifying $\omega$-Automata through the ACD}

\author{Antonio Casares}{LaBRI, Université de Bordeaux, France \and University of Warsaw, Poland \and \url{https://antonio-casares.github.io/}}{antoniocasares@mimuw.edu.pl}{https://orcid.org/0000-0002-6539-2020}{}

\author{Corto Mascle}{LaBRI, Université de Bordeaux, France \and \url{https://corto-mascle.github.io/}}{corto.mascle@labri.fr}{}{}

\authorrunning{A. Casares and C. Mascle} \Copyright{Antonio Casares and Corto Mascle}

\category{} 
\nolinenumbers 

\funding{Antonio Casares: Supported by the Polish National Science Centre (NCN) grant ``Polynomial finite state computation'' (2022/46/A/ST6/00072).}

\hideLIPIcs

\newcommand\restr[2]{{
		\left.\kern-\nulldelimiterspace 
		#1 
		\littletaller 
		\right|_{#2} 
}}

\newcommand{\littletaller}{\mathchoice{\vphantom{\big|}}{}{}{}}

\newcommand{\ts}{\textsuperscript}

\newrobustcmd\inv[1]{#1^{-1}}
\newcommand{\quotient}[2]{{\raisebox{0.1em}{$#1$\hspace{-1mm}}\left/\hspace{-0.5mm}\raisebox{-.1em}{{\scriptsize$#2$}}\right.}}

\newcommand{\tand}{\text{ and }}

\newcommand{\tin}{\text{ in }}
\newcommand{\tif}{\text{ if }}

\newcommand{\NP}{\ensuremath{\mathsf{NP}}}
\newcommand{\NPc}{\ensuremath{\mathsf{NP}\hyphen\mathrm{complete}}}

\newcommand{\PTimeFull}{\ensuremath{\mathsf{PTIME}}}
\newcommand{\coNP}{\ensuremath{\mathsf{coNP}}}

\newcommand{\PSPACE}{\ensuremath{\mathsf{PSPACE}}}

\newcommand{\for}{\ensuremath{\mathsf{for}}}

\DeclareMathAlphabet{\mathpzc}{OT1}{pzc}{m}{it}

\newrobustcmd{\NN}{\mathbb{N}}
\newrobustcmd{\ZZ}{\mathbb{Z}}
\newrobustcmd{\QQ}{\mathbb{Q}}
\newrobustcmd{\RR}{\mathbb{R}}
\newrobustcmd{\CC}{\mathbb{C}}
\newrobustcmd{\WW}{\mathbb{W}}

\newrobustcmd{\I}{\mathcal{I}}
\newrobustcmd{\F}{\mathcal{F}}
\newrobustcmd{\D}{\mathcal{D}}
\newrobustcmd{\N}{\mathcal{N}}
\newrobustcmd{\G}{\mathcal{G}}

\renewcommand{\L}{\mathcal{L}}
\newrobustcmd{\M}{\mathcal{M}}
\newrobustcmd{\Q}{\mathcal{Q}}
\newrobustcmd{\C}{\mathcal{C}}
\newrobustcmd{\A}{\mathcal{A}}
\newrobustcmd{\B}{\mathcal{B}}
\newrobustcmd{\Z}{\mathcal{Z}}
\newrobustcmd{\R}{\mathcal{R}}
\newrobustcmd{\T}{\mathcal{T}}
\newrobustcmd{\U}{\mathcal{U}}
\newrobustcmd{\W}{\mathcal{W}}

\renewcommand{\P}{\mathcal{P}}

\renewcommand{\O}{\mathcal{O}}

\renewcommand{\S}{\mathcal{S}}

\newrobustcmd{\kk}{\kappa}
\newrobustcmd{\uu}{\upsilon}
\newrobustcmd{\dd}{\delta}

\renewcommand{\ss}{\sigma}

\newrobustcmd{\rr}{\rho}

\renewcommand{\aa}{\alpha}

\newrobustcmd{\bb}{\beta}

\newrobustcmd{\oo}{\omega}
\newrobustcmd{\pp}{\varphi}

\renewcommand{\gg}{\gamma}
\newrobustcmd{\ee}{\varepsilon}

\renewcommand{\SS}{\Sigma}
\newrobustcmd{\GG}{\Gamma}
\newrobustcmd{\DD}{\Delta}

\knowledgerenewmathcommand\nu{\cmdkl{\LaTeXnu}}
\knowledgenewmathcommand\nuAcd{\cmdkl{\LaTeXnu}}
\knowledgerenewmathcommand\eta{\cmdkl{\LaTeXeta}}

\newrobustcmd{\pow}[1]{2^{#1}}
\newrobustcmd{\powplus}[1]{\kl[\powplus]{2^{#1}_{+}}}
\knowledge{\powplus}{notion}
\newrobustcmd\SigmaInfty{\kl[\SigmaInfty]{\Sigma^\infty}}
\knowledge{\SigmaInfty}[E'^\infty | E^\infty | \GammaInfty|\powInf|\GG^\infty]{notion}
\newrobustcmd\GammaInfty{\kl[\SigmaInfty]{\Gamma^\infty}}
\newrobustcmd{\powInf}[1]{\kl[\SigmaInfty]{#1^{\infty}}}
\newcommand{\disjUnion}{\sqcup}

\newrobustcmd{\restSubsets}[2]{\kl[\restSubsets]{\restr{#1}{#2}}}
\knowledge{\restSubsets}{notion}
\newrobustcmd{\restLang}[2]{\kl[\restLang]{\restr{#1}{#2}}}
\knowledge{\restLang}{notion}

\newrobustcmd{\partialF}{\mathrel{\kl[\partialF]{\rightharpoonup}}}
\knowledge{\partialF}{notion}

\newrobustcmd{\complTS}[1]{\kl[\complTS]{\overline{#1}}}
\knowledge{\complTS}{notion}
\newrobustcmd{\complSet}[1]{\kl[\complSet]{\overline{#1}}}
\knowledge{\complSet}{notion}

\newrobustcmd{\infEL}{\mathtt{Inf}}
\newrobustcmd{\finEL}{\mathtt{Fin}}

\newrobustcmd{\emptyword}{\kl[\emptyword]{\varepsilon}}
\knowledge{\emptyword}{notion}

\newrobustcmd{\prefix}{\mathrel{\kl[\prefix]{\sqsubseteq}}}
\knowledge{\prefix}{notion}
\newrobustcmd{\nprefix}{\mathrel{\kl[\nprefix]\sqsubset}}
\knowledge{\nprefix}{notion}

\newrobustcmd{\first}{\kl[\first]{\mathsf{first}}}
\knowledge\first{notion}
\newrobustcmd{\last}{\kl[\last]{\mathsf{last}}}
\knowledge\last{notion}

\newrobustcmd{\minf}{\kl[\minf]{\mathsf{Inf}}}
\knowledge{\minf}{notion}
\newrobustcmd{\mocc}{\kl[\mocc]{\mathsf{Occ}}}
\knowledge{\mocc}{notion}

\newcommand{\re}[1]{\xrightarrow{#1}}
\usetikzlibrary{calc,decorations.pathmorphing,shapes} 

\newcounter{sarrow}
\newcommand\lrp[1]{%
	\stepcounter{sarrow}%
	\mathrel{\begin{tikzpicture}[baseline= {( $ (current bounding box.south) + (0,-0.5ex) $ )}]
			\node[inner sep=.5ex] (\thesarrow) {$\scriptstyle #1$};
			\path[draw,<-,decorate,
			decoration={zigzag,amplitude=0.7pt,segment length=1.2mm,pre=lineto,pre length=4pt}] 
			(\thesarrow.south east) -- (\thesarrow.south west);
	\end{tikzpicture}}%
}

\newrobustcmd\lrpResolver[2]{\kl[\lrpResolver]{
		\stepcounter{sarrow}%
		\mathrel{\hspace{1mm}\begin{tikzpicture}[baseline= {( $ (current bounding box.south) + (0,2.2mm) $ )}]
				\node[inner sep=.5ex] (\thesarrow) {$\scriptstyle #2$};
				\draw[<-, decorate,
				decoration={zigzag,amplitude=0.7pt,segment length=1.2mm,pre=lineto,pre length=4pt}] 
				(\thesarrow.south east) -- (\thesarrow.south west) node[below, pos=0.1, inner sep=1mm] {$\scriptstyle #1$};
			\end{tikzpicture}\hspace{0.7mm}}%
}}
\knowledge\lrpResolver{notion}

\newrobustcmd\lrpResolverMem[2]{\kl[\lrpResolverMem]{\lrpResolver{#1}{#2}}}
\knowledge\lrpResolverMem{notion}

\newrobustcmd\lrpAllResolver[2]{\kl[\lrpAllResolver]{
	\stepcounter{sarrow}%
	\mathrel{\hspace{1mm}\begin{tikzpicture}[baseline= {( $ (current bounding box.south) + (0,2.7mm) $ )}]
			\node[inner sep=.5ex] (\thesarrow) {$\scriptstyle #2$};
			\draw[<-, decorate,
			decoration={zigzag,amplitude=0.7pt,segment length=1.2mm,pre=lineto,pre length=4pt}] 
			(\thesarrow.south east) -- (\thesarrow.south west) node[below, pos=0.1, inner sep=1mm] {$\mathsmaller{\mathsmaller{\forall},\, #1}$};
	\end{tikzpicture}\hspace{0.7mm}}%
}}
\knowledge\lrpAllResolver{notion}

\newrobustcmd\lrpExistsResolver[2]{\kl[\lrpExistsResolver]{
		\stepcounter{sarrow}%
		\mathrel{\hspace{1mm}\begin{tikzpicture}[baseline= {( $ (current bounding box.south) + (0,2.7mm) $ )}]
				\node[inner sep=.5ex] (\thesarrow) {$\scriptstyle #2$};
				\draw[<-, decorate,
				decoration={zigzag,amplitude=0.7pt,segment length=1.2mm,pre=lineto,pre length=4pt}] 
				(\thesarrow.south east) -- (\thesarrow.south west) node[below, pos=0.1, inner sep=1mm] {$\mathsmaller{\mathsmaller{\exists},\, #1}$};
			\end{tikzpicture}\hspace{0.7mm}}%
}}
\knowledge\lrpExistsResolver{notion}

\newrobustcmd\lrpResolverWord[3]{\kl[\lrpResolverWord]{
		\stepcounter{sarrow}%
		\mathrel{\hspace{1mm}\begin{tikzpicture}[baseline= {( $ (current bounding box.south) + (0,2.7mm) $ )}]
				\node[inner sep=.5ex] (\thesarrow) {$\scriptstyle #3$};
				\draw[<-, decorate,
				decoration={zigzag,amplitude=0.7pt,segment length=1.2mm,pre=lineto,pre length=4pt}] 
				(\thesarrow.south east) -- (\thesarrow.south west) node[below, pos=0.1, inner sep=1mm] {$\scriptstyle #2,#1$};
			\end{tikzpicture}\hspace{0.7mm}}%
}}
\knowledge\lrpResolverWord{notion}

\newcommand{\init}{\mathsf{init}}
\newcommand{\col}{\mathsf{col}}

\newrobustcmd{\TS}{\mathcal{T\hspace{-1.1mm}S}}
\newrobustcmd{\colTS}{\kl[\colTS]{\mathsf{col}}}
\knowledge\colTS{notion}
\newrobustcmd{\colTSStates}{\kl[\colTSStates]{\mathsf{col}_{\mathsf{st}}}}
\knowledge\colTSStates{notion}

\newrobustcmd{\outputTS}{\kl[\outputTS]{\mathsf{col}}}
\knowledge\outputTS{notion}

\newrobustcmd{\Graph}[1]{\kl[\Graph]G_{#1}}
\knowledge\Graph{notion}

\newrobustcmd{\underlyingGraph}[1]{\kl[\underlyingGraph]G_{#1}}
\knowledge\underlyingGraph{notion}

\newrobustcmd{\macc}[1]{\kl[\macc]{\mathrm{Acc}}_{#1}}
\knowledge\macc{notion}

\newrobustcmd{\msource}{\kl[\msource]{\mathsf{source}}}
\knowledge{\msource}[source]{notion}
\newrobustcmd{\mtarget}{\kl[\mtarget]{\mathsf{target}}}
\knowledge{\mtarget}[target]{notion}

\newrobustcmd{\msourcePath}{\kl[\msourcePath]{\mathsf{source}}}
\knowledge{\msourcePath}[source@path]{notion}
\newrobustcmd{\mtargetPath}{\kl[\mtargetPath]{\mathsf{target}}}
\knowledge{\mtargetPath}[target@path]{notion}

\newrobustcmd{\PathSet}[2]{\kl[\PathSet]{\mathpzc{Path}_{#2}(#1)}}
\knowledge\PathSet[\Paths]{notion}
\newrobustcmd{\PathSetFin}[2]{\kl[\PathSetFin]{\mathpzc{Path}^{\mathsf{fin}}_{#2}(#1)}}
\knowledge\PathSetFin[\PathFin]{notion}
\newrobustcmd{\PathFin}[1]{\kl[\PathFin]{\mathpzc{Path}^{\mathsf{fin}}(#1)}}
\newrobustcmd{\Paths}[1]{\kl[\PathFin]{\mathpzc{Path}^{\oo}(#1)}}

\newrobustcmd{\Runs}[1]{\kl[\Runs]{\mathpzc{Run}(#1)}}
\knowledge\Runs{notion}
\newrobustcmd{\RunsFin}[1]{\kl[\RunsFin]{\mathpzc{Run}^{\mathsf{fin}}(#1)}}
\knowledge\RunsFin{notion}
\newrobustcmd{\RunsInfty}[1]{\kl[\RunsInfty]{\mathpzc{Run}^{\infty}(#1)}}
\knowledge\RunsInfty{notion}
\newrobustcmd{\PathSetInfty}[2]{\kl[\PathSetInfty]{\mathpzc{Path}^{\infty}_{#2}(#1)}}
\knowledge\PathSetInfty{notion}

\newrobustcmd{\mout}{\kl[\mout]{\mathsf{Out}}}
\knowledge\mout{notion}
\newrobustcmd{\mIn}{\kl[\mIn]{\mathsf{In}}}
\knowledge\mIn{notion}

\newrobustcmd{\initialTS}[2]{\kl[\initialTS]{#1_{#2}}}
\knowledge\initialTS{notion}
\newrobustcmd{\initialAut}[2]{\kl[\initialAut]{#1_{#2}}}
\knowledge\initialAut{notion}

\newcommand{\size}[1]{|#1|}
\newrobustcmd{\sizeAut}[1]{\kl[\sizeAut]{|#1|}}
\knowledge\sizeAut{notion}

\newrobustcmd{\Vrec}{V_\mathrm{rec}}
\newrobustcmd{\Vtrans}{V_\mathrm{trans}}

\newrobustcmd{\transAut}{\kl[\transAut]{\delta}}
\newrobustcmd{\transAutA}[1]{\kl[\transAutA]{\delta_{#1}}}
\knowledge\transAut[\transAutA]{notion}

\newcommand{\Prio}{\Pi}

\newrobustcmd{\letAut}{\kl[\letAut]{\mathsf{let}}}
\knowledge\letAut{notion}
\newcommand{\colAut}{\mathsf{col}}
\newrobustcmd{\colAutDet}{\kl[\colAutDet]{\mathsf{col}}}
\knowledge\colAutDet{notion}
\newrobustcmd{\colAutStates}{\kl[\colAutStates]{\mathsf{col}_{\mathsf{st}}}}
\knowledge\colAutStates{notion}

\newrobustcmd{\outputAut}{\kl[\outputAut]{\mathsf{col}}}
\knowledge\outputAut{notion}

\newrobustcmd{\compositionAut}{\mathbin{\kl[\compositionAut]{\ltimes}}}
\knowledge\compositionAut{notion}

\newrobustcmd{\edgesProduct}{\kl[\edgesProduct]E^\ltimes}
\knowledge\edgesProduct{notion}

\newrobustcmd\Lang[1]{\kl[\Lang]{\mathcal{L}(#1)}}
\knowledge\Lang{notion}

\newrobustcmd\LangTS[1]{\kl[\LangTS]{\mathcal{L}_{\mathpzc{Runs}}(#1)}}
\knowledge\LangTS{notion}

\newrobustcmd{\greenPair}{\mathfrak{g}}
\newrobustcmd{\redPair}{\mathfrak{r}}

\newrobustcmd{\Muller}[1]{\kl[\Muller]{\mathsf{Muller}(#1)}}
\newrobustcmd{\MullerC}[2]{\kl[\MullerC]{\mathsf{Muller}_{#2}(#1)}}
\knowledge{\Muller}[\MullerC]{notion}
\newrobustcmd{\Rabin}[1]{\kl[\Rabin]{\mathsf{Rabin}(#1)}}
\newrobustcmd{\RabinC}[2]{\kl[\RabinC]{\mathsf{Rabin}_{#2}(#1)}}
\knowledge{\Rabin}[\RabinC]{notion}
\newrobustcmd{\Streett}[1]{\kl[\Streett]{\mathsf{Streett}(#1)}}
\newrobustcmd{\StreettC}[2]{\kl[\StreettC]{\mathsf{Streett}_{#2}(#1)}}
\knowledge{\Streett}[\StreettC]{notion}
\knowledgenewrobustcmd{\GHC}[2]{\cmdkl{\mathsf{GH}}_{#2}(#1)}
\newrobustcmd{\parity}{\kl[\parity]{\mathsf{parity}}}
\knowledge{\parity}{notion}
\newrobustcmd{\Buchi}[1]{\kl[\Buchi]{\mathsf{Buchi}(#1)}}
\newrobustcmd{\BuchiC}[2]{\kl[\BuchiC]{\mathsf{Buchi}_{#2}(#1)}}
\knowledge{\Buchi}[\BuchiC]{notion}
\newrobustcmd{\genBuchi}[1]{\kl[\genBuchi]{\mathsf{genBuchi}(#1)}}
\newrobustcmd{\genBuchiC}[2]{\kl[\genBuchiC]{\mathsf{genBuchi}_{#2}(#1)}}
\knowledge{\genBuchi}[\genBuchiC]{notion}

\newrobustcmd{\coBuchi}[1]{\kl[\coBuchi]{\mathsf{coBuchi}(#1)}}
\newrobustcmd{\coBuchiC}[2]{\kl[\coBuchiC]{\mathsf{coBuchi}_{#2}(#1)}}
\knowledge{\coBuchi}[\coBuchiC]{notion}
\newrobustcmd{\gencoBuchi}[1]{\kl[\gencoBuchi]{\mathsf{genCoBuchi}(#1)}}
\newrobustcmd{\gencoBuchiC}[2]{\kl[\gencoBuchiC]{\mathsf{genCoBuchi}_{#2}(#1)}}
\knowledge{\gencoBuchi}[\gencoBuchiC]{notion}

\newrobustcmd{\Weak}[1]{\kl[\Weak]{\mathsf{Weak}_{#1}}}
\knowledge{\Weak}{notion}

\newrobustcmd{\WeakIndex}[1]{\kl[\WeakIndex]{\mathsf{Weak}_{#1}}}
\knowledge{\WeakIndex}{notion}

\newrobustcmd{\MullerFamily}[1]{\kl[\MullerFamily]{\F_{#1}}}
\knowledge{\MullerFamily}{notion}

\newrobustcmd{\impliesMuller}[2]{\mathbin{\kl[\impliesMuller]{#1 \rightarrow #2}}}
\knowledge{\impliesMuller}{notion}

\newrobustcmd{\equivAut}{\mathrel{\kl[\equivAut]{\simeq}}}
\knowledge{\equivAut}{notion}
\newrobustcmd{\isoTS}{\mathrel{\kl[\isoTS]{\simeq}}}
\knowledge{\isoTS}{notion}
\newrobustcmd{\equivCond}[1]{\mathrel{\kl[\equivCond]{\simeq}_{#1}}}
\knowledge{\equivCond}{notion}

\newrobustcmd{\cycles}[1]{\kl[\cycles]{\mathpzc{Cycles}(#1)}}
\knowledge\cycles{notion}
\newrobustcmd{\cyclesState}[2]{\kl[\cyclesState]{\mathpzc{Cycles}_{#2}(#1)}}
\knowledge\cyclesState{notion}
\newrobustcmd{\states}[1]{\kl[\states]{\mathsf{States}(#1)}}
\knowledge\states{notion}
\newrobustcmd{\localMuller}[2]{\kl[\localMuller]{\mathsf{LocalMuller}_{#2}(#1)}}
\knowledge\localMuller{notion}

\newrobustcmd{\rInit}{\kl[\rInit]{\resolv_{\mathsf{Init}}}}
\knowledge\rInit[\trInit]{notion}
\newrobustcmd{\trInit}{\kl[\rInit]{\tilde{\resolv}_{\mathsf{Init}}}}

\newrobustcmd{\rRuns}{\kl[\rRuns]\resolv_{\mathpzc{Runs}}}
\knowledge\rRuns[\rRunsOption]{notion}
\newrobustcmd{\rRunsOption}[1]{\kl[\rRunsOption]\resolv_{#1,\mathpzc{Runs}}}

\newrobustcmd{\ppRuns}{\kl[\ppRuns]\pp_{\mathpzc{Runs}}}
\newrobustcmd{\ppRunsP}[1]{\kl[\ppRuns]{#1}_{\mathpzc{Runs}}}
\knowledge\ppRuns[\ppRunsP]{notion}
\newrobustcmd{\Id}[1]{\kl[\Id]{\mathsf{Id}_{#1}}}
\knowledge\Id{notion}

\newrobustcmd{\autMorphism}[1]{\kl[\autMorphism]{\mathcal{A}_{#1}}}
\knowledge\autMorphism{notion}

\newrobustcmd\resolv{\mathsf{r}}

\newrobustcmd{\ancestor}{\mathrel{\kl[\ancestor]{\preceq}}}
\knowledge\ancestor{notion}
\newrobustcmd{\descendant}{\mathrel{\kl[\descendant]{\succeq}}}
\knowledge\descendant{notion}
\newrobustcmd{\roundnodes}{\kl[\roundnodes]{N_\bigcirc}}
\knowledge\roundnodes{notion}
\newrobustcmd{\squarenodes}{\kl[\squarenodes]{N_\Box}}
\knowledge\squarenodes{notion}
\newrobustcmd{\leaves}{\kl[\leaves]{\mathsf{Leaves}}}
\knowledge\leaves{notion}
\newrobustcmd{\children}{\kl[\children]{\mathsf{Children}}}
\knowledge\children{notion}
\newrobustcmd{\nextChild}{\kl[\nextChild]{\mathsf{Next}}}
\knowledge\nextChild{notion}
\newrobustcmd{\jump}{\kl[\jump]{\mathsf{Jump}}}
\knowledge\jump[intermediate node]{notion}
\newrobustcmd{\depth}{\kl[\depth]{\mathsf{Depth}}}
\knowledge\depth{notion}
\newrobustcmd{\orderTree}[1]{\kl[\orderTree]{\leq_{#1}}}
\knowledge\orderTree{notion}

\newrobustcmd{\pred}{\kl[\pred]{\mathsf{pred}}}
\knowledge\pred{notion}
\newrobustcmd{\predBranch}{\kl[\predBranch]{\mathsf{pred}^*}}
\knowledge\predBranch{notion}

\newrobustcmd\zielonkaTree[1]{\kl[\zielonkaTree]{\mathcal{Z}_{#1}}}
\knowledge\zielonkaTree{notion}
\newrobustcmd{\supp}{\kl[\supp]{\mathsf{Supp}}}
\knowledge\supp{notion}
\newrobustcmd{\parityNodes}{\kl[\parityNodes]{p_\Z}}
\knowledge\parityNodes{notion}
\newrobustcmd{\minparityZ}[1]{\kl[\minparityZ]{\min_{#1}}}
\knowledge\minparityZ{notion}
\newrobustcmd{\maxparityZ}[1]{\kl[\maxparityZ]{\max_{#1}}}
\knowledge\maxparityZ{notion}
\newrobustcmd{\memTree}[1]{\kl[\memTree]{\mathsf{rbw}(#1)}}
\knowledge\memTree{notion}

\newrobustcmd\zielonkaDAG[1]{\kl[\zielonkaDAG]{\mathcal{Z}\hyphen\small{\mathsf{DAG}}_{#1}}}
\knowledge\zielonkaDAG{notion}
\newrobustcmd{\ancestorDAG}{\mathrel{\kl[\ancestorDAG]{\preceq}}}
\knowledge\ancestorDAG{notion}
\knowledgenewrobustcmd\roundnodesDAG{D_{\bigcirc}}
\knowledgenewrobustcmd\squarenodesDAG{D_{\Box}}
\newrobustcmd\acdDAG[1]{\kl[\acdDAG]{\mathcal{ACD}\hyphen\small{\mathsf{DAG}}(#1)}}
\knowledge\acdDAG{notion}
\newrobustcmd\altDAG[1]{\kl[\altDAG]{\mathsf{AltDAG}(#1)}}
\knowledge\altDAG{notion}
\newrobustcmd\dagVertex[1]{\kl[\dagVertex]{\mathpzc{D}_{#1}}}
\knowledge\dagVertex{notion}

\newrobustcmd\zielonkaAutomaton[1]{\kl[\zielonkaAutomaton]{\mathcal{A}^{\mathsf{parity}}_{\mathcal{Z}_{#1}}}}
\knowledge\zielonkaAutomaton{notion}

\newrobustcmd\zielonkaHDAutomaton[1]{\kl[\zielonkaHDAutomaton]{\mathcal{A}^{\mathsf{Rabin}}_{\mathcal{Z}_{#1}}}}
\knowledge\zielonkaHDAutomaton{notion}
\newrobustcmd{\sizeHDRabin}{\memTree{\zielonkaTree{\F}{\SS}}}

\newrobustcmd\acd[1]{\kl[\acd]{\mathcal{ACD}(#1)}}
\knowledge\acd{notion}
\newcommand\acdNoP{\mathcal{ACD}}
\newrobustcmd\altTree[1]{\kl[\altTree]{\mathsf{AltTree}(#1)}}
\knowledge\altTree{notion}
\newrobustcmd\nuStates{\kl[\nuStates]{\nu_{\mathsf{States}}}}
\knowledge\nuStates{notion}
\newrobustcmd\nodesAcdCycle[1]{\kl[\nodesAcdCycle]{N_{#1}}}
\knowledge\nodesAcdCycle{notion}
\newrobustcmd\treeVertex[1]{\kl[\treeVertex]{\mathpzc{T}_{#1}}}
\knowledge\treeVertex{notion}
\newrobustcmd\nodesTreeVertex[1]{\kl[\nodesTreeVertex]{N_{#1}}}
\knowledge\nodesTreeVertex{notion}

\newrobustcmd\acdVertex[2]{\kl[\acdVertex]{\mathcal{ACD}_{(#1,#2)}}}
\knowledge\acdVertex{notion}

\knowledgenewrobustcmd\roundnodesq{N_{q,\bigcirc}}
\knowledgenewrobustcmd\squarenodesq{N_{q,\Box}}

\newrobustcmd\autCyclePreimage[1]{\kl[\autCyclePreimage]{\A_{(\inv{\pp},#1)}}}
\knowledge\autCyclePreimage{notion}

\newrobustcmd\unfold{\kl[\unfold]{\mathsf{Unfold}}}
\knowledge\unfold{notion}

\newrobustcmd\nodesAcd[1]{\kl[\nodesAcd]{\mathsf{Nodes}(\mathcal{ACD}_{#1})}}
\knowledge\nodesAcd{notion}
\newrobustcmd\nodesAcdRound[1]{\kl[\nodesAcdRound]{\mathsf{Nodes}_\bigcirc(\mathcal{ACD}_{#1})}}
\knowledge\nodesAcdRound{notion}
\newrobustcmd\nodesAcdSquare[1]{\kl[\nodesAcdSquare]{\mathsf{Nodes}_\Box(\mathcal{ACD}_{#1})}}
\knowledge\nodesAcdSquare{notion}

\newrobustcmd{\leavesAcd}{\kl[\leavesAcd]{\mathsf{Leaves}(\acd{\TS})}}
\knowledge\leavesAcd{notion}

\newrobustcmd{\suppAcd}{\kl[\suppAcd]{\mathsf{Supp}}}
\knowledge\suppAcd{notion}

\newrobustcmd{\parityNodesAcd}{\kl[\parityNodesAcd]p_{\mathcal{ACD}}}
\knowledge\parityNodesAcd{notion}
\newrobustcmd{\minparityAcd}[1]{\kl[\minparityAcd]{{\min}_{#1}}}
\knowledge\minparityAcd{notion}
\newrobustcmd{\maxparityAcd}[1]{\kl[\maxparityAcd]{{\max}_{#1}}}
\knowledge\maxparityAcd{notion}

\newrobustcmd{\coloursNodesAcd}{\kl[\coloursNodesAcd]\gg_{\mathcal{ACD}}}
\knowledge\coloursNodesAcd{notion}

\newrobustcmd{\minparityAcdVertex}[2]{\kl[\minparityAcdVertex]{{\min}_{(#1,#2)}}}
\knowledge\minparityAcdVertex{notion}
\newrobustcmd{\maxparityAcdVertex}[2]{\kl[\maxparityAcdVertex]{{\max}_{(#1,#2)}}}
\knowledge\maxparityAcdVertex{notion}

\newrobustcmd\acdParityTransform[1]{\kl[\acdParityTransform]{\P^{\mathsf{ACD}}_{#1}}}
\knowledge\acdParityTransform{notion}

\newrobustcmd\acdRabinTransform[1]{\kl[\acdRabinTransform]{\mathsf{ACD}_{\mathsf{Rabin}}(#1)}}
\knowledge\acdRabinTransform{notion}
\newrobustcmd\acdRabinTransformGFG[1]{\kl[\acdRabinTransformGFG]{\mathsf{ACD}_{\mathsf{Rabin}}^{\mathsf{game}}(#1)}}
\knowledge\acdRabinTransformGFG{notion}
	
\newrobustcmd{\ppAcd}{\kl[\ppAcd]{\pp_{\acdNoP}}}
\knowledge\ppAcd{notion}

\newrobustcmd{\VAsucc}{\kl[\VAsucc]{V_{\mathrm{A}\hyphen\mathrm{succ}}}}
\knowledge\VAsucc{notion}
\newrobustcmd{\Vnormal}{\kl[\Vnormal]{V_{\mathrm{normal}}}}
\knowledge\Vnormal{notion}
\newrobustcmd{\Apred}{\kl[\Apred]{\mathsf{pred}}}
\knowledge\Apred{notion}

\newrobustcmd{\pop}{\kl[\pop]{\mathtt{pop}}}
\knowledge\pop{notion}
\newrobustcmd{\push}{\kl[\push]{\mathtt{push}}}
\knowledge\push{notion}
\newrobustcmd{\maxInclusion}{\kl[\maxInclusion]{\mathtt{MaxInclusion}}}
\knowledge\maxInclusion{notion}
\newrobustcmd{\SCCDec}{\kl[\SCCDec]{\mathtt{SCC\hyphen Decomposition}}}
\knowledge\SCCDec{notion}

\newrobustcmd{\computeChildrenACD}{\kl[\computeChildrenACD]{\mathtt{ComputeChildren}}}
\knowledge\computeChildrenACD{notion}
\newrobustcmd{\computeChildrenTwoACD}{\kl[\computeChildrenTwoACD]{\mathtt{ComputeChildren2}}}
\knowledge\computeChildrenTwoACD{notion}

\newrobustcmd{\findIntSubcycles}{\kl[\findIntSubcycles]{\mathtt{FindInterestingSubcycles}}}
\knowledge\findIntSubcycles{notion}

\newcommand{\childrenList}{\mathsf{children}}
\newcommand{\newChildren}{\mathsf{newChildren}}

\newrobustcmd{\edges}{\kl[\edges]{\mathtt{Edges}}}
\knowledge\edges{notion}

\newrobustcmd{\computeMaxAltSubsets}{\kl[\computeMaxAltSubsets]{\mathtt{MaxAltSubsets}}}
\knowledge\computeMaxAltSubsets{notion}

\newcommand{\nodesStack}{\mathsf{nodesToTreat}}
\newcommand{\nodeACD}{\mathsf{n}}

\newrobustcmd{\costComputationAltSubsets}[1]{\kl[\costComputationAltSubsets]{\aa_{#1}}}
\knowledge\costComputationAltSubsets{notion}

\newrobustcmd{\doubleFact}[1]{\kl[\doubleFact]{#1!!}}
\knowledge\doubleFact{notion}

\newrobustcmd{\evenLetters}[1]{\kl[\evenLetters]{\mathsf{EvenLetters}_{#1}}}
\knowledge\evenLetters{notion}
\newrobustcmd{\smallDAGCondition}[1]{\kl[\smallDAGCondition]{\mathsf{MinOddAndSucc}_{#1}}}
\knowledge\smallDAGCondition{notion}

\newrobustcmd{\AlternatingSets}{\kl[\AlternatingSets]{\mathtt{AlternatingSets}}}
\knowledge\AlternatingSets{notion}
\newrobustcmd{\leavingLetters}{\kl[\leavingLetters]{\mathtt{LeavingLetters}}}
\knowledge\leavingLetters{notion}

\newrobustcmd{\minCol}{\kl[\minCol]{\mathtt{MinColour}}}
\knowledge\minCol{notion}
\newrobustcmd{\mletters}{\kl[\mletters]{\mathtt{Letters}}}
\knowledge\mletters{notion}

\newrobustcmd{\graphMuller}[1]{\kl[\graphMuller]{G_{#1}}}
\knowledge\graphMuller{notion}
\newrobustcmd{\graphMullerVert}[1]{\kl[\graphMullerVert]{V_{#1}}}
\knowledge\graphMullerVert{notion}
\newrobustcmd{\graphMullerEdge}[1]{\kl[\graphMullerEdge]{E_{#1}}}
\knowledge\graphMullerEdge{notion}

\newrobustcmd{\chromNum}{\kl[\chromNum]{\chi}}
\knowledge\chromNum{notion}

\newrobustcmd{\neighbourhood}[1]{\kl[\neighbourhood]{N[#1]}}
\knowledge\neighbourhood{notion}
\newrobustcmd{\neighbourhoodOpen}[1]{\kl[\neighbourhoodOpen]{n(#1)}}
\knowledge\neighbourhoodOpen{notion}

\newcommand{\graphcolouring}{\mathbf{c}}
\newrobustcmd{\env}{\kl[\env]{\mathsf{env}}}
\knowledge\env{notion}

\newrobustcmd{\langVertex}[1]{\kl[\langVertex]{L_{#1}}}
\knowledge\langVertex{notion}
\newrobustcmd{\langGraph}{\kl[\langGraph]{L_G}}
\knowledge\langGraph{notion}

\newrobustcmd{\langClique}{\kl[\langClique]{L_{\exists\mathrm{noRep}}}}
\knowledge\langClique{notion}

\newrobustcmd{\FTwoLetters}{\kl[\FTwoLetters]{\F_G}}
\knowledge\FTwoLetters{notion}

\newrobustcmd{\Stab}{\kl[\Stab]{\mathsf{Stab}}}
\knowledge\Stab{notion}
\newrobustcmd{\adj}{\kl[\adj]{\mathsf{adj}}}
\knowledge\adj{notion}
\newrobustcmd{\neigh}{\kl[\neigh]{\mathsf{neigh}}}
\knowledge\neigh{notion}

\newrobustcmd{\autChromG}{\kl[\autChromG]{\mathcal{A}_G}}
\knowledge\autChromG{notion}

\newrobustcmd{\FHalfLetters}[1]{\kl[\FHalfLetters]{\F_{#1}}}
\knowledge\FHalfLetters{notion}
\newrobustcmd{\LHalfLetters}[1]{\kl[\LHalfLetters]{L_{#1}}}
\knowledge\LHalfLetters{notion}

\newrobustcmd{\minRabin}[1]{\ensuremath{\kl[\minRabin]{\mathsf{minDetRabin}(#1)}}}
\knowledge\minRabin{notion}

\NewDocumentCommand{\pbMinRabin}{}{\kl[\pbMinRabin]{\normalfont{\textsc{\small{Minimisation of Rabin automata}}}}}
\knowledge\pbMinRabin[\pbMinGenBuchi|\pbMinGenCoBuchi|\pbMinMuller|\pbMinHDRabin]{notion}
\NewDocumentCommand{\pbMinMuller}{}{\kl[\pbMinMuller]{\normalfont{\textsc{\small{Minimisation of Muller automata}}}}}
\NewDocumentCommand{\pbMinGenBuchi}{}{\kl[\pbMinGenBuchi]{\normalfont{\textsc{\small{Minimisation of generalised B\"uchi automata}}}}}
\NewDocumentCommand{\pbMinGenCoBuchi}{}{\kl[\pbMinGenCoBuchi]{\normalfont{\textsc{\small{Minimisation of generalised coB\"uchi automata}}}}}

\NewDocumentCommand{\pbMinHDRabin}{}{\kl[\pbMinHDRabin]{\normalfont{\textsc{\small{Minimisation of HD Rabin automata}}}}}
\NewDocumentCommand{\pbMinHDParity}{}{\kl[\pbMinHDParity]{\normalfont{\textsc{\small{Minimisation of HD parity automata}}}}}

\NewDocumentCommand{\pbMinCoBuchi}{}{\kl[\pbMinCoBuchi]{\normalfont{\textsc{\small{Minimisation of coB\"uchi automata}}}}}
\knowledge\pbMinCoBuchi{notion}

\NewDocumentCommand{\pbMinBuchiMuller}{}{\kl[\pbMinBuchiMuller]{\normalfont{\textsc{\small{Minimisation B\"uchi-to-Muller}}}}}
\knowledge\pbMinBuchiMuller{notion}

\NewDocumentCommand{\pbTransformGenBuchi}{}{\kl[\pbTransformGenBuchi]{\normalfont{\textsc{\small{Transform-GenBuchi}}}}}
\knowledge\pbTransformGenBuchi{notion}

\NewDocumentCommand{\pbTransformGenBuchiStates}{}{\kl[\pbTransformGenBuchiStates]{\normalfont{\textsc{\small{Transform-GenBuchi-to-States}}}}}
\knowledge\pbTransformGenBuchiStates{notion}
\NewDocumentCommand{\pbTransToStates}{}{\kl[\pbTransToStates]{\normalfont{\textsc{\small{Transitions-to-States}}}}}
\knowledge\pbTransToStates{notion}

\NewDocumentCommand{\pbChromMem}{}{\kl[\pbChromMem]{\normalfont{\textsc{\small{Chromatic-Memory}}}}}
\knowledge\pbChromMem{notion}

\NewDocumentCommand{\pbColorMinML}{}{\kl[\pbColorMinML]{\normalfont{\textsc{\small{Colour-Minimisation-ML}}}}}
\knowledge\pbColorMinML{notion}
\NewDocumentCommand{\pbColorMinAut}{}{\kl[\pbColorMinAut]{\normalfont{\textsc{\small{Colour-Minimisation-Aut}}}}}
\knowledge\pbColorMinAut{notion}
\NewDocumentCommand{\pbMultiColourMinAut}{}{\kl[\pbMultiColorMinAut]{\normalfont{\textsc{\small{Multiple-Colour-Minimisation-Aut}}}}}
\knowledge\pbMultiColorMinAut{notion}

\NewDocumentCommand{\pbRabinPairMinML}{}{\kl[\pbRabinPairMinML]{\normalfont{\textsc{\small{Rabin-Pair-Minimisation-ML}}}}}
\knowledge\pbRabinPairMinML{notion}
\NewDocumentCommand{\pbRabinPairMinAut}{}{\kl[\pbRabinPairMinAut]{\normalfont{\textsc{\small{Rabin-Pair-Minimisation-Aut}}}}}
\knowledge\pbRabinPairMinAut{notion}

\NewDocumentCommand{\pbRabinTypeness}{}{\kl[\pbRabinTypeness]{\normalfont{\textsc{\small{Rabin-Typeness}}}}}
\knowledge\pbRabinTypeness{notion}

\NewDocumentCommand{\pbChromNum}{}{\kl[\pbChromNum]{\normalfont{\textsc{\small{Chromatic number}}}}}
\knowledge\pbChromNum{notion}
\NewDocumentCommand{\pbThreeCol}{}{\kl[\pbThreeCol]{\normalfont{\textsc{\small{3-colorability}}}}}
\knowledge\pbThreeCol{notion}
\NewDocumentCommand{\pbVertexCover}{}{\kl[\pbVertexCover]{\normalfont{\textsc{\small{Vertex Cover}}}}}
\knowledge\pbVertexCover{notion}
\NewDocumentCommand{\pbEdgeToEdgeCover}{}{\kl[\pbEdgeToEdgeCover]{\normalfont{\textsc{\small{Edge-to-edge Cover}}}}}
\knowledge\pbEdgeToEdgeCover{notion}
\NewDocumentCommand{\pbMaxClique}{}{\kl[\pbMaxClique]{\normalfont{\textsc{\small{Max-Clique}}}}}
\knowledge\pbMaxClique{notion}

\NewDocumentCommand{\pbunSAT}{}{\kl[\pbunSAT]{\normalfont{\textsc{\small{un-SAT}}}}}
\knowledge\pbunSAT{notion}

\NewDocumentCommand{\pbSubsetCover}{}{\kl[\pbSubsetCover]{\normalfont{\textsc{\small{Subset Cover}}}}}
\knowledge\pbSubsetCover{notion}

\newrobustcmd{\minCover}[1]{\kl[\minCover]{\mathsf{minCover}(#1)}}
\knowledge\minCover{notion}
\newrobustcmd{\lineGraph}[1]{\kl[\lineGraph]{\mathsf{Line}(#1)}}
\knowledge\lineGraph{notion}

\newrobustcmd{\sizeDet}{\kl[\sizeDet]{\mathsf{sizeDet}}}
\knowledge\sizeDet{notion}

\newrobustcmd{\transMem}{\kl[\transMem]{\mu}}
\knowledge\transMem{notion}

\newrobustcmd{\nextmove}{\kl[\nextmove]{\mathsf{next}\hyphen\mathsf{move}}}
\knowledge\nextmove{notion}

\newrobustcmd{\stratMem}{\kl[\stratMem]{\mathsf{strat}_{\M}}}
\knowledge\stratMem{notion}

\newrobustcmd{\prodMem}[1]{\mathbin{\kl[\prodMem]{\lhd_{#1}}}}
\knowledge\prodMem{notion}

\newcommand{\mem}{\mathsf{mem}}
\newcommand{\gen}{\mathsf{gen}}
\newcommand{\chrom}{\mathsf{chr}}
\newcommand{\vertic}{\mathsf{vrt}}
\newcommand{\arInd}{\mathsf{ArInd}}
\newcommand{\epsFree}{\ee\hyphen \mathsf{free}}

\newrobustcmd{\memGen}{\kl[\memGen]{\mathsf{mem}_{\mathsf{gen}}}}
\knowledge\memGen[\mem_\gen^X|\mem_Y]{notion}
\newrobustcmd{\memChrom}{\kl[\memChrom]{\mathsf{mem}_{\mathsf{chr}}}}
\knowledge\memChrom[\mem_\chrom^X]{notion}
\newrobustcmd{\memArenaInd}{\kl[\memArenaInd]{\mathsf{mem}_{\mathsf{ArInd}}}}
\knowledge\memArenaInd[\mem_\arInd^X]{notion}

\newrobustcmd{\memGenEpsFree}{\kl[\memGenEpsFree]{\mathsf{mem}_{\mathsf{gen}}^{\ee\hyphen \mathsf{free}}}}
\knowledge\memGenEpsFree[\mem_Y^{\epsFree}]{notion}
\newrobustcmd{\memChromEpsFree}{\kl[\memChromEpsFree]{\mathsf{mem}_{\mathsf{chr}}^{\ee\hyphen \mathsf{free}}}}
\knowledge\memChromEpsFree{notion}
\newrobustcmd{\memArenaIndEpsFree}{\kl[\memArenaIndEpsFree]{\mathsf{mem}_{\mathsf{ArInd}}^{\ee\hyphen \mathsf{free}}}}
\knowledge\memArenaIndEpsFree{notion}

\newrobustcmd{\memGenVert}{\kl[\memGenVert]{\mathsf{mem}_{\mathsf{gen}}^{\mathsf{vrt}}}}
\knowledge\memGenVert[\mem_Y^\vertic]{notion}
\newrobustcmd{\memChromVert}{\kl[\memChromVert]{\mathsf{mem}_{\mathsf{chr}}^{\mathsf{vrt}}}}
\knowledge\memChromVert{notion}
\newrobustcmd{\memArenaIndVert}{\kl[\memArenaIndVert]{\mathsf{mem}_{\mathsf{ArInd}}^{\mathsf{vrt}}}}
\knowledge\memArenaIndVert{notion}

\newrobustcmd{\Valpha}[1]{\kl[\Valpha]V_{#1}}
\knowledge\Valpha{notion}
\newrobustcmd{\Aalpha}[1]{\kl[\Aalpha]A_{#1}}
\knowledge\Aalpha{notion}

\newrobustcmd{\FO}{\ensuremath{\mathrm{FO}}}
\newcommand{\LTL}{\ensuremath{\mathrm{LTL}}}

\newrobustcmd{\MSO}{\ensuremath{\mathrm{MSO}}}
\newrobustcmd{\SOneS}{\ensuremath{\mathrm{S1S}}}
\newrobustcmd{\STwoS}{\ensuremath{\mathrm{S2S}}}

\newcommand{\set}[1]{\{#1\}}

\let\ab\allowbreak
\mathchardef\hyphen=45 

\newcommand{\nocontentsline}[3]{}
\newcommand{\tocless}[2]{\bgroup\let\addcontentsline=\nocontentsline#1{#2}\egroup}

\usepackage[super]{nth}

\knowledge{notion}
| notion@example
|definition@example

\knowledge{notion}
  | preorder
  | preordered
  | preorders

\knowledge{notion}
| $\omega $-word
| $\omega $-words
| infinite words
| infinite word

\knowledge{notion}
| word
| words

\knowledge{notion}
 | empty word

\knowledge{notion}
| prefix

\knowledge{notion}
  | factor

\knowledge{notion}
  | partial mapping

\knowledge{notion}
  | state complexity

\knowledge{notion}
| directed graph
| graph
| graphs
| undirected graph
| (undirected) graph
| undirected 
| undirected graphs

\knowledge{notion}
| size@graph

\knowledge{notion}  
 |  connected@und
 | connected
  
\knowledge{notion}
| simple

\knowledge{notion}
  | sink
  | sinks

\knowledge{notion}
| final@SCC
| final SCC

\knowledge{notion}
| recurrent

\knowledge{notion}
| transient
| transient vertex

\knowledge{notion}
  | labelled graph
  | labelled@graph
  | labelled graphs
  | labelled
  | (labelled)
  | labelling
  
\knowledge{notion}
| pointed graph
| pointed graphs
| (pointed) graphs

\knowledge{notion}
  | directed acyclic graph
  | DAG
  | DAGs
  
\knowledge{notion}  
  | subDAG rooted at a node
  | subDAGs rooted at
  | subDAG rooted at

\knowledge{notion}
  | $A$-labelled DAG
  | labelled@DAG
  | labelled DAG

\knowledge{notion}
  | size of the representation
  | representation of@Muller
  | representations@Muller
  | representation of@Muller
  | representations@Muller
  | represented@Muller
  | representation
  | representation@Muller
  | represent@aut
  | representation@aut
  | represented
  | formalisms@repMuller
  | represent
  | representation
  | representations
  | represent@Muller
  
\knowledge{notion}
  | explicit@Muller
  | explicitly@Muller
  | explicit representation@Muller
  | represented explicitly@Muller
  | Explicit Muller
  | explicit representation
  | explicit
  | explicit list@Muller
  | colour-explicit@Muller
  | Colour-explicit@Muller
  | colour-explicitly@Muller
  | colour-explicit representation@Muller
  | colour-explicit
  | colour-explicitly

\knowledge{notion}
  | Emerson-Lei
  | Emerson-Lei condition
  | Emerson-Lei conditions

\knowledge{notion}
  | size of the representation of@parity
  | size of its representation@parity
  | size of the representation@parity
  | representation@parity
  | representations@parity

\knowledge{notion}
  | size of the representation of@genBuchi
  | size of its representation@genBuchi
  | size of the representation@genBuchi
  
\knowledge{notion}
  | size of its representation@Rabin
  | size of the representation@Rabin
  | representation of the automata@Rabin
  | representation@Rabin

\knowledge{notion}
  | size@aut

\knowledge{notion}
 | transition system
 | transition systems
 | transitions systems
 | Transition system
 | Transition systems
 
\knowledge{notion}
  | state-based transition systems
  | state-based TS
  | state-based@TS
  | transition-based@trans
  | transition-based@TS
  | defined over states
  | edge-coloured 
  
\knowledge{notion}
 | Vertex-coloured
 | vertex-coloured
 | vertex-coloured game
 | vertex-coloured games
 | Vertex-coloured@game

\knowledge{notion}
| state-based automata
| state-based 
| State-based
| state-based@aut
| transition-based
| transition-based@aut
| transition-based automata
| state-based automaton
| over transitions
| transition-based acceptance
| transition-based automaton 
| state-based acceptance
   
\knowledge{notion}
 | TS

\knowledge{notion}
| underlying graph@TS
| underlying graph
| underlying graphs

\knowledge{notion}
| underlying graph@game
| underlying graphs@game
| game graphs
| game graph
| underlying graph@games

\knowledge{notion}
| underlying graph@aut
| underlying graphs@aut

\knowledge{notion}
| initial vertex@TS
| initial vertices@TS
| initial vertex
| initial set of vertices@TS
| initial vertices

\knowledge{notion}
| initial vertex@pointed
| initial vertex@game
| initial vertices@game
| initial vertices@PG
| initial vertex@PG
| initial@game
| initial vertices@games
| initial@games

\knowledge{notion}
| initial state
| initial@state
| initial states
| initial states@aut
| initial state@aut
| initial

\knowledge{notion}
  | uncoloured edges
  | $\ee $-edges
  | uncoloured
  | uncoloured edge
  | $\ee $-edges@games
  | coloured
  | Uncoloured edges
  | uncoloured@games
  | uncoloured@game
  | uncoloured vertices
  | totally coloured
  
\knowledge{notion}  
  | uncoloured edges@TS
  | uncoloured@TS

\knowledge{notion}
 | colours@TS 
 | output colours
 | output colour
 | colours
 | colouring@aut
 | sets of colours
 | colour
 | colour@output
 | output colours@aut
 | output colour@aut
 | colours@aut
 | couleurs@aut
 | colours@out
 | colours@outAut
 | colour@aut
 | colouring@TS
 | colouring function@TS
 | edge-colouring@TS
 | edge-colouring 
 | colouring function
 | output alphabet

\knowledge{notion}
| output
| output@run
| output@TS

\knowledge{notion}
 | accessible@vertex
 | reachable
 | accessible
 | unreachable
 
\knowledge{notion}
 | accessible from
 | accessible@fromVertex

\knowledge{notion}
 | accessible@set
 | accessible@sets

\knowledge{notion}
| accessible part

\knowledge{notion}
| accessible part of $\TS $ from $v$
| accessible part of $\P $ from $v_P$
| accessible part from a vertex
| accessible part of $\widetilde {\TS }$ from $\tilde {v}$

\knowledge{notion}
| accepting@run
| accepting
| acceptance@run
| accepting run
| accepting runs
| acceptance@runs
| accepting@runTS
| acceptance of runs
| acceptants
| acceptants@aut
| accepting@runAut
| accepting@aut
| acceptance@runAut
| accepting run@aut
| accepting runs@aut
| acceptance@runAut
| acceptance of runs@aut

\knowledge{notion}
| rejecting@run
| rejecting run
| rejecting@runTS
| rejecting

\knowledge{notion}
| rejecting@runAut
| rejecting run@aut

\knowledge{notion}
| labelled transition system
| labelled transition systems
| (labelled) transition system
| labelling with input letters
| (labelled) TS

\knowledge{notion}
| vertex-labelling
| vertex-labellings
| vertex-labelled

\knowledge{notion}
| edge-labelled
| edge-labelled@graph

\knowledge{notion}
| size@transSys

\knowledge{notion}
  | product construction
  | composition
  | product automaton
  | product@aut
  | composing
  | composition of automata
  | compositions
  | composition@aut
  | composition operation@aut
  
\knowledge{notion}
  | suitable for transformations
  | suitability for transformations
  | suitable for composition

\knowledge{notion}
| acceptance condition
| acceptance conditions
| conditions
| condition
| condition@accep
| condition@acc
| conditions@acc
| acceptance condition@TS

\knowledge{notion}  
| acceptance condition@aut
| condition d'acceptation
| condition d'accep\-ta\-tion
| condition@aut

\knowledge{notion}
  | prefix-independent
  | prefix-independence

\knowledge{notion}
| automaton
| (non-deterministic) 
| automata
| (non-deterministic) automaton
| automate
| automates
| $\oo $-automates
| $\oo $-automata
| automata over infinite words

\knowledge{notion}
| input letters
| input alphabet
| input letter

\knowledge{notion}
| deterministic
| deterministic automaton
| deterministic automata
| Deterministic
| non-deterministic automata
| non-deterministic
| non-deterministic automaton
| non-determinism
| determinism 
| déterministes

\knowledge{notion}
| unambiguous

\knowledge{notion}
 | strongly unambiguous
 | (strongly) unambiguity
 | strongly@unamb

\knowledge{notion}
 | equivalent@automaton
 | equivalent@automata
 | equivalent@aut
 | equivalence@aut
 | equivalent to@aut
 | equivalent to
 
\knowledge{notion}
| transition function

\knowledge{notion}
  | $\SS $-complete
  | complete
  | completeness
  | completeness@aut

\knowledge{notion}
| run over $w$
| run@automaton
| run over
| run@aut
| over@run
| run $\rr $ over
| runs@aut
| run of@aut
| run over a word
| run over $w_kw_{k+1}\dots $
| run over $w^\oo $
| run over@aut
| calculs@run
| run@transSys
| infinite run
| infinite runs
| finite runs
| finite run
| runs
| run
| run from $v_0$
| runs@TS

\knowledge{notion}
  | output of@aut
  | output@aut

\knowledge{notion}
  | \GG ^\infty

\knowledge{notion}
  | transitions@aut

\knowledge{notion}
  | underlying graph@aut

\knowledge{notion}
  | subautomaton induced by
  | subautomaton@induced
  | subautomaton
  | -subautomaton
  | subautomata
  | automaton induced by
  | induces@aut
  | induced by@aut
  | subautomaton induced
  | subautomata induced by subgraphs
  | induced by@subautomaton

\knowledge{notion}
| accepted@word
| accepts@aut
| accepted by@aut
| accepted
| accepted by
| accepted@aut
| rejected@aut
| accepting word
| accepted words

\knowledge{notion}
| language accepted
| accepting@automaton
| recognising
| recognise
| recognises
| recognising@automaton
| recognising@aut
| recognised
| recognises@aut
| language recognised
| languages recognised
| language
| languages
| reconnaissant
| reconnaître

\knowledge{notion}
  | $\oo $-regular language
  | $\oo $-regular languages
  | $\oo $-regular@language
  | $\oo $-regular automata
  | $\oo $-regular
  | $\oo $-regularity
  | $\oo $-regular objective
  | $\omega $-regular objectives
  | $\omega $-regular languages
  | non-$\oo $-regular
  | $\oo $-regular conditions
  | langage $\oo $-régulier
  | $\oo $-réguliers
 | $\oo $-regular objectives

\knowledge{notion}
  | history-deterministic
  | HD-parity automaton
  | history-determinism
  | history-deterministic@aut
  | HD@aut
  | HD
  | history-determinism@aut
  | HD automaton
  | HD automata
  | (history-)deterministic
  | history-deterministic automaton
  | History-deterministic automata
  | history-deterministic automata
  | History-
  | deterministic@hist
  | History-deterministic
  | déterminisme en histoire
  | déterministes en histoire
  | history-
  | History-determinism

\knowledge{notion}
  | good-for-games
  | good-for-gameness
| GFG
| GFG automata
 
\knowledge{notion}
  | pruned 
  | prune
 
\knowledge{notion}
  | Determinisable by pruning (DBP)
  | determinisable by pruning

\knowledge{notion}
| resolver
| resolvers
| resolver@aut
| resolvers@aut
| resolve@aut

\knowledge{notion}
| sound@aut
| sound resolver@aut
| soundness@resolverAut
| sound@resolverAut
| sound resolver
| sound@resAut

\knowledge{notion}
| run induced by@aut
| induced by $\resolv$@aut
  | induced run of $\resolv $
  | run induced over $w$
  | induced run
  | induced run of $\resolv $ over $u_0$
  | induced run of $\resolv $ over
  | run induced by $\resolv $
  | induced by@resolver
  | run induced by@res
  | run induced over $w$@res
  | induced by@resolv
  | induce@resolv
  | induced by

\knowledge{notion}
  | reachable using some sound resolver
  | reachable using@resolver
  | reachable using a sound resolver 
  | reachable using this resolver
  | reachable using $\resolv $
  | reachable using the resolver $(r_0, \resolv )$
  | reachable using $(r_0, \resolv )$

\knowledge{notion}  
 | reachable@memoryRes 
 | reachable using $\resolv $@memory

\knowledge{notion}
| memory structure@aut
| memory structure@resolver
| finite memory resolvers
| finite memory structures@resolv
| finite memory@res

\knowledge{notion}
| implements@resolver
| resolver implemented by
| implemented by@resolver
| implemented by@memRes
| implemented by@resMem
| implemented by@aut
| implemented by@resolvMem

\knowledge{notion}
  | composition@memory
  | product game@mem
  | product@memoryRes

\knowledge{notion}
 | Streett language
 | Streett@language
 | Streett languages
 | Streett language associated to
 | Streett condition
 | Streett conditions
 | Streett

\knowledge{notion}
  | Rabin pairs
  | Rabin pair

\knowledge{notion}
  | Rabin language associated to
  | Rabin
  | Rabin condition
  | Rabin language
  | Rabin@language
  | Rabin languages
  
\knowledge{notion}
  | green in@rabPair
  | green@rabin

\knowledge{notion}
  | red in@rabPair
  | red@rabin
  | red colours@Rabin

\knowledge{notion}
  | orange in@rabPair

\knowledge{notion}
| satisfies@Rabin
| satisfied by@Rabin
| accepted@Rabin
| rejected@Rabin
| accepted by@Rabin
| rejected by@Rabin
| accepted by the Rabin pair
| accepted by@Rabinpair
| accepted@Rabinpair
| accepts@RabinPair
| accepted by@RabinPair
| accept@rabinPair
| accepted by@rabinPair
| accepting@RabinPair
| accepted@Streett

\knowledge{notion}
  | priority
  | priorities

\knowledge{notion}
  | $[c,d]$-parity language
  | $[0,1]$-parity languages
  | $[1,2]$-parity languages
  | $[d_{\min }, d_{\max }]$-parity language
  | $[\minparityZ {\F },\maxparityZ {\F }]$-parity language
  | $[1,d]$-parity
  | $[0,1]$-parity
  | $[1,2]$-parity
  | $[d_{\min },d_{\max }]$-parity
  
\knowledge{notion}
  | parity@language
  | parity language
  | Parity conditions
  | parity condition
  | parity
  | parity languages
  | parity conditions
  | Parity languages
  | Parity
  | parity objective
  | parity objectives
  | parité
  | parite

\knowledge{notion}
| accepting set@Muller
| accepting sets@Muller
| accepting set
| accepting@set
| accepting@Muller
| accepting@Muller
| accepting sets

\knowledge{notion}
| rejecting@Muller
| rejecting set@Muller
| rejecting sets@Muller
| rejecting@set
| rejecting sets

\knowledge{notion}
  | Muller language associated to
  | Muller
  | language associated@Muller
  | language associated to@Muller
  | Muller language associated to the family
  | Muller language
  | Muller@condition
  | Muller languages
  | Muller conditions
  | Muller condition
  | Muller objectives
  | langages de Muller
  | Muller objective
  | Muller acceptance condition

\knowledge{notion}
  | $\C$ condition
  | $\C$ acceptance condition
  | $\C$ transition system
  | $\C$ automaton
  | $\C$ language
  | $\C$ languages
  | $\C$ automata
  | Muller automaton
  | Rabin automaton
  | Streett automaton
  | parity automaton
  | Muller automata
  | Rabin automata
  | Streett automata
  | parity automata

\knowledge{notion}
  | cycles
  | cycle
  | cycle@aut
  | Cycles
  | subcycles
  | subcycle
  
\knowledge{notion}
  | lassos

\knowledge{notion}
  | cycle@TS
  | cycles@TS

\knowledge{notion}
| states of the cycle
| containing@cycle
| cycle over $(q,m_i)$
| cycles over $q$
| set of states@cycle
| over $v$@cycle

\knowledge{notion}
| state in common

\knowledge{notion}
| accepting@cycle
| accepting cycle
| accepting cycles
| accepting@cycles

\knowledge{notion}
| rejecting@cycle
| rejecting cycle
| rejecting cycles
| rejecting@cycles

\knowledge{notion}
  | accepting SCC
  | accepting@SCC

\knowledge{notion}
  | rejecting SCC
  | rejecting@SCC

\knowledge{notion}
| $d$-flower
| $d$-flowers
| flowers
  | $k$-flower
  | $k+1$-flower
  | $(d+1)$-flower
  | $(d-1)$-flower

\knowledge{notion}
| positive flower
| positive@flower

\knowledge{notion}
  | negative@flower
  | negative flowers

\knowledge{notion}
| positive@SCC

\knowledge{notion}
| negative@SCC

\knowledge{notion}
  | language of accepting runs of

\knowledge{notion}
  | deterministic parity hierarchy

\knowledge{notion}
  | parity index
  | parity index is not less
  | parity index@atMost  
  | parity index at most
  | parity indices

\knowledge{notion}
  | parity index at least

\knowledge{notion}
  | parity-index-tight

\knowledge{notion}
| automaton structure

\knowledge{notion}
  | parity automaton on top of
  | parity automaton on top
  | automaton on top of
  
\knowledge{notion}
 | acceptance condition on top of it
 | on top of@accCond
 | acceptance condition@onTop
 | on top of
 | acceptance condition on top of
 | acceptance condition on top
 | equipped with@onTopOf
 | on top of@cond
 
\knowledge{notion}
| on top of@TS

\knowledge{notion}
  | tree
  | trees

\knowledge{notion}
  | ancestor relation
  | incomparable@descendant
  | incomparable@ancestor
  | minimal@ancestor
  | maximal@ancestor

\knowledge{notion}
  | ancestor
  | ancestors
  | above

\knowledge{notion}
  | descendant
  | descendants
  | descendent
  | descendants@nodes
  | below
  | below@tree
  | lowest@tree
  | below@nodes
  | strict descendant

\knowledge{notion}
  | $A$-labelled (ordered) tree
  | $\powplus {\SS }$-labelled tree
  | labelled tree
  | $A$-labelled tree

\knowledge{notion}
| morphism@graphs
| morphism of graphs
| morphism of (labelled) graphs

\knowledge{notion}
| morphism of pointed graphs
| morphism of (labelled) pointed graphs

\knowledge{notion}
| morphism of labelled graphs

\knowledge{notion}
  | weak morphism of transition systems
  | (weak) morphism@transSys
  | weak morphism
  | weak morphism of labelled transition systems
  | weak morphism of (labelled) transition systems
  | labelled transition systems@mapping
  | (weak) morphism
  | weak morphisms
  | weak morphism of (labelled) TS
  | weak morphism of TS

\knowledge{notion}
   | morphism of automata
   | of automata@morphism
   | automata@mapping
   | morphism@aut
   | of automata@morph
  
\knowledge{notion}
   | morphism of games
   | of games@morphism  
  
\knowledge{notion}
  | isomorphism@TS
  | isomorphism of transition systems
  | isomorphism

\knowledge{notion}
  | isomorphic@TS
  
\knowledge{notion}  
  | morphism of labelled transition system
  | morphism of (labelled) transition systems
  | morphism of labelled transition systems
  | morphism of labelled TS
  | morphism@transSys
  | morphisms@transSys
  | morphism of transition systems
  | morphism
  | morphisms
  | morphisms of transition systems
  | morphism@TS
  | morphism of TS
  | morphisms@TS
  | morphism of transition system
  | morphismes
  | Morphisms

\knowledge{notion}
  | surjective
  | surjectivity
  | surjective morphism

\knowledge{notion}
  | preserves the acceptance
  | preserve the acceptance of runs
  | preserves the acceptance of runs

\knowledge{notion}
  | preserves accepting runs
  | preserve rejecting runs
  | preserve accepting runs

\knowledge{notion}
  | Locally surjective
  | locally surjective
  | local surjectivity
  | locally surjective morphism
  | locally surjective morphisms
  | Locally surjective morphisms

\knowledge{notion}
  | Locally injective
  | locally injective
  | local injectivity

\knowledge{notion}
  | Locally bijective
  | locally bijective
  | locally bijective morphism
  | Locally bijective morphisms
  | local bijectivity
  | local bijective
  | locally bijective morphisms
  | bijective@loc

\knowledge{notion}
| automaton of the morphism
 | automaton@morphism
 | automaton of morphism $\pp $
 |automaton@autMorphism
 | of the morphism@aut

\knowledge{notion}
  | resolver simulating $\pp $
  | resolver simulating
  | resolver@HDmorphism
  | simulating@resolver
  | resolver@morphism
  | resolver@mapping
  | simulating
  | resolver@HDmapping
  | resolver@morph
  | simulate@morph
  | resolver@HDmap
  
\knowledge{notion}
  | sound@resolverMorphism
  | sound resolver@morph
  | sound@morph
  | soundness@morph

\knowledge{notion}
| run induced by@morphism 
| run induced by@morph
| induced by $\resolv$@morph
| induced by@resMorph
| induced by@morph

\knowledge{notion}
  | history-deterministic mapping
  | history-deterministic mappings
  | history-deterministic@mapping
  | HD mappings
  | HD mapping
  | HD@mapping
  | history-determinism@mapping
  | history-determinism@morp
  | history-determinism@morph

\knowledge{notion}
  | HD mapping of labelled transition systems

\knowledge{notion}
| history-deterministic-for-games
  | HD@games
  | HD-for-games mapping
  | HD-for-games
  | HD-for-games mappings
  | mappings@HDForGames
  
\knowledge{notion}
  | resolver sound for $\G $
  | sound for $\G $
  | sound for $\TS $
  | sound for $\acdRabinTransformGFG {\G }$

\knowledge{notion}
  | consistent with@resolver
  | consistent with@mapping

\knowledge{notion}
  | Zielonka tree automaton

\knowledge{notion}
  | Zielonka tree
  | Zielonka trees
  | tree@ZT
  
\knowledge{notion}
  | Zielonka DAG
  | Zielonka@DAG
  | DAG@Zielonka

\knowledge{notion}
  | round nodes
  | round node
  | round
  | round@nodes
  | round@node

\knowledge{notion}
  | square nodes
  | square node
  | square
  | square@nodes
  | square@node
  | square@tree

\knowledge{notion}
  | round-branching width

\knowledge{notion}
  | ZT-HD-Rabin-automaton

\knowledge{notion}
  | ZT-parity-automaton
  | Zielonka-tree-parity-automaton

\knowledge{notion}
| round-branching-width

\knowledge{notion}
  | letter game
  | Letter game

\knowledge{notion}
  | attractor to $X$
  | attractor to $0$
  | attractor to $x$

\knowledge{notion}
| $x$-attractor decomposition
| attractor decompositions
| attractor decomposition
| $0$-attractor decomposition
| $2$-attractor decomposition

\knowledge{notion}
| $x$-recursive attractor decomposition
| $0$-recursive attractor decomposition
| $x+2$-recursive attractor decomposition
| decomposition@attrRecursive

\knowledge{notion}
| full attractor decomposition

\knowledge{notion}
  | $\SS _i$-region
  | $\SS _{i'}$-region
  | $\SS _i$-region@simple
  | $\SS _i$-region of the attractor decomposition
  | $\SS _i$-region@simpleAttr
  | $\SS _i$-regions

\knowledge{notion}
  | $1$-avoiding region
  | $x+1$-avoiding region
  | $y$-avoiding region@simpleAttr
  | $y+2$-avoiding region@simple
  | $3$-avoiding region
  
\knowledge{notion}
 | $\SS _i$-region of $\attrDec {\G '}$
 | $\SS _i$-region of $\attrDec {\letterGame {\A }}$
 | $\SS _i$-regions@recursive
 | $\SS _i$-region@rec
 | $\SS _i$-regions@rec

\knowledge{notion}
| $y$-avoiding region of $\attrDec {\G '}$
| $x-1$-avoiding region@recursive
| $y$-avoiding region@rec

\knowledge{notion}
  | $X$-Adam-closed@letterGame
  | $\SS _i$-Adam-closed@letterGame
  | $X$-Adam-closed subgraph@letterGame
  
\knowledge{notion}
  | $X$-closed subgame
  | $\SS _i$-closed subgame
  | $\SS _i$-closed subgames

\knowledge{notion}
  | $X$-FSCC@letterGame
  | $\SS _i$-FSCC@letterGame

\knowledge{notion}
  | tree of alternating subcycles
  | trees of alternating subcycles
  
\knowledge{notion}
  | DAG of alternating subcycles

\knowledge{notion}
  | local subDAG at $q$
  | local subDAG

\knowledge{notion}
  | round nodes@acd
  | round node@acd
  | round@acd

\knowledge{notion}
  | square nodes@acd
  | square node@acd
  | square@acd
  
 \knowledge{notion}
 | same shape
  
\knowledge{notion}
  | alternating cycle decomposition
  | Alternating Cycle Decomposition
  | ACD
  | Alternating cycle decomposition
  | ACDs
  | décomposition en cycles alternants

\knowledge{notion}
  | ACD-DAG
  | ACD-DAGs
  | DAG-version@ACD

\knowledge{notion}
  | set of nodes of $\acd {\TS }$

\knowledge{notion}
  | local subtree at $q$
  | local subtree at
  | local subtree
  | local subtrees

\knowledge{notion}
| local Muller condition of $\T $ at $q$
| local Muller condition of $\TS $ at $v$
| local Muller condition
| local Muller condition of $\TS '$ at $v'$
| local Muller conditions
| local Muller language of $\TS $ at $v$

\knowledge{notion}
  | positive@tree
  | positive tree

\knowledge{notion}
  | negative@tree

\knowledge{notion}
  | positive@acd

\knowledge{notion}
  | negative@acd

\knowledge{notion}
  | equidistant@acd

\knowledge{notion}
  | ACD-parity-transform
  | ACD-parity-transformation
  | ACD-parity@ACDtrans
  | ACD-transforms
  | ACD-transform
  | ACD transformations

\knowledge{notion}
  | ACD-HD-Rabin-transform
  | ACD-HD-Rabin-transformation
  | ACD-HD-Rabin@ACDtrans
  
\knowledge{notion}
  | ACD-HD-Rabin-transform-for-games
  
\knowledge{notion}
| A-successor
| A-successors

\knowledge{notion}
  | cycle-preimage-automaton at $v'$

\knowledge{notion}
  | unfolding

\knowledge{notion}
| equivalent acceptance conditions
| equivalent acceptance conditions over
| equivalent acceptance condition over
| equivalent over $\S $
| equivalent parity condition over
| equivalent over $G$
| equivalence of acceptance conditions
| equivalent to@accCond
| equivalent@accCond
| equivalent@acceptCond
| equivalent acceptance condition
| equivalent@acc
| equivalent over
| equivalent to@cond
| equivalent@cond
| equivalent to@acc
| equivalent parity condition over
| equivalent over $\autChromG $

\knowledge{notion}
| equivalent@TS

\knowledge{notion}
  | decide the typeness
  | decide its typeness
  | decidability of typeness
  | deciding the typeness
  | decide typeness
  
\knowledge{notion}
  | decide the typeness of a class of Muller automata in polynomial time
  | decidability in polynomial time of typeness

\knowledge{notion}
| parity language@isa

\knowledge{notion}
  | $X$ type
  | $\C$-type
  | Rabin type
  | Streett type
  | parity type
  | $[1,d+1]$-parity type
  | $[0,d-1]$ (resp. $[1,d]$)-parity type
  | Rabin@type
  | B\"uchi type
  | coB\"uchi type
  | B\"uchi@type
  | (resp. coB\"uchi) type
  | generalised B\"uchi@type
  | (resp. generalised coB\"uchi) type
  | generalised B\"uchi type
  | generalised coB\"uchi type
  | $[0,d-1]$ (resp. $[1,d]$ / $\WeakIndex {d}$)-parity type
  | typeness
  | $[0,d-1]$-parity type
  | $[1,d]$-parity type
  | types
  | type
  | Typeness
  | typeness problem
  
\knowledge{notion}
| Streett shape

\knowledge{notion}
| Streett shape

\knowledge{notion}
| Streett shape

\knowledge{notion}
 | $\Weak {d}$ type
 | $\Weak {d}$-type

\knowledge{notion}
  | Rabin shape
  | $X$ shape@ZT

\knowledge{notion}
  | Streett shape

\knowledge{notion}
  | Parity shape
  | parity shape
  
\knowledge{notion}
  | $[0,1]$-parity shape
  | $[1,2]$-parity shape

\knowledge{notion}
  | Rabin ACD
  | Rabin@ACD
  | $X$-ACD
  
\knowledge{notion}
  | Streett ACD

\knowledge{notion}
  | parity ACD
  | Parity ACD 

\knowledge{notion}
  | $[0,d-1]$-parity ACD
  | $[1,d]$-parity ACD
  | $[0,1]$-parity ACD
  | $[1,2]$-parity ACD
  | $[0,d]$-parity ACD
  | $[1,d+1]$-parity ACD
  | $[0,d-1]$(resp. $[1,d]$)-parity ACD
  | $[0,d-1]$ (resp. $[1,d]$)-parity ACD

\knowledge{notion}
  | shapes@ZDAG
  | shape of a Zielonka DAG

\knowledge{notion}
| $X$-ACD-DAG
| types@ACDDAG
| Rabin ACD-DAG
| Streett ACD-DAG
| parity ACD-DAG

\knowledge{notion}
| relabelled
| relabelling
| relabel

\knowledge{notion}
  | normalised
  | normal form
  | normalise
  | normalising
  | normalisation
  | forme normale
  | form@normal
  | normality
  | Normality
  
\knowledge{notion}
| normalised@neg
| normal form@neg
  
\knowledge{notion}
  | negative@TS
  | not negative@TS
  | non-negative@TS

\knowledge{notion}
  | duplicated edges
  | duplicated transitions

\knowledge{notion}
  | generalised Büchi language associated to

\knowledge{notion}
  | generalised B\"uchi language
  | generalised B\"uchi@language
  | generalised B\"uchi
  | generalised B\"uchi languages
  | generalised Büchi
  | generalised B\"uchi-
  | gen. B\"uchi
  | generalised (co)B\"uchi
  | Büchi généralisés
  | generalised (co)B\"uchi languages
  
\knowledge{notion}
  | generalised coB\"uchi language associated to

\knowledge{notion}
  | generalised coB\"uchi language
  | generalised coB\"uchi
  | coB\"uchi@generalised
  | gen. coB\"uchi

\knowledge{notion}
  | B\"uchi shape
 | B\"uchi@shapeZT
 
\knowledge{notion}
  | coB\"uchi shape
 | (resp. coB\"uchi) shape
 
\knowledge{notion}
  | Generalised B\"uchi shape
  | generalised B\"uchi shape
 | generalised B\"uchi@shapeZT
 
\knowledge{notion}
  | Generalised coB\"uchi shape
  | generalised coB\"uchi shape
  | (resp. generalised coB\"uchi) shape

\knowledge{notion}
  | B\"uchi ACD

\knowledge{notion}
  | coB\"uchi ACD

\knowledge{notion}
  | Generalised B\"uchi ACD
  | generalised B\"uchi ACD

\knowledge{notion}
  | Generalised coB\"uchi ACD
  | generalised coB\"uchi ACD

\knowledge{notion}
| $\Weak {d}$ ACD

\knowledge{notion}
 | generalised weak
 | Generalised weak
 | weak automata
 | generalised weak automata
 | generalised weak conditions
 | weak languages
 | weak

\knowledge{notion}
  | $k$-colour type@lang

\knowledge{notion}
  | $k$-Rabin-pair type@lang

\knowledge{notion}
  | $k$-colour type@TS
  | $|V|$-colour type@TS
  | $3$-colour type@TS
  | $3$-colour type@TS

\knowledge{notion}
  | $k$-multiple-colour type@TS
  | $3$-multiple-colour type@TS

\knowledge{notion}
  | $k$-Rabin-pair type@TS

\knowledge{notion}
  | $\F $-equivalent
  | $\F $-equivalence
  | $\sim_\F$ relation

\knowledge{notion}
  | double factorial

\knowledge{notion}
  | graph of rejecting subsets associated to $\F $
  | graphs of rejecting subsets associated to $\FHalfLetters {n}$
  | graph of rejecting subsets

\knowledge{notion}
  | neighbourhood
  | neighbour

\knowledge{notion}
  | open neighbourhood

\knowledge{notion}
| pseudo-path
| pseudo-paths

\knowledge{notion}
| stabilises around
| stabilise around
| stabilising around

\knowledge{notion}
  | $(\S ,m)$-cover
  
\knowledge{notion}
  | cover
  | vertex cover
  
\knowledge{notion}
  | colouring@chrm
  | colouring@graph
  | $k$-colouring@graph
  | $3$-colouring@graph
  | $3$-colourable
  | valid colouring
  
\knowledge{notion}  
  | size@colouring
  
\knowledge{notion}
  | chromatic number
  
\knowledge{notion}
	| undi

\knowledge{notion}
  | independent set
  | independent sets
 
\knowledge{notion} 
  | clique

\knowledge{notion}
| Horn clause
| Horn clauses

\knowledge{notion}
| negative@Horn

\knowledge{notion}
| Horn formula
| Horn formulas

\knowledge{notion} 
| Generalised Horn clause
| generalised Horn clause
| GH clause
| GH clauses

\knowledge{notion} 
| Generalised Horn formula
| generalised Horn formulas
| generalised Horn formula
| Generalised Horn formulas
| GH formula
| GH formulas

\knowledge{notion} 
| GH language

\knowledge{notion}
  | objective
  | objectives

\knowledge{notion}
| winning conditions
| condition de victoire
| winning condition
| winning objective

\knowledge{notion}
| game
| games
| Games
| infinite duration games
| jeux
| jeu

\knowledge{notion}
| $\locLang {W}{u}$-game
| $\locLangBuchi {W}{u}$-game
| $\locLang x q$-game
| $W$-game
| $W$-games
  | $\Muller {\F }$-game
  | $\Muller {\F }$-games
  | $W_n$-game
  | $\Lang {\A }$-game
  | $\addNeutral {W}$-games

\knowledge{notion}
| play
| plays
| finite play

\knowledge{notion}
  | Eve-game
  | Eve-games
  | Adam-games
  
\knowledge{notion}
  | bipositional over (finite) one-player games
  | bipositional over finite one-player games
  | bipositional over one-player games

\knowledge{notion}
  | output@game

\knowledge{notion}
  | winning@play

\knowledge{notion}
  | losing@play
  | losing
  | loses

\knowledge{notion}
  | subplay

\knowledge{notion}
| strategy
| strategies
| stratégies

\knowledge{notion}
| uniformly
| uniformly@strat
| uniform strategies

\knowledge{notion}
| wins
| wins $\G $ from $v$
| wins $\G \compositionAut \A $ from $(v, q_0)$
| win
| from@winGame
| winner
| wins $\G $ from
| winners
| wins from
| win from $v_0$
| winning from
| wins from $v_0$
| losing@strat
| winning from $v$
| stratégie gagnante
  | wins $\G \compositionAut \A $ from $(v, q_\init )$
  | win from

\knowledge{notion}
  | win positionally
  | positionally@win
  | positionally

\knowledge{notion}
| consistent with the strategy
| consistent with
| consistent with@strat
| consistency with@strat

\knowledge{notion}
| winning region
  | winning regions

\knowledge{notion}
| full winning region

\knowledge{notion}
| players
| player
| controls
| controlled by

\knowledge{notion}
| Eve
| Eve's
| Eve-vertex
| controlled by Eve
| owners of the vertices

\knowledge{notion}
| Adam
| Adam-vertex
| controlled by Adam
| Adam's vertex
| Adam's

\knowledge{notion}
| winning strategy for Eve
| winning strategy
| winning@strat
| winning for Eve
| winning strategy from
| winning
| winning strategies
| winning from@strat
| winning strategy from $v$

\knowledge{notion}
| solve@game
| solve a game
| Solving@game
| Solving a game
| solving@game
| solving@games
| solved@games
| résolution de ces jeux
| résolution des jeux
 | solving a game
 | solving games

\knowledge{notion}  
| determined

\knowledge{notion}
  | positional@strategy
  | positional strategy
  | positional@strat
  | positional strategies
  | positionnalité
  | positionally@strat
  | stratégies positionnelles
  | positionnelle

\knowledge{notion}
| optimal strategy
| optimal (for Eve)
| optimal@strat
| playing optimally
| play optimally

\knowledge{notion}
  | play optimally in $\G $ using a positional strategy
  | play optimally using positional strategies
  | jouer de manière optimale
  | play optimally using@pos
  | play optimally using a positional strategy
  | play optimally@pos
| play optimally using

\knowledge{notion}
  | half-positional@objective
  | half-positional
  | half-positionality
  | half-positionally
  | Half-positional
  | positionality@half
  | positionality
  | positional
  | positionnels
  
\knowledge{notion}  
  | half-positional over $\mathcal {X}$ games
  | half-positional over
  | half-positional over@eve

\knowledge{notion}
  | bipositional@objective
  | bipositionality@objective
  | bipositional
  | bipositionality
  | Bipositionality

\knowledge{notion}
  | memory skeleton over a set $X$
  | memory skeleton
  | memory skeleton $\M $ over
  | memory skeleton over
  | memory skeletons
  | memory skeleton over $\GG $
  | skeleton@mem

\knowledge{notion}
  | memory states
  | memory state

\knowledge{notion}
  | initial memory state

\knowledge{notion}
| update function@mem
| updates@mem
| update function
| memory updates

\knowledge{notion}
  | memory structure for a game $\G $
  | memory structure
  | finite memory structure
  | memory structures
  | memory@strat
  | memory
  | memories
  | mémoire
  | structures de mémoire
  | finite state automaton@memoryStruct
  | finite memory
  | finite memory strategy
  | finite memory@strat
  | memory@games
  | memory for games
  | structure@mem
  | memory@struct
  
\knowledge{notion}
  | $\ee $-memory

\knowledge{notion}
  | size@mem

\knowledge{notion}
  | play optimally in@mem
  | play optimally@mem
  | play optimally in $\G $ using memory $k$
  | play optimally with finite memory
  | play optimally in $\G $ using general memory $k$

\knowledge{notion}
| next-move function
| next-move

\knowledge{notion}
| implements@memStrat
| implemented by@mem
| implementing@mem
| implemented by
| implementing
| implements
| implements@strat
| implemented by@stratMem
| implements a strategy
| strategy implemented by
| implemented
| implements a@mem
| implemented by@strat
  | implement@mem
  | implement an
  | implementing@strat
  | implemented@mem

\knowledge{notion}
  | general memory structure
  | general memories
  | mémoires générales
  | general memory
  | general ones@mem
  | general memory structure for a game $\G $
  | memory structure@gen
  | General@mem
  | general
  | General
  | memory@gen

\knowledge{notion}  
  | memory required
  | Memory requirements
  | memory requirement
  | memory@req
  | memory requirements
  
\knowledge{notion}
  | general memory requirements
    | general ones@memReq
    | general ones@reqMem
    | general@memReq
    | memory requirements@gen

\knowledge{notion}
| general memory requirements over $\ee$-free games

\knowledge{notion}
| over vertex-coloured games@genreq

\knowledge{notion}
| over vertex-coloured games@chromreq

\knowledge{notion}
| over vertex-coloured games@arIndReq
 
\knowledge{notion}
  | chromatic memory structure
  | chromatic memory
  | chromatic memories
  | chromatic
  | mémoires chromatiques
  | Chromatic memories
  | Chromatic
  | chromatic memory structures
  | chromatic@mem
  | Chromatic@mem
  | chromatic memory@struc
  
\knowledge{notion}
  | chromatic memory requirements
  | chromatic memory requirements over
  | chromatic@reqMem

\knowledge{notion}
| chromatic memory requirements over $\ee$-free games

\knowledge{notion}
| arena-independent memory structure
| arena-independent memory
  | arena-independent memory structure for $W$
  | arena-independent memories
  | arena-independent memory
  | arena-independent memory for $W$
  | chromatic memory structures@arenaInd
  | Arena-Independent
  | arena-independent structure

\knowledge{notion}
| arena-independent memory requirements

\knowledge{notion}
| arena-independent memory requirements over $\ee$-free games
| over $\ee $-free games@arInd

\knowledge{notion}
| arena-independent memory structures over $\ee $-free games

\knowledge{notion}
  | $\ee $-free game
  | $\ee $-free
  | $\ee $-free@game
  | $\ee $-free games
  | $\ee $-free@games

\knowledge{notion}
| arbitrary games
| Arbitrary@game
| Arbitrary
| arbitrary@games
  
\knowledge{notion}
  | progress consistency
  | progress consistent

\knowledge{notion}
  | bi-progress consistent
  | bi-progress consistency

\knowledge{notion}
| fully progress consistent
| fully progress consistent@sigAut
| full progress consistency

\knowledge{notion}
  | $a$-transitions@out
	| $a$-transition@out
	| $a$-transitions@in
	| $a$-transition@in
	| $\NLetters {x}{q}$-transitions@in
	  | $a$-transition

\knowledge{notion}
  | $y$-transitions@out
  | $x$-transitions@out
  | $x-1$-transitions@out
  | $0$-transitions
  | $y$-transitions
  | ${>}x$-transitions
  | $x-1$-transitions
  | $x$-transitions
  | $({\geq }x)$-transitions
  | $({\geq }x-2)$-transitions
  | $({<}x)$-transitions
  | ${\geq }x$-transitions
  | $({>}x)$-transitions@out
  | $(x-1)$-transitions@out
  | $(x-1)$-transitions

\knowledge{notion}
  | deterministic over~$\DD '$
  | deterministic over
    | deterministic over $2$-transitions
    | deterministic over transitions
    | determinism over transitions
  | determinism over

\knowledge{notion}
  | accepted with priority $x$
  | accepted with priority
  | accepted with a priority
  | accepted with an even priority
  | accepted or rejected with priority
  | accepted with
  | accepted@byPriority

\knowledge{notion}
  | is rejected with priority $x$
  | is rejected with priority
  | rejected with an odd priority
    | rejected with priority
    | rejected with a priority

\knowledge{notion}
  | automaton with $\ee $-transitions
  | $\ee $-transitions@aut
  | $\ee $-letters

\knowledge{notion}
| signature automaton
| signature automata
| signature

\knowledge{notion}
 | $d$-signature automaton

\knowledge{notion}
  | $d$-structured signature automaton
  | $(x-2)$-structured signature automaton
  | $(x-2)$-structured signature automaton
  | $x$-structured signature automaton
  | $(x-2)$-structured
  | $x$-structured signature
  | structured signature
  | $(x-2)$-structured signature

\knowledge{notion}
  | structured signature automaton
  | structured
  | structured signature automata

\knowledge{notion}
  | reduced@sig
  | reduction@sig
  | reduced signature automata  
  | reduced signature automaton
  | reduced
  
\knowledge{notion}
 | strongly reduced@sig
 | refine this notion
 | refine this notion@strReducedSig

\knowledge{notion}
  | Local monotonicity
  | local monotonicity
  | local monotonicity of $({\geq } x)$-transitions
 | local monotonicity of $({\geq }x)$-transitions

\knowledge{notion}
| $({<}x)$-safe separation
  
\knowledge{notion}
  | $\SS $-graph
    | $\SS $-graphs
    | $\Sigma $-graphs

\knowledge{notion}
  | monotone@graph
  | monotone graph@univ
  | monotone@univ
  | monotonicity@graph
  | Monotonicity
  
\knowledge{notion}
    | $v$ satisfies@univ
    | satisfies@univ
    | satisfy@
    | satisfies@univTree
    | satisfies@tree
    | satisfies@graph
 | satisfy@univ
    
 \knowledge{notion}
  | totally ordered graph
  | ordered@univ
  
  \knowledge{notion}
    | well-ordered graph
    | well-ordered@graph
    | well-ordered@univ
    | well-order
    | well-ordered
  
  \knowledge{notion}
    | $\SS $-tree
    | $[0,\dots ,d]$-tree
    | trees@univ
    | $[0,d+1]$-tree
  
  \knowledge{notion}
    | root@univ
  
  \knowledge{notion}
    | satisfies@treeUniv
    | satisfy@treeUniv
  
  \knowledge{notion}
    | $(\kappa ,W)$-universal for trees
    | $(\kappa ,\parity )$-universal for trees
    | universality for trees
    | $(\kk ,W)$-universal for trees
    | $(\kappa ,W)$-universality for trees  
    | tree@univ

\knowledge{notion}
  | morphism@univ
  | morphism of $\SS $-graphs@univ

\knowledge{notion}
  | $(\kappa ,W)$-universal  
  | $(\kappa , W)$-universal
  | universal
  | universal graphs
  | universal graph
  | $(\kappa ,\parity )$-universal
  | universality
  | $(\kappa ,\parity )$-universality
  | $(\kk ,W)$-universal
  | Universal graphs
  | Universal graph
  | $(\aleph _0,W)$-universal

\knowledge{notion}
  | Universal graph for the parity objective
  | universal graph $\UPar $ for the parity objective

\knowledge{notion}
  | neutral letter
  | neutral for $W$
  | neutral letters
  
\knowledge{notion}
| induced order@eqRel

\knowledge{notion}
  | total@preorder
  | total preorder
  | total preorders
  | totally preordered
  | total@preord
  | total@order
  
\knowledge{notion}
  | refinement
  | refining
  | refines
  | refine

\knowledge{notion}
  | collection of nested (total) preorders over a set $Q$
  | collection of nested total preorders
  | nested total preorders
  | nested preorders
  | nested@preorder
  | collection of nested total preorders over a set $Q$

\knowledge{notion}
  | congruence for $\DD '$
  | congruences for
  | congruence for
  | congruency
  | congruence for $y$-transitions
  | congruence for $({>}x)$-transitions
  | preserve this congruence

\knowledge{notion}
   | congruence
   | congruences

\knowledge{notion}
  | strong congruence for
  | strong congruency
  | strong congruence
  | strong congruence for transitions

\knowledge{notion}
  | congruence class
  | equivalence class

\knowledge{notion}
  | are monotone for@cong
  | monotone for $\DD '$
  | monotone@cong
  | monotone for@cong
  | are monotone for
  | monotonicity@cong
  | monotone for $\leqSig {x}$
  | monotone for
  | monotone for a preorder

\knowledge{notion}
 | strictly monotone for
  
\knowledge{notion}
 |	(strictly) monotone@aut
 | monotone@aut
 | monotone automaton

\knowledge{notion}
  | induced equivalence relation
  | equivalence relation induced from a preorder 
  | induced equivalence relation@preord

\knowledge{notion}
  | residual of $L$ with respect to
  | residuals
  | residual
  | Residuals

\knowledge{notion}
  | residual classes
  | residual class
  | class@res
  | class
  
  \knowledge{notion}
| residual class@state
| class of states

\knowledge{notion}
  | equivalent@res

\knowledge{notion}
  | associated to@residual
  | associated to@res
 | states associated to@res
 | states of@res

\knowledge{notion}
    | act uniformly
    | uniformity
    | uniform
    | uniform over
    | uniformity for
    | $y_1$-uniformity
    | are uniform over $\sim$-classes
  | are uniform over $\sim $-classes
  | uniform@trans
  | uniformity@trans
  | are uniform
  | behave uniformly
  | uniformity@cong
 
 \knowledge{notion}
  | produces priority $y$ uniformly
  | produces priority $x-1$ uniformly
  | producing priority $x$ uniformly
  | produce priority $x$ uniformly
  | produces priority $x$ uniformly
  | producing priority $x$ uniformly in  
  | produces priority $x$ uniformly in

\knowledge{notion}
  | automaton of residuals
  | residual automaton

\knowledge{notion}
  | redirected by the transformation

\knowledge{notion}
 | preserves the residuals@transform

\knowledge{notion}
| leading to@class

\knowledge{notion}
| extended tuple

\knowledge{notion}
  | $[0,x]$-faithful
  | $[0,y]$-faithful
  | faithful
  | faithfulness
  | $[0,x-2]$-faithfulness
  | $[0,x]$-faithfulness
  | $[0,x-2]$-faithful
  | $[0,x-1]$-faithful
  | $[0,x-1]$-faithfulness
  | $[0,x]$-faithful congruence
  | faithful congruences
  | $[0,x-1]$-faithful congruence
  | $[0,x-2]$-faithful congruence
  | $[0,y]$@faithful
  | priority-faithful

\knowledge{notion}
  | quotient of $\A $ by $\sim $
  | quotient of $\A $ by
  | quotient automaton
    | quotient automaton@res
    | quotient@aut
    
\knowledge{notion}
  | $({\leq }x)$-quotient of $\A $ by $\sim $
  | ${\leq }x$-quotient
  | $({\leq }d)$-quotient by $\eqSig {d}$
  | quotient automata@leq
  | quotient of automata@leq
  | quotient automaton@leq
    | ${\leq }(x-1)$-quotient automata
    | quotient@leq
  | $({\leq }x)$-quotient
  
  \knowledge{notion}
    | projection of $\rr $@quot
    | projection@quot
    | induces the run@resAut

\knowledge{notion}
  | nice transformations
  | nice transformation
  | $\eqSig {x-2}$-nice transformation  
    | $\sim $-nice transformation of $\A $ at level $x$
    | $\eqSig {x-2}$-nice transformation at level $x-1$
    | $\eqRes {x-2}$-nice transformation of $\A $ at level $x-1$
    | $\eqRes {x-2}$-nice transformation at level $x-1$
    | $\eqSig {x-2}$-nice transformation of $\A $ at level $x-1$
    | $\eqSig {x}$-nice transformation of $\A $ at level $x$
    | nice transformation at level $x$
    | $\eqSig {x}$-nice transformation at level $x$

\knowledge{notion}
  | semantically deterministic

\knowledge{notion}
  | proper alphabet
  | proper alphabets

\knowledge{notion}
  | prefix codes
  | prefix code
  
\knowledge{notion}
  | local alphabet at $\resClass {u}$
  
\knowledge{notion}
  | $x$-local alphabet at $\classSig {x}{q}$

\knowledge{notion}
  | localisation of $W$ to~$\resClass {u}$
    | localisation to residuals

\knowledge{notion}
  | localisation of $W$ to~$\classSig {x}{q}$

\knowledge{notion}
  | local automaton of the residual $\resClassState {q}$
  | local automaton of the residual
  
\knowledge{notion}
  | local automaton of the class $\classSig {x}{q}$
  | local automaton of a $\eqSig {x}$-class
  | local automata
  | local automaton

\knowledge{notion}
  | $({\leq}1)$-safe languages
  | $({<}2)$-safe language@coB
  | $({<}2)$-safe languages@coB
  | safe language@coB
  | safe language@coB
  | safe languages@coB
  
\knowledge{notion}  
  | $(\leq 1)$-safe languages
  | safe languages
  | safe language
  | $({<}x)$-safe language
  | ${<}x$-safe languages
  | $({<}x)$-safe languages
  | safe languages@sig  
  | ${<}x$-safe language
  
\knowledge{notion}
  |	$({<}2)$-safe component@coB
  | safe component@coB
  | safe components@coB
  
\knowledge{notion}  
  | safe component
  | safe components
  | $({<}x)$-safe component
  | ${<}x$-safe components
  | ${<}y$-safe components
  | ${<}1$-safe components
  | ${<}(x-1)$-safe components
  | $({<}x)$-safe components
  | ${<}x$-safe component
  | $({<}x-1)$-safe component
  | $({<}x-1)$-safe components
  | $({<})y$-safe component
  | $({<}y)$-safe components
  | safe component@sig

\knowledge{notion}
 | safe path@coB
 | safe run@coB
 | safe@runCoB
 | $1$-safe run@coB
  
\knowledge{notion}  
 | safe path
| safe run
| safe@run
  | $({<}x)$-safe@path
  | ${<}x$-safe run
  | $({<}x)$-safe path
  | ${<}x$-safe@path
  | ${<}x$-safe path
  | ${<}x$-safe@run
  | safe run@sig
  | $(<x)$-safe path

\knowledge{notion}
| redundant@coB

\knowledge{notion}
 | redundant@sig
 | redundant@safec
 | redundant@safeCentr

\knowledge{notion}
  | closed objectives
  | topologically closed
  | closed@obj
  | closed
  | closed objective
  | closeness@obj
  | closed languages

\knowledge{notion}
  | open objectives
  | topologically open
  | open@obj
  | open
  | open objective

\knowledge{notion}
  | $\oo $-regular closed objectives
  | $\omega $-regular closed
  | $\oo $-regular closed
  | $\oo $-regular closed objective

\knowledge{notion}
  | safety objective associated to $L$

\knowledge{notion}
  | reachability objective associated to $L$

\knowledge{notion}
  | $\oo $-regular open objectives
  | $\oo $-regular open
  | $\omega $-regular open
  | open@reg

\knowledge{notion}
  | concave
  | Concavity

\knowledge{notion}
  | super word@buchi
  | super words@buchi
  
\knowledge{notion}
  | super word
  | super words
  | super word for

\knowledge{notion}
  | super letter@buchi
  | super letters@buchi
  | non-super letters@buchi

\knowledge{notion}
  | super letter
  | non-super letters
  | super letters for
  | super letters

\knowledge{notion}  
  | neutral letters@polishSig
  | neutral letter@sig
  | neutral letters@sig

\knowledge{notion}
  | polished@classBuchi
  | polish@classBuchi
  | polished in $\A $@buchi
  | polished class@buchi

\knowledge{notion}
  | polished@autBuchi
  | polished automata@buchi
  | polished automaton@buchi

\knowledge{notion}
  | redirecting@buchi
  | redirected@buchi
  | redirected transitions@buchi

\knowledge{notion}
  | polished@aut
  | polish@aut
  | $x$-polished@aut
  | $x$-polished automaton
  | $x$-polished automata
  | polishing operation

\knowledge{notion}
  | polished@class
  | polished in $\A $ 
  | polish@class
  | $x$-polished@class
  | polish
  | $x$-polished
  | $x$-polish@class
  | polish@sig
  | $x$-polished class
  | polished@classSig

\knowledge{notion}
| redirected in $\A '$@sc
 | redirected@sc
  | redirected transition@safec
  | redirected from@safec
  | redirected@safec

\knowledge{notion}
  | redirected transition@polish
  | redirected transitions@polish
  | redirected@polish

\knowledge{notion}  
| recurrent@class

\knowledge{notion}  
| transient@class
| transient classes

\knowledge{notion}
  | safe centralised@coB
  | safe centralisation@coB
    | Safe centrality@coB
  | safe centrality@coB
  
\knowledge{notion}
  | safe centralised
  | safe centralisation
  | Safe centrality
  | safe centrality
  | $({<}x)$-safe centralise
  | $({<}x)$-safe centralised
  | $({<}x)$-safe centrality
  | $({<}x)$-safe centralisation
  | ${<}x$-safe centralise
  | ${<}x$-safe centrality
  | safe centralisation@sig

\knowledge{notion}
  | homogeneous
  | homogeneity

\knowledge{notion} 
   | $({\leq }x)$-homogeneous 

\knowledge{notion}
  | safe minimal@coB
  | safe minimal
  | safe minimality

\knowledge{notion}
  | $1$-saturated

\knowledge{notion}
  | $(x-1)$-saturated
  | $(x-1)$-saturation
  | $x-1$-saturated
  | $x-1$-saturation

\knowledge{notion}
  | $1$-saturation@coB

\knowledge{notion}
  | $\ee$-complete automata
  | $\ee$-complete
  | $\ee$-complete automaton
  | $\ee $-completeness

\knowledge{notion}
  | $x$-jump
  | $x'$-jump
  | $(x-2)$-jump
  | $y$-jump
  | $(y-2)$-jump
  | $0$-jumps
  | $(x-1)$-jump
  | $(y-1)$-jump
  
\knowledge{notion}
  | priority-closed
  | priority-closed@aut
  | priority-closure  
  | closed by monotonicity
  | closed by monotonicity@HM
  
\knowledge{notion}
| makes a reset
| reset@make
  
\knowledge{notion}
  | reset-stable
  | reset-stability 
  | makes infinitely many resets
  | makes finitely many resets

\knowledge{notion}
| Kopczyński@conjUnion
| Kopczyński's conjecture
| this conjecture@KopczUnion

\knowledge{notion}
| Ohlmann@conj
| Ohlmann's conjecture
| Neutral letter conjecture
| Neutral Letter Conjecture
| neutral letter conjecture

\knowledge{notion}
| games originating from logical formulas
| originating from logical formulas

\knowledge{notion}
| notion
| definition

\knowledge{notion}
| hyperlink
| hyperliens
| Hyperlinks
| hyperlinks

\IfKnowledgePaperModeTF{
	}{
		\knowledgestyle{intro notion}{color={Dark Ruby Red}, emphasize}
		\knowledgestyle{notion}{color={Dark Blue Sapphire}}
		\hypersetup{
				colorlinks=true,
				breaklinks=true,
				linkcolor={}, 
				citecolor={}, 
				filecolor={Blue Marine}, 
				urlcolor={Blue Marine},
			}
	}
\input{colours.sty}

\usepackage{marginnote}

\begin{document}

\maketitle              
\begin{abstract}
In 2021, Casares, Colcombet and Fijalkow introduced the Alternating Cycle Decomposition (ACD), a structure used to define optimal transformations of Muller into parity automata and to obtain theoretical results about the possibility of relabelling automata with different acceptance conditions.
In this work, we study the complexity of computing the ACD and its DAG-version, proving that this can be done in polynomial time for suitable representations of the acceptance condition of the Muller automaton. As corollaries, we obtain that we can decide typeness of Muller automata in polynomial time, as well as the parity index of the languages they recognise.

Furthermore, we show that we can minimise in polynomial time the number of colours (resp. Rabin pairs) defining a Muller (resp. Rabin) acceptance condition, but that these problems become NP-complete when taking into account the structure of an automaton using such a condition.
\end{abstract}

\paragraph*{}
This document contains hyperlinks.
\AP Each occurrence of a "notion" is linked to its ""definition"".
On an electronic device, the reader can click on words or symbols (or just hover over them on some PDF readers) to see their definition.

\vskip1em

\newpage
\tableofcontents
\newpage

\section{Introduction}
\label{sec:introduction}
\subsection{Context}

\subparagraph*{Automata for the synthesis problem.} Since the 60s, automata over infinite words have provided a fundamental tool to study problems related to the decidability of different logics~\cite{Buchi1962decision,Rabin69S2S}.
Recent focus has centered on the study of synthesis of controllers for reactive systems with the specification given in Linear Temporal Logic ($\LTL$). 
The original automata-theoretic approach by Pnueli and Rosner~\cite{PR89Synthesis}  remains at the heart of the state-of-the-art $\LTL$-synthesis tools~\cite{EKRS17FromLTLtoParity,LMS20SynthesisLTL,MichaudColange18Synt,MullerSickert17LTLtoDeterministic}.
Their method consists in translating the $\LTL$ formula into a "deterministic" $\oo$-automaton which is then used to build an infinite duration game; a winning strategy in this game provides a correct controller for the system.

\subparagraph*{Different acceptance conditions.}
There are different ways of specifying which "runs" of an automaton over infinite words are "accepting@@run".
Generally, we label the transitions of the automaton with some "output colours", and we then indicate which colours should be seen (or not) infinitely often.
This can be expressed in a variety of ways, obtaining different "acceptance conditions", such as "parity", "Rabin" or "Muller".
The complexity of such "acceptance conditions" is crucial in the performance of algorithms dealing with automata and games over infinite words. For instance, "parity" games can be solved in quasi-polynomial time~\cite{CJKLS22} and parity games solvers are extremely performing in practice~\cite{SyntCompReport22}, while solving "Rabin" and "Muller" games is, respectively, $\NP$-complete~\cite{EJS93ModelChecking} and $\PSPACE$-complete~\cite{Dawar2005ComplexityBounds}.
Moreover, many existing algorithms for solving these games are polynomial on the size of the game graph, and the exponential dependency is only on parameters coming from the "acceptance condition":  "Muller" games can be solved in time $\O(k^{5k}n^5)$~\cite[Theorem~3.4]{CJKLS22},  where $n$ is the size of the game and $k$ is the number of "colours" used by the "acceptance condition", and "Rabin" games can be solved in time $\O(n^{r+3}rr!)$~\cite[Theorem~7]{PP06Faster}, where $r$ is the number of  "Rabin pairs" of the "acceptance condition".
Also, the emptiness check of "Muller" automata with the condition represented by a Boolean formula $\phi$ ("Emerson-Lei" condition) can be done in time $\O(2^kkn^2|\phi|)$~\cite[Theorem~1]{BBDKMS19EmptEL}.

Some important objectives are therefore: (1) transform an automaton $\A$ using a complex "acceptance conditions" into an automaton $\B$ using a simpler one, and (2) simplify as much as possible the "acceptance condition" used by an automaton $\A$ (without adding further states).

\subparagraph*{The Zielonka tree and Zielonka DAG.} The "Zielonka tree" is an informative representation of "Muller conditions", introduced for the study of strategy complexity in "Muller" games~\cite{Zielonka1998infinite,DJW1997memory}. 
Zielonka showed that we can use this structure to tell whether a "Muller language" can be expressed as a "Rabin" or a "parity language"~\cite[Section~5]{Zielonka1998infinite}.
Moreover, it has been recently proved that the "Zielonka tree" provides minimal "deterministic" "parity" "automata" "recognising" a "Muller condition"~\cite{CCFL24FromMtoP,MeyerSickert21OptimalPracticalNote}, and can thus be used to transform "Muller" automata using this condition into "equivalent@@aut" "parity" automata.

A natural alternative is to consider the more succinct "DAG"-version of this structure: the "Zielonka DAG". Hunter and Dawar studied the complexity of building the "Zielonka DAG" from an "explicit representation" of a "Muller condition", and the complexity of solving "Muller" games for these different representations~\cite{HD08ComplexityMuller}. Recently, Hugenroth showed that many decision problems concerning "Muller" automata become tractable when using the "Zielonka DAG" to represent the "acceptance condition"~\cite{Hugenroth23ZielonkaDAG}.

\subparagraph*{The ACD: Theoretical applications.} In 2021, Casares, Colcombet and Fijalkow~\cite{CCF21Optimal} proposed the "Alternating Cycle Decomposition" (ACD) as a generalisation of the "Zielonka tree".
The main motivation for the introduction of the "ACD" was to define optimal transformations of automata: given a "Muller" "automaton" $\A$, we can build using the "ACD" an equivalent parity automaton that is minimal amongst all "parity" automata that can be obtained by duplicating states of $\A$~\cite[Theorem~5.32]{CCFL24FromMtoP}.
Moreover, the "ACD" (or its "DAG-version@@ACD") can be used to tell whether a "Muller" automaton can be relabelled with an "acceptance condition" of a simpler "type"~\cite[Section~6.1]{CCFL24FromMtoP}.

However, the works introducing the "ACD"~\cite{CCF21Optimal,CCFL24FromMtoP} are of theoretical nature, and no study of the computational cost of constructing it and performing the related transformations is presented. 

\subparagraph*{The ACD: Practice.} The transformations based on the "ACD" have been implemented in the tools Spot~2.10~\cite{Spot2.10CAV22} and Owl~21.0~\cite{KMS18Owl}, and are used in the $\LTL$-synthesis tools \texttt{ltlsynt}~\cite{MichaudColange18Synt} and \texttt{STRIX}~\cite{LMS20SynthesisLTL, MS21ModernisingStrix} (top-ranked in the \textsc{SYNTCOMP} competitions~\cite{SyntCompReport22}).
In the tool paper~\cite{CDMRS22Tacas}, these transformation are compared with the state-of-the-art methods to transform "Emerson-Lei" automata into "parity" ones.
Surprisingly, the transformation based on the "ACD" does not only produce the smallest "parity" automata, but also outperforms all other existing paritizing methods in computation time.

In~\cite[Section~4]{CDMRS22Tacas}, an algorithm computing the "ACD" is proposed. However, the focus is made in the handling of Boolean formulas to enhance the algorithm's performance in practice, but no theoretical analysis of its complexity is provided.

\subparagraph*{Simplification of acceptance conditions.} 

As already mentioned, the complexity of the "acceptance conditions" play a crucial role in algorithms. 
One can simplify the acceptance condition of a "Muller automaton" by adding further states (and the optimal way of doing this is determined by the "ACD"~\cite{CCFL24FromMtoP}).
However, in some cases this leads to an exponential blow-up in the number of states~\cite{Loding1999Optimal}.
A natural question is therefore to try to simplify the acceptance condition while avoiding adding so many states. 
We consider two versions of this problem:

\begin{description}
	\item[Typeness problem.] Can we "relabel" the "acceptance condition" of $\A$ with one of a simpler "type", such as "Rabin", "Streett" or "parity"?
	\item[Minimisation of colours and Rabin pairs.] Can we minimise the number of "colours" used by the "acceptance condition" (or, in the case of "Rabin automata", the number of "Rabin pairs")?
\end{description}

The "ACD" has proven fruitful for studying the "typeness problem": just by inspecting the "ACD" of $\A$, we can tell whether we can "relabel" it with an "equivalent@@accCond" "Rabin", "parity" or "Streett" "acceptance condition"~\cite{CCFL24FromMtoP}.
Also, it is a classical result that we can minimise in polynomial time the number of "colours" used by a "parity" automaton~\cite{CartonMaceiras99RabinIndex}.
However, it was still unclear whether the "ACD" could be of any help for minimising the number of "colours" of "Muller conditions" or the number of "Rabin pairs" of "Rabin" "acceptance conditions", question that we tackle in this work. 

The minimisation of colours in "Muller" automata has recently been studied by Schwarzová, Strejček and Major~\cite{SSJ23MarksEL}. In their approach, they use heuristics to reduce the number of colours by applying QBF-solvers. The final "acceptance condition" is however not guaranteed to have a minimal number of colours.
There have also been attempts to minimise the number of "Rabin pairs" of "Rabin" automata coming from the determinisation of B\"uchi automata~\cite{TD14BuchiTight}.
Also, in their work about minimal "history-deterministic" "Rabin" "automata", Casares, Colcombet and Lehtinen left open the question of the minimisation of "Rabin pairs"~\cite{CCL22SizeGFG}.

\subsection{Contributions}
We outline the main contributions of this work.
\begin{description}
	\item[1. Computation of the ACD and the ACD-DAG.] Our main contribution is to show that we can compute the "ACD" of a "Muller automaton" in polynomial time, provided that we are given the "Zielonka tree" of its "acceptance condition" as input (Theorem~\ref{th-comp:compt-ACD-poly-ZT}). This shows that the computation of the "ACD" is not harder than that of the "Zielonka tree", (partially) explaining the strikingly favourable experimental results from~\cite{CDMRS22Tacas}.
	We also show that we can compute the "DAG-version@@ACD" of the "ACD" in polynomial time if the "acceptance condition" of $\A$ is given "colour-explicitly@@Muller" or by a "Zielonka DAG" (Theorem~\ref{th-comp:compt-ACD-DAG-poly-ZDAG}).
	The main technical challenge is to prove that the "ACD" (resp. "ACD-DAG") has polynomial size in the size of the "Zielonka tree" (resp. "Zielonka DAG").
	
	\item[2. Deciding typeness and the parity index in polynomial time.] Combining the previous contributions with the results from~\cite{CCFL24FromMtoP}, we directly obtain that we can decide in polynomial time whether a "Muller" automaton can be "relabelled" with an "equivalent@@accCond" "parity", "Rabin" or "Streett" "acceptance condition" (Corollary~\ref{cor-comp:decision-typ-poly}).
	Moreover, we recover a result from Wilke and Yoo~\cite{WilkeYoo96RabinIndex}: we can compute in polynomial time the "parity index" of the language of a "Muller automaton".

	\item[3. Minimisation of colours and Rabin pairs of acceptance conditions.] Given a "Muller" (resp. "Rabin") language $L$, we show that we can minimise the number of "colours" (resp. "Rabin pairs") needed to define $L$ in polynomial time (Theorems~\ref{th-min:colorMinML-poly} and~\ref{th-min:RabinPairsMinML-poly}).

	\item[4. Minimisation of colours and Rabin pairs over an automaton structure.] Given an automaton $\A$ using a "Muller" (resp. "Rabin") "acceptance condition", we show that the problem of minimising the number of "colours" (resp. "Rabin pairs") to "relabel" $\A$ with an "equivalent acceptance condition" over its structure is $\NP$-complete, even if the "ACD" is given as input (Theorems~\ref{th-min:colorMin-Aut-NPhard} and~\ref{th-min:RabinPairsMin-Aut-NPHard}).
	This result is obtained for both automata using single and multiple colours per edge.
	This came as a surprise to us, as our first intuition was in fact that the "ACD" would allow to lift the previous polynomial-time minimisation results to the problem in which we take into account the structure of the automaton.

	\item[5. Analysis on the size of different representations of Muller conditions.] 
	We provide tight bounds on the size of the "Zielonka tree" in the worst case (Proposition~\ref{prop-size:worst-case-ZT}). Combining them with~\cite[Theorem~4.13]{CCFL24FromMtoP}, we recover results from L\"oding~\cite{Loding1999Optimal} giving bounds on the size of "deterministic" "parity" automata, and extend them to "history-deterministic" automata.
	We moreover provide examples showing the exponential gap on the size of the different representations of "Muller conditions" (Section~\ref{subsec:size-comparison}).	
\end{description}

Furthermore, we include an appendix (Appendix~\ref{sec:Gen-Horn}) in which we study a subclass of interest of Boolean formulas, which we call "generalised Horn formulas", and relate them to the problem of minimising the number of "Rabin pairs" of a "Rabin language".

\section{Preliminaries}
\label{sec:preliminaries}

\paragraph*{Basic notations}

\AP For a set $A$ we let $|A|$ denote its cardinality, $\pow{A}$ its power set and $\intro*\powplus{A} = \pow{A} \setminus \{\emptyset\}$.
\AP For a family of subsets $\F\subseteq \pow{A}$ and $A'\subseteq A$, we write $\intro*\restSubsets{\F}{A'} = \F \cap \pow{A'}$.
For natural numbers $i\leq j$, $[i,j]$ stands for $\{i,i+1, \dots, j-1,j \}$.

\AP For a set $\Sigma$, a word over $\Sigma$ is a sequence of elements from $\Sigma$.
\AP The sets of finite and infinite words over $\Sigma$ will be written $\Sigma^*$ and $\Sigma^{\oo}$, respectively, 
\AP and we let $\intro*\SigmaInfty = \Sigma^* \cup \Sigma^\oo$. 
\AP Subsets of $\SS^*$ and $\SS^\oo$ will be called languages.  For a word $w\in \SigmaInfty$ we write $w_i$ to represent the i\ts{th} letter of $w$.
\AP The concatenation of two words $u\in \Sigma^*$ and $v\in \SigmaInfty$ is written $u\cdot v$, or simply $uv$. 
\AP For a word $w\in \Sigma^\oo$, we let $\intro*\minf(w)=\{ a\in \Sigma \mid w_i=a \text{ for infinitely many } i\in \NN\}$.

\subsection{Automata over infinite words and their acceptance conditions}
\paragraph*{Automata}

\AP A ""(non-deterministic) automaton"" is a tuple $\A = (Q,q_\init,\Sigma, \DD, \GG, \colAut, W)$, where $Q$ is a finite set of states, $q_\init\in Q$ is an initial state, 
\AP $\SS$ is an ""input alphabet"", $\DD\subseteq Q\times \SS \times Q$ is a set of transitions, 
\AP $\GG$ is a finite set of ""output colours"", 
\AP $\colAut\colon \DD \to \GG$ is a colouring of the transitions,
\AP and $W\subseteq\GG^\oo$ is a language over $\GG$.
\AP We call the tuple $(\colAut,W)$ the ""acceptance condition"" of $\A$. 
We write $q\re{a}q'$ to denote a transition $e=(q,a,q')\in\DD$, and $q\re{a:c}q'$ to further indicate that $\colAut(e)=c$. We write $q\lrp{w:u}q'$ to represent the existence of a path from $q$ to $q'$ labelled with the "input letters" $w\in \SS^*$ and "output colours@@aut" $u\in\GG^*$.

\AP We say that $\A$ is ""deterministic"" (resp. ""complete"") if for every $q\in Q$ and $a\in \Sigma$, there is at most (resp. at least) one transition of the form $q\re{a}q'$.

\AP Given an "automaton" $\A$ and a word $w\in \SS^\oo$, a  ""run over $w$"" in $\A$ is a path 
\[q_\init \re{w_0:c_0} q_1 \re{w_1:c_1} q_2 \re{w_2:c_2} q_3 \re{w_3:c_3} \dots \in \DD^\oo.\]
Such a  "run" is ""accepting@@runAut"" if $c_0c_1c_2\dots\in W$, and ""rejecting@@runAut"" otherwise.
\AP A word $w\in \SS^\oo$ is  ""accepted@@word"" by $\A$ if it admits an "accepting@@runAut" "run".
\AP The ""language recognised"" by an automaton $\A$ is the set 
\[ \intro*\Lang{\A}= \{ w\in \Sigma^\oo \mid w \text{ is "accepted@@word" by } \A \}.\]
\AP We say that two "automata" over the same input alphabet are ""equivalent@@aut"" if they "recognise" the same "language".

\AP We let the \emph{size of $\A$} be $\intro*\sizeAut{\A} = |Q| + |\SS| + |\DD| + |\GG|$. We note that this does not take into account the size of the representation of its "acceptance condition", which can admit different forms (see page~\pageref{par:representation}). When necessary, we will indicate the size of the representation of the "acceptance condition" separately.

\begin{remark}
	Most results in this paper concern the set of "accepting runs" of "automata", rather than their "languages". For instance we will try to modify the "acceptance condition" while preserving the set of "accepting runs". 
	However in the case of "deterministic" "complete" "automata", those two notions coincide: preserving the set of "accepting runs" is exactly the same as preserving the "language".
	Hence all results pertaining to the "languages recognised" by automata appearing in this paper will concern "deterministic automata".
\end{remark}

\begin{remark}[Transition-based acceptance]
	We remark that the colours used to define the acceptance of runs appear \emph{over transitions}, instead of over states. This makes an important difference for many decisions problems on automata over infinite words such as the ones considered in this paper.
	For a discussion on the differences between transition-based and state-based automata, and arguments on why the first should be preferred, we refer to~\cite[Chapter~VI]{Casares23Thesis}.
\end{remark}

\begin{remark}[Multiple colours and transitions]
	We note that in the definition above, each transition is labelled with a single colour in $\GG$. It is sometimes useful to let transitions carry multiple colours -- for instance, this is the standard model in the HOA format~\cite{HOAFormat2015}.
	For many results of this paper (those from Section~\ref{sec:computation-ACD}), allowing or not multiple colours per edge does not make a difference; we could always take $\GG' = \pow{\GG}$ or $\GG'=\DD$ as new set of colours.
	However, multiple labels become significant for the problem of the minimisation of colours over a "Muller automaton", studied in Section~\ref{subsec:min-colours-automata}. We refer to that section for more details.
	
	Also, the HOA format allows for multiple transitions between the same two states with the same input letter. These transitions can always be replaced by one carrying multiple colours (we refer to~\cite[Prop. 18]{CCL22SizeGFG} for details).
\end{remark}

Some corollaries of our results will refer to "history-deterministic" automata, although this model will not play a central role in our work.
\AP An automaton $\A$ is ""history-deterministic"" if there is a function $\ss\colon \SS^+ \to \DD$ resolving its non-determinism in such a way that for every $w\in \Lang{\A}$, the "run" built by this function is an "accepting run".

\paragraph*{Acceptance conditions}

We now define the main classes of languages used by "automata" over infinite words as "acceptance conditions".
We let $\Gamma$ stand for a finite set of "colours".

\begin{description}\setlength\itemsep{2mm}
	\item[Muller.] \AP We define the ""Muller language"" of a family $\F\subseteq \powplus{\Gamma}$ of non-empty subsets of $\Gamma$ as:
	\[ \intro*\MullerC{\F}{\GG} = \{w\in \GG^\oo \mid \minf(w) \in \F\}. \]
	\AP We will often refer to sets in $\F$ as ""accepting sets@@Muller"" and sets not in $\F$ as ""rejecting sets@@Muller"".

	\item[Rabin.] \AP A Rabin condition is represented by a family $\R=\{(\greenPair_1,\redPair_1),\dots,(\greenPair_r,\redPair_r)\}$  of ""Rabin pairs"", where $\greenPair_j,\redPair_j\subseteq \Gamma$.
	\AP We define the ""Rabin language"" of a single Rabin pair $(\greenPair,\redPair)$ as 
	\[ \intro*\RabinC{(\greenPair,\redPair)}{\GG} = \{w\in \GG^\oo \mid \minf(w)\cap \greenPair \neq \emptyset \land \minf(w)\cap \redPair = \emptyset\}, \]
	and the Rabin language of a family of "Rabin pairs" $\R$ as:
	$ \reintro*\RabinC{\R}{\GG} = \bigcup_{j=1}^r \RabinC{(\greenPair_j,\redPair_j)}{\GG}$.
	
	\item[Streett.] \AP The ""Streett language"" of a family  $\R=\{(\greenPair_1,\redPair_1),\dots,(\greenPair_r,\redPair_r)\}$ of "Rabin pairs" is defined as the complement of its "Rabin language":
	\[ \intro*\StreettC{\R}{\GG} = \GG^\oo \setminus \RabinC{\R}{\GG}. \]

	\item[Parity.] \AP We define the ""parity language"" over a finite alphabet $\Prio \subseteq \NN$ 
	as:
	\[ \intro*\parity_{\Prio} = \{ w\in \Prio^\oo \mid \min \minf(w) \text{ is even}\}.\]
\end{description}

\AP We say that an "automaton" is a ""$\C$ automaton"", for $\C$ one of the classes of languages above, if its "acceptance condition" uses a $\C$ language. 
We refer to the survey~\cite{Boker18WhyTypes} for a more detailed account on different types of "acceptance conditions".

\begin{remark}\label{rmk-p0-prelim:characterisation_Muller_Languages}
	"Muller languages" are exactly the languages characterised by the set of letters seen infinitely often. 
	They are also the "languages" "recognised" by "deterministic" "Muller" "automata" with one state.
\end{remark}

We observe that "parity languages" are special cases of "Rabin" and "Streett languages" which are in turn special cases of "Muller languages".

\begin{example}\label{ex-p0-prelim:multiple-automata-examples}
	In Figure~\ref{fig-prelim:multiple-aut-ex} we show different types of "automata" over the alphabet $\SS=\{a,b\}$ "recognising" the language of words that contain infinitely many $b$s and  eventually do not encounter the factor $abb$.
\end{example}
\begin{figure}[ht]
	\centering 
	\includegraphics[]{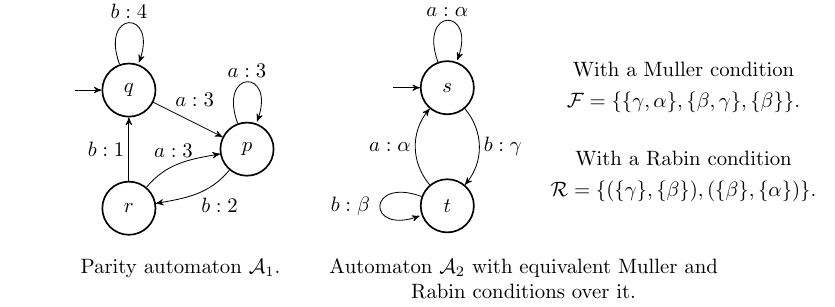}
	\caption{Different types of "automata" "recognising" the language $L = \SS^*b^\oo + \SS^* (a^+b)^\oo$ $=$\\ $\{ w \in \SS^\oo \mid w \text{ has infinitely many } b\text{s} \text{ and finitely many } abb \text{ factors} \}$. (Note that $\{\beta,\gamma\}$ cannot result as the set of outputs occurring infinitely often in a run of $\A_2$.)}
	\label{fig-prelim:multiple-aut-ex}
\end{figure}

The $8$ classes of "automata" obtained by combining the $4$ types of "acceptance conditions" above with "deterministic" and "non-deterministic" models are equally expressive~\cite{McNaughton1966Testing,Mostowski1984RegularEF}. 
\AP We call the class of languages that can be "recognised" by these automata ""$\oo$-regular languages"".

\paragraph*{Typeness}
Let $\A_1= (Q,q_\init,\Sigma, \DD, \GG_1, \colAut_1, W_1)$ be a "deterministic automaton", and let $\C$ be a class of languages (potentially containing languages over different alphabets).
\AP We say that $\A_1$ can be ""relabelled"" with a $\C$-"acceptance condition", or that $\A_1$ is ""$\C$-type"", if there is $W_2\subseteq \GG_2^\oo$, $W_2\in \C$, and a "colouring function" $\colAut_2\colon \DD\to \GG_2$ such that $\A_1$ is "equivalent@@aut" to $\A_2= (Q,q_\init,\Sigma, \DD, \GG_2, \colAut_2, W_2)$. 
In this case, we say that $(\col_1,W_1)$ and $(\col_2,W_2)$ are ""equivalent acceptance conditions over"" $\A_1$.

\begin{remark}
	In this work, we only consider typeness for "deterministic" automata. For "non-deterministic" models, typeness admits two non-equivalent definitions~\cite{KMM06Typeness}: (1) the acceptance status of each individual infinite path coincide for both "acceptance conditions", or (2) both automata "recognise" the same language.
\end{remark}

\begin{example}
	The "automaton" $\A_2$ from Figure~\ref{fig-prelim:multiple-aut-ex} is "Rabin type", as we have labelled it with a "Rabin" "acceptance condition" that is "equivalent over" $\A$ to the "Muller condition" given by $\F$ (in this case, both conditions use the same set of "colours" $\GG=\{\aa,\bb,\gg\}$).
	However, we note that $\RabinC{\R}{\GG}\neq \MullerC{\F}{\GG}$, as $\gg^\oo\in \RabinC{\R}{\GG}$, while $\gg^\oo\notin \MullerC{\F}{\GG}$.
	This is possible, as no infinite path in $\A_2$ is labelled by a word that differentiates both languages (such as $\gamma^\omega$).
\end{example}

\AP Given a "Muller automaton" $\A$, we use the expression ""deciding the typeness"" of $\A$ for the problem of answering if:
\begin{itemize}
	\item $\A$ is "Rabin type",
	\item $\A$ is "Streett type", and
	\item $\A$ is "parity type".
\end{itemize}

Formally, these are three different decision problems. 
\AP We say that we can ""decide the typeness of a class of Muller automata in polynomial time"" if the three of them can be decided in polynomial time.\footnote{Here, we could consider further classes of "acceptance conditions" such as B\"uchi, coB\"uchi, generalised B\"uchi, weak, etc... We refer to~\cite[Appendix~A]{CCFL24FromMtoP} for more details on these acceptance types. Our main result establishing "decidability in polynomial time of typeness" for "Muller automata" also holds for these acceptance conditions, as they are characterised by the "ACD-DAG".}

\paragraph*{Parity index}

Let $L\subseteq \SS^\oo$ be an "$\oo$-regular language". 
\AP The ""parity index"" of $L$ is the minimal number $k$ such that $L$ can be "recognised" by a "deterministic" "parity" "automaton" using $k$ "output colours".\footnote{This notion can be refined by taking into account whether the minimal colour needed is odd or even. We omit these details here for the sake of simplicit	y of the presentation, and refer to~\cite[Definition~2.14]{CCFL24FromMtoP} for formal definitions.}
Such number is well-defined, as any "Muller automaton" admits an equivalent "deterministic" "parity" "automaton"~\cite{Mostowski1984RegularEF}. 
Moreover, it does not depend on the particular "parity" "automaton" used to "recognise" $L$:

\begin{proposition}[\cite{NiwinskiWalukievicz1998Relating}]
	Let $\A$ be a "deterministic" "parity" "automaton" "recognising" a language $L\subseteq \SS^\oo$. If $L$ has "parity index" $k$, then $\A$ admits an "equivalent parity condition over" it using only $k$ "output colours".
\end{proposition}

	As a matter of fact, the "parity index" of a language coincides with the minimal number of colours used by a "Muller automaton" recognising it~\cite[Proposition 6.14]{CCFL24FromMtoP}. However, in contrast with the previous proposition, in order to reduce the number of colours of a "Muller automaton" we may need to modify its structure.

\subsection{The Zielonka tree and the Zielonka DAG}\label{subsec:ZielonkaTree}

We now introduce two closely-related ways of representing "Muller conditions", the "Zielonka tree" and the "Zielonka DAG", which are obtained by recursively listing the maximal "accepting@@Muller" and "rejecting@@Muller" subsets of colours of a family $\F\subseteq \powplus{\GG}$.

\paragraph*{Trees and DAGs}

\AP We represent a ""tree"" as a pair $T=(N,\ancestor)$ with $N$ a non-empty finite set of nodes and $\intro*\ancestor$ the ""ancestor relation"" ($n\ancestor n'$ meaning that $n$ is above $n'$).
We assume the reader to be familiar with the usual vocabulary associated with trees. 
The set of leaves of $T$ is written $\intro*\leaves(T)$.
\AP An ""$A$-labelled tree"" is a "tree" $T$ together with a labelling function $\nu \colon N \to A$.

\AP A ""directed acyclic graph"" (DAG)~$(D,\ancestorDAG)$ is a non-empty finite set of nodes~$D$ equipped with an order relation~$\intro*\ancestorDAG$ called the \emph{ancestor relation} such that there is a minimal node for~$\ancestorDAG$, called the \emph{root}.
We apply to DAGs similar vocabulary than for "trees" (children, leaves, depth, subDAG rooted at a node, ...).
\AP An ""$A$-labelled DAG"" is a "DAG" together with a labelling function $\nu \colon D \to A$.

\paragraph*{The Zielonka tree}

\begin{definition}[{\cite{Zielonka1998infinite}}]\label{def-zt:zielonkaTree}
	\AP Let~$\F\subseteq\powplus{\GG}$ be a family of non-empty subsets of a finite set~$\GG$. The ""Zielonka tree"" for~$\F$ (over $\GG$),\footnotemark{} denoted $\intro*\zielonkaTree{\F} = (N, \ancestor, \intro*\nu :N \to \powplus{\GG})$ is a $\powplus{\GG}$-labelled "tree" with nodes partitioned into ""round nodes"" and ""square nodes"", $N= \intro*\roundnodes \disjUnion \intro*\squarenodes$, such that:
	\footnotetext{The definition of $\zielonkaTree{\F}$, as well as most subsequent definitions, do not only depend on $\F$ but also on the alphabet $\GG$. Although this dependence is important,  we do not explicitly include it in the notations in order to lighten them, as most of the times the alphabet will be clear from the context.}
	\begin{itemize}
		\item The root is labelled $\GG$.
		\item If a node is labelled~$X\subseteq \GG$, with~$X\in\F$, then it is a "round node", and it has a child for each maximal non-empty subset~$Y\subseteq X$ such that~$Y\not\in\F$, which is labelled $Y$.
		\item If a node is labelled~$X\subseteq \GG$, with~$X\not\in\F$, then it is a "square node", and it has a child for each maximal non-empty subset~$Y\subseteq X$ such that~$Y\in\F$, which is labelled $Y$. 
	\end{itemize}

\end{definition}

We write $|\zielonkaTree{\F}|$ to denote the number of nodes in $\zielonkaTree{\F}$.


\begin{remark}\label{rmk-zt:union-changes-acceptance}
	Let $n$ be a node of $\zielonkaTree{\F}$ and let $n_1$ be a child of it. If $\nu(n_1) \subsetneq X \subseteq \nu(n)$, then $\nu(n_1)\in \F \iff X\notin \F \iff \nu(n)\notin \F$. In particular, if $n_1, n_2$ are two different children of $n$, then  
	$\nu(n_1)\in \F \iff \nu(n_2)\in \F \iff \nu(n_1)\cup \nu(n_2)\notin \F$.
\end{remark}

The next lemma provides a simple way to decide if a subset $C\subseteq \GG$ belongs to $\F$ given the "Zielonka tree". It follows directly from the previous remark.

\begin{lemma}\label{lemma-zt:accepting-set-in-ZT}
	Let $C\subseteq \GG$ and let $n$ be a node of $\zielonkaTree{\F}$ such that $C\subseteq \nu(n)$ and that is maximal for $\ancestor$ amongst nodes containing $C$ in its label. Then, $C\in \F$ if and only if $n$ is "round".
\end{lemma}


\begin{example}\label{example-ZT:zielonka-tree}
	Let $\F$ be the "Muller condition" used by the automaton from Example~\ref{ex-p0-prelim:multiple-automata-examples}:
	$\F=\{ \{\gg,\aa\}, \{\gg,\bb\}, \{\bb\}\}$, over the alphabet $\set{\aa, \bb, \gg}$.
	In Figure~\ref{fig-zt:zielonkaTree-example} we show the "Zielonka tree" of $\F$. 
\end{example}
\begin{figure}[ht]
	\centering 
	\includegraphics[scale=1]{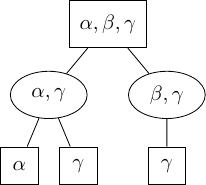}
	\caption{"Zielonka tree" $\zielonkaTree{\F}$ for 
		$\F=\{ \{\gg,\aa\}, \{\gg,\bb\}, \{\bb\}\}$.}
	\label{fig-zt:zielonkaTree-example}
\end{figure}

One important application of the "Zielonka tree" is that it provides minimal "parity" automata "recognising" "Muller languages".

\begin{proposition}[{\cite[Theorem~4.13]{CCFL24FromMtoP}}]\label{prop-prelim:minimal-det-parity}
	Let $L=\MullerC{\F}{\SS}$ be a "Muller language". There is a "deterministic" "parity" automaton of size $|\leaves(\zielonkaTree{\F})|$ "recognising" $L$. Moreover, such an automaton is minimal both amongst "deterministic" and "history-deterministic" "parity" automata "recognising" $L$.
\end{proposition}

The "Zielonka tree" can also be used to obtain minimal "history-deterministic" "Rabin" automata. The formal statement of this result can be found in~\cite[Proposition~11]{CCL22SizeGFG}.

We can use "Zielonka trees" to represent the "Muller" "acceptance condition" of an automaton. Not all labelled trees correspond to "Zielonka trees" arising from "Muller languages", nonetheless, this does not pose a problem since we can efficiently check whether a tree  corresponds to some "Zielonka tree", and in this case, it defines a unique "Muller language".

\begin{lemma}[{Implied by~\cite[Lemma~1]{Hugenroth23ZielonkaDAG}}]\label{lemma-checkZielonkaTree}
	Given a $\powplus{\Gamma}$-labelled tree $T$ with a partition into round and square nodes, we can decide in polynomial time whether there is a family $\F \subseteq \powplus{\Gamma}$ such that $T = \zielonkaTree{\F}$.
	In the affirmative case, this family is uniquely determined.
\end{lemma}
\begin{proof}
	First, it is a necessary condition that children of round nodes of $T$ are square nodes and vice-versa.
	We say that a subset $C\subseteq \Gamma$ is accepted by a round node $n_r$ of $T$ if $C\subseteq \nu(n_r) \tand C\nsubseteq \nu(n_r') \text{ for any children } n_r' \text{ of } n_r$.
	We define symmetrically to be rejected by a square node of $T$.
	Let 
	\begin{align*}
		\F_+  = & \{ C\subseteq \Gamma \mid C \text{ is accepted by some round node of } T\};\\
		\F_-  = & \{ D\subseteq \Gamma \mid D \text{ is rejected by some square node of } T\}.
	\end{align*}
	
	If $T$ is the "Zielonka tree" of a family of subsets, this family must be $\F_+$. 
	If it is not the case that $T = \zielonkaTree{F_+}$, it must be because there is $C\in \F_+\cap \F_-$. Let $n_r$ and $n_s$ be a round and a square node accepting and rejecting $C$, respectively. Let $C' = \nu(n_r)\cap \nu(n_s)$. This set also has the property that is accepted according to $n_r$, and rejecting according to $n_s$.
	Therefore, $T = \zielonkaTree{\F}$ if and only if for all pairs of a square node and a round node the intersection of their labels is not both in $\F_+$ and in $\F_-$. This can be done in polynomial time.	
\end{proof}

\paragraph*{The Zielonka DAG}

\AP The ""Zielonka DAG"" of a family $\F\subseteq \powplus{\GG}$ is the "labelled@@DAG" "directed acyclic graph" obtained by merging the nodes of $\zielonkaTree{\F} = (N,\ancestor, \nu)$ that share a common label.
Formally, it is the "labelled DAG" $\intro*\zielonkaDAG{\F} = (N',\ancestorDAG', \nu')$ where $N' = \{C\subseteq \GG \mid \exists n \text{ node of the "Zielonka tree" such that }\ab C=\nu(n)\}$, $\nu'$ is the identity and the relation $\ancestorDAG'$ is inherited from the "ancestor relation" of the "tree": $C\ancestorDAG' D$ if there are $n_C,n_D$ nodes of the "Zielonka tree" such that $\nu(n_C)=C$, $\nu(n_D)=D$ and $n_C\ancestor n_D$.
In particular, $C\ancestorDAG' D$ implies $D\subseteq C$ (but the converse does not hold in general).

We remark that $\zielonkaDAG{\F}$ inherits the partition of nodes into "round@@nodes" and "square@@nodes" ones. 
Moreover, children of a "round node" of the "Zielonka DAG" are "square nodes" and vice-versa. 
We also note that Remark~\ref{rmk-zt:union-changes-acceptance} and Lemma~\ref{lemma-zt:accepting-set-in-ZT} hold similarly replacing $\zielonkaTree{\F}$ by $\zielonkaDAG{\F}$ in their statement.

\begin{example}
	The "Zielonka DAG" of the condition $\F=\{ \{\gg,\aa\}, \{\gg,\bb\}, \{\bb\}\}$ is obtained by merging the nodes labelled $\{\gg\}$ in the tree from Figure~\ref{fig-zt:zielonkaTree-example}.
	Another example can be found in Figure~\ref{fig-size:ZT-vs-ZDAG} (page~\pageref{fig-size:ZT-vs-ZDAG}).
\end{example}

\paragraph*{Representation of acceptance conditions}\label{par:representation}
\AP There is a wide variety of ways to ""represent@@Muller"" a "Muller language" $W = \MullerC{\F}{\GG}$, and the complexity and practicality of algorithms manipulating "Muller" "automata" may greatly differ depending on which of these representations is used~\cite{Horn2008Explicit, Dawar2005ComplexityBounds}.
\AP A ""colour-explicit representation@@Muller"" is given simply as a list of the subsets appearing in $\F\subseteq \powplus{\GG}$.
In this section we have defined two further representations for "Muller languages": the "Zielonka tree" and the "Zielonka DAG". A thorough study of automata with "acceptance condition" given as a "Zielonka DAG" was conducted by Hugenroth~\cite{Hugenroth23ZielonkaDAG}. Our results (of orthogonal nature) reinforce his thesis that the "Zielonka DAG" is a well-suited way of representing "Muller" "acceptance conditions" providing a good balance between succinctness and algorithmic properties.

In Figure~\ref{fig-prel:representationMuller} we provide a summary of the relationship between these three "representations". These will be proved and studied in further detail in Section~\ref{sec:size}.
We highlight that the "Zielonka DAG"  can be built in polynomial time from both the "Zielonka tree" and from a "colour-explicit@@Muller" representation of a "Muller condition"~\cite[Theorem~3.17]{HD08ComplexityMuller}, being the most succinct representation of the three.
The exponential-size separation between the "Zielonka tree" and "colour-explicit" representations, as well as explicit examples showing the gap between "Zielonka trees" and "DAGs" are original contributions. We provide full proofs and further discussions in Section~\ref{sec:size}.

\begin{figure}[ht]
	\centering
	\includegraphics[]{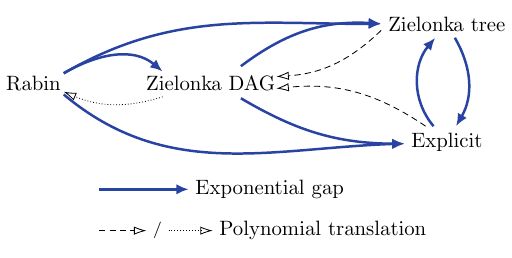}
	\caption{Comparison between the different "representations@@Muller" of "Muller" "conditions@@acc". A blue bold arrow from $X$ to $Y$ means that converting an $X$-representation into the form $Y$ cannot be done in polynomial time. A dashed arrow from $X$ to $Y$ means that a conversion can be computed in polynomial time. The dotted arrow indicates that the polynomial translation can only be applied on a fragment of $X$, as it is more expressive than $Y$.}
	\label{fig-prel:representationMuller}
\end{figure}

\AP In practical applications it is sometimes useful to have a more succinct representation, and a common choice are ""Emerson-Lei conditions"", which describe a family $\F$ as a positive Boolean formula over the primitives $\infEL(c)$ and $\finEL(c)$, for $c\in \GG$.
We do not consider Emerson-Lei representations in this work, as they inherit the complexity analysis of Boolean formulas; for instance, the problem of emptiness of Emerson-Lei automata (even with a single state) is essentially the SAT problem, and thus \NP-complete.

\subsection{The Alternating Cycle Decomposition}\label{subsec:ACD}
We now present the "Alternating Cycle Decomposition" and its DAG-version, following~\cite{CCFL24FromMtoP}. We also justify their interest by listing some key properties, mainly, optimal transformations of automata (Proposition~\ref{prop-prel:optimal-transform-ACD}) and characterisations of the "typeness" and the "parity index" of automata (Propositions~\ref{prop-prel:typeness-ACD} and~\ref{prop-prel:parity-index-ACD}).
We start with some definitions about "cycles" of automata.

\paragraph*{Cycles}
Let $\A$ be an "automaton" with $Q$ and $\DD$ as set of states and transitions, respectively. A ""cycle"" of $\A$ is a subset $\ell \subseteq \DD$ such that there is a finite path $q_0 \re {a_0}q_1 \re {a_1}q_2 \re{a_2} \dots q_r \re {a_r} q_0$ with $e_i=(q_i,a_i,q_{i+1})\in\DD$ and $\ell = \{e_0, e_1, \dots, e_r\}$.
Note that we do not require this path to be simple, that is, edges and vertices may appear multiple times. 
The set of states of the cycle $\ell$ is $\intro*\states{\ell} = \{q_0, q_1, \dots q_r\}$.
\AP The set of "cycles" of an "automaton" $\A$ is written $\intro*\cycles{\A}$. 
We will consider the set of "cycles" ordered by inclusion.
\AP For a state $q\in Q$, we note $\intro*\cyclesState{\A}{q}$ the subset of "cycles" of $Q$ containing $q$.
Note that $\cyclesState{\A}{q}$ is closed under union; moreover, the union of two "cycles" $\ell_1, \ell_2\in \cycles{\A}$ is again a "cycle" if and only if they have some state in common.
\AP A state is called ""recurrent"" if it belongs to some "cycle" and ""transient"" if it does not.
\AP If we see $\A$ as a graph, its "cycles" are the strongly connected subgraphs of that graph, and the maximal "cycles" are its strongly connected components (SCCs).

Let $\A$ be a "Muller" "automaton" with "acceptance condition" $(\colAut, \MullerC{\F}{\GG})$.
\AP Given a "cycle" $\ell\in \cycles{\A}$, we say that $\ell$ is ""accepting@@cycle"" (resp. ""rejecting@@cycle"") if $\colAut(\ell) \in \F$ (resp.~$\colAut(\ell) \notin \F$). 


\paragraph*{Tree of alternating subcycles and the Alternating Cycle Decomposition}

\begin{definition}\label{def:tree_alternating_cycles}
	Let $\ell_0 \in \cycles{\A}$ be a "cycle". 
	\AP We define the ""tree of alternating subcycles"" of~$\ell_0$, denoted $\intro*\altTree{\ell_0} = (N, \ancestor, \intro*\nuAcd\colon N \to \cycles{\A})$ as a  $\cycles{\A}$-"labelled tree" with nodes partitioned into ""round nodes@@acd"" and ""square nodes@@acd"", $N= \roundnodes \sqcup \squarenodes$, such that:
	\begin{itemize}
		\item The root is labelled $\ell_0$.
		\item If a node is labelled $\ell\in \cycles{\A}$, and $\ell$ is an "accepting cycle" ($\colAut(\ell)\in \F$), then it is a "round node@@acd", and its children are labelled exactly with the maximal subcycles $\ell' \subseteq \ell$ such that $\ell'$ is "rejecting@@cycle" ($\colAut(\ell') \notin \F$).
		\item If a node is labelled $\ell\in \cycles{\A}$, and $\ell$ is a "rejecting cycle" ($\colAut(\ell)\notin \F$), then it is a "square node@@acd", and its children are labelled exactly with the maximal subcycles $\ell' \subseteq \ell$ such that $\ell'$ is "accepting@@cycle" ($\colAut(\ell') \in \F$).
	\end{itemize}
\end{definition}

%


\begin{definition}[Alternating cycle decomposition]\label{def:acd}
	\AP Let $\A$ be a "Muller" "automaton", and let $\ell_1, \ell_2, \dots, \ell_k$ be an enumeration of its maximal "cycles". We define the ""alternating cycle decomposition"" of $\A$ to be  the forest $\intro*\acd{\A} = \{\altTree{\ell_1},\dots, \altTree{\ell_k}\}$.
\end{definition}


\begin{remark}\label{rmk-acd:ZT-as-ACD}
	The "Zielonka tree" can be seen as the special case of the "alternating cycle decomposition" for "automata" with a single state. Indeed, a "Muller language" $\MullerC{\F}{\SS}$ can be trivially "recognised" by a "deterministic" "Muller" "automaton" $\A$ with a single state $q$ and self loops $q\re{a:a}q$. The "ACD" of this "automaton" is exactly the "Zielonka tree" of $\F$.
\end{remark}

\begin{example}{}\label{ex-acd:example-acd}
	We show the "alternating cycle decomposition" of the "automata" $\A_1$ and $\A_2$ from Figure~\ref{fig-prelim:multiple-aut-ex} in Figure~\ref{fig-acd:acd}. 
	As these "automata" are deterministic, we can represent their transitions as pairs $(q,a)\in Q\times\SS$.
	Since both of them are strongly connected, each "ACD" consists in a single "tree", whose root is the whole set of transitions.	
	The bold coloured subtrees correspond to "local subtrees" at states $p$ and $t$, respectively, as defined below.
\end{example}
\begin{figure}[ht]
	\centering 
	\includegraphics[]{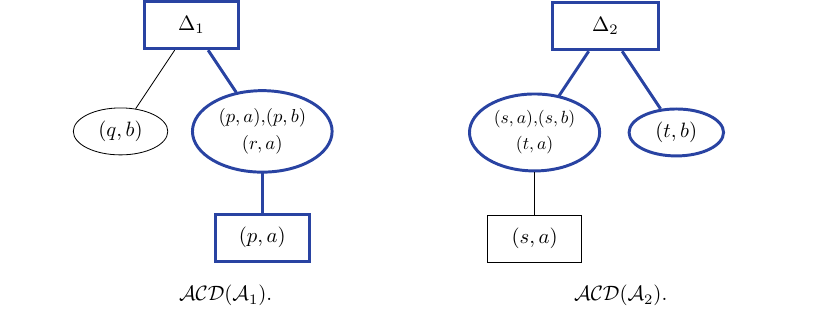}
	\caption{"Alternating cycle decomposition" of $\A_1$ and $\A_2$, from Figure~\ref{fig-prelim:multiple-aut-ex}. In bold blue, the "local subtrees" of $\acd{\A_1}$ at $p$ and of $\acd{\A_2}$ at $t$.}
	\label{fig-acd:acd}
\end{figure}

\paragraph*{Local subtrees}
We remark that for a "recurrent" state $q$ of $\A$, there is one and only one "tree" $\altTree{\ell_i}$ in $\acd{\A}$ such that $q$ appears in such a "tree". On the other hand, "transient" vertices do not appear in the trees of $\acd{\A}$.

\AP If $q$ is a "recurrent" state of $\A$, appearing in the SCC $\ell_i$, we define the ""local subtree at $q$"", noted $\intro*\treeVertex{q}$, as the subtree of $\altTree{\ell_i}$ containing the nodes 
$ \intro*\nodesTreeVertex{q} = \{ n \in \altTree{\ell_i} \mid q \text{ is a state in } \nuAcd(n)\}$.
(If $q$ is a "transient" state, $\treeVertex{q}$ is empty.)



\paragraph*{The ACD-parity-transform}
As mentioned in the introduction, the "ACD" was introduced as a structure to build small "parity" automata from  "Muller" ones.
\AP Casares, Colcombet and Fijalkow~\cite{CCF21Optimal} defined the ""ACD-parity-transform""  $\intro*\acdParityTransform{\A}$ of a "Muller automaton" $\A$, which is an "equivalent@@aut" "parity" automaton with the property that it is minimal amongst "parity" automata that can be obtained from $\A$ by duplication of states.
They formalise this minimality statement using "morphisms" of automata.\footnote{For simplicity, here we define only "morphisms" of "deterministic" automata. More general statements use the notions of locally bijective and history-deterministic morphisms. We refer to~\cite{CCFL24FromMtoP} for details.}

Let $\A$ and $\B$ be two "deterministic" automata over the same input alphabet.
\AP  A ""morphism"" $\varphi\colon \A \to \B$ is a function sending states of $\A$ to states of $\B$ such that:
\begin{itemize}
	\item $\varphi(q_\init^\A)$ is the initial state of $\B$,
	\item  for all transition $(q,a,q')$ in $\A$, $(\varphi(q),a,\varphi(q'))$ is a transition in $\B$, and
	\item a "run" $\rr$ in $\A$ is "accepting@@run" if and only if $\varphi(\rr)$ is "accepting@@run" in $\B$.
\end{itemize} 

\begin{proposition}[{\cite[Section~5.2]{CCFL24FromMtoP}}]\label{prop-prel:optimal-transform-ACD}
	Given a "Muller automaton" $\A$ and its "ACD", we can build in polynomial time a "parity" automaton $\acdParityTransform{\A}$ "equivalent@@aut" to $\A$ such that no "deterministic" "parity" automaton admitting a "morphism" to $\A$ is smaller than $\acdParityTransform{\A}$.
\end{proposition}

This optimality result has been generalised to "history-deterministic" "parity" automata, and optimal transformations towards "history-deterministic" "Rabin" automata based on the "ACD" have been also proposed~\cite{CCFL24FromMtoP}.

\paragraph*{The ACD-DAG}
\AP In the same way as we obtained the "Zielonka DAG" from the "Zielonka tree", we define a "DAG" obtained by merging the nodes of the "ACD" sharing the same label. 

Let $\A$ be a "Muller" "automaton".
\AP The ""DAG of alternating subcycles"" of a "cycle"~$\ell$, denoted $\intro*\altDAG{\ell}$ is the  $\cycles{\A}$-"labelled DAG" obtained by merging the nodes of $\altTree{\ell}$ with a same label.
\AP The ""ACD-DAG"" of a "Muller" "automaton" $\A$  is $\intro*\acdDAG{\A} = \{\altDAG{\ell_1},\dots, \ab\altDAG{\ell_k}\}$, where $\ell_1,\dots,\ell_k$ is an enumeration of the maximal "cycles" of $\A$ (that is, of its SCCs).

\AP For $q$ a state of $\A$, we define the ""local subDAG at $q$"", noted $\intro*\dagVertex{q}$, as the "DAG" obtained by merging the nodes of $\treeVertex{q}$ with a same label.
We note that if $\ell_i$ is the maximal "cycle" containing $q$, $\dagVertex{q}$ coincides with the subDAG of $\altDAG{\ell_i}$ consisting of the nodes labelled with cycles containing $q$.

The "ACD-DAG" of a "deterministic" "Muller" "automaton" $\A$ can be used to "decide its typeness"  and the "parity index" of $\Lang{\A}$.

\begin{proposition}[{\cite[Section~6.1]{CCFL24FromMtoP}}]\label{prop-prel:typeness-ACD}
	Given a "deterministic" "Muller" "automaton" $\A$ and its "ACD-DAG", we can "decide the typeness" of $\A$ in polynomial time. More precisely, $\A$ is:
	\begin{itemize}
		\item "Rabin type" if and if for all $q\in Q$ and "round node" $n\in \dagVertex{q}$, $n$ has at most one child in $\dagVertex{q}$;
		\item "Streett type" if and if for all $q\in Q$ and "square node" $n\in \dagVertex{q}$, $n$ has at most one child in $\dagVertex{q}$;
		\item "parity type" if and if for all $q\in Q$, $\dagVertex{q}$ has a single branch.
	\end{itemize} 
\end{proposition}

\begin{proposition}[{\cite[Proposition~6.13]{CCFL24FromMtoP}}]\label{prop-prel:parity-index-ACD}
	Let $\A$ be a "deterministic" "Muller" "automaton". The "parity index" of $\Lang{\A}$ coincides with the maximal height of a "DAG" from $\acdDAG{\A}$ (which coincides with the maximal height of a "tree" from $\acd{\A}$).
\end{proposition}



\section{Computation of the Alternating Cycle Decomposition}
\label{sec:computation-ACD}

We present in this section the main contribution of the paper: a polynomial-time algorithm to compute the "alternating cycle decomposition" of a "Muller" "automaton", and its analysis. 
We prove that if the "acceptance condition" of the automaton is "represented@@Muller" as a "Zielonka tree", we can compute $\acd{\A}$ in polynomial time (Theorem~\ref{th-comp:compt-ACD-poly-ZT}). 
This shows that the computation of the "ACD" is not harder than that of the "Zielonka tree", (partially) explaining the strikingly performing experimental results from~\cite{CDMRS22Tacas}.
We also show that if the "acceptance condition" is "represented@@Muller" as a "Zielonka DAG", we can compute $\acdDAG{\A}$ in polynomial time (Theorem~\ref{th-comp:compt-ACD-DAG-poly-ZDAG}), from which we can derive "decidability in polynomial time of typeness" of "Muller" "automata" (Corollary~\ref{cor-comp:decision-typ-poly}) and of the "parity index" of the languages they "recognise" (Corollary~\ref{cor-comp:decision-parity-index-poly}).

\subsection{Statements of the results}

We first state the results that will be obtained in this section. In all the section, given an "automaton" $\A$, $Q$ will stand for its set of states.

\begin{theorem}[Computation of the ACD]\label{th-comp:compt-ACD-poly-ZT}
	Given a "Muller" "automaton" $\A$ with "acceptance condition" represented by a "Zielonka tree" $\zielonkaTree{\F}$, we can compute $\acd{\A}$ in polynomial time in $\sizeAut{\A}+|\zielonkaTree{\F}|$.
\end{theorem}

As explained in Proposition~\ref{prop-prel:optimal-transform-ACD}, given the "ACD" of a "Muller automaton" $\A$, we can transform $\A$ in polynomial time into its "ACD-parity-transform": a "parity automaton" equivalent to $\A$ that is minimal amongst parity automata obtained as a transformation of $\A$. The previous theorem implies that this can be done even if only the "Zielonka tree" of the "acceptance condition" of $\A$ is given as input.

\begin{corollary}\label{cor-comp:acdParityTransformPoly}
	We can compute the "ACD-parity-transform" of a "Muller automaton" in polynomial time, if its "acceptance condition" is given by a "Zielonka tree".
\end{corollary}

\begin{theorem}[Computation of the ACD-DAG]\label{th-comp:compt-ACD-DAG-poly-ZDAG}
	Given a "Muller" "automaton" $\A$ with acceptance condition represented  by a "Zielonka DAG" $\zielonkaDAG{\F}$ (resp. "colour-explicitly@@Muller"), we can compute $\acdDAG{\A}$ in polynomial time in $\sizeAut{\A}+|\zielonkaDAG{\F}|$ (resp.~$\sizeAut{\A}+|\F|$).
\end{theorem}

Combining Theorem~\ref{th-comp:compt-ACD-DAG-poly-ZDAG} with Propositions~\ref{prop-prel:typeness-ACD} and~\ref{prop-prel:parity-index-ACD}, we directly obtain that we can "decide typeness" of "Muller automata" and the "parity index" of their languages in polynomial time.

\begin{corollary}[Polynomial-time decidability of typeness]\label{cor-comp:decision-typ-poly}
	Given a "deterministic" "Muller automaton" $\A$ with its "acceptance condition" represented "colour-explicitly@@Muller", as a "Zielonka tree", or as a "Zielonka DAG", we can "decide the typeness" of $\A$ in polynomial time.
\end{corollary}

\begin{corollary}[Polynomial-time decidability of parity index]\label{cor-comp:decision-parity-index-poly}
	Given a "deterministic" "Muller" "automaton" $\A$ with its "acceptance condition" represented "colour-explicitly@@Muller", as a "Zielonka tree", or as a "Zielonka DAG", we can determine the "parity index" of $\Lang{\A}$ in polynomial time.
\end{corollary}

The decidability of the "parity index" in polynomial time had already been obtained by Wilke and Yoo~\cite{WilkeYoo96RabinIndex}.
This result contrasts with the fact that deciding the "parity index" of a language represented by a "deterministic" "Rabin" or "Streett" "automaton" is $\NP$-complete~\cite[Theorem~28]{KPB95Structural}. It was already well-known that the "parity index" was computable in polynomial time from a "deterministic" "parity" automata~\cite{NiwinskiWalukievicz1998Relating,CartonMaceiras99RabinIndex}.

\subsection{Main algorithm}

We present the pseudocode of an algorithm computing $\acdDAG{\A}$ (Algorithm~\ref{algo:computationACD}) from a "Muller automaton" $\A$.
The full procedures requires a time polynomial in $|Q| + |\zielonkaDAG{\F}| + |\acdDAG{\A}|$; we will then obtain Theorem~\ref{th-comp:compt-ACD-DAG-poly-ZDAG} by showing that $|\acdDAG{\A}| \leq |Q|\cdot |\zielonkaDAG{\F}|$ (if the "acceptance condition" is represented "colour-explicitly@@Muller", we can compute the "Zielonka DAG" from it in polynomial time~\cite{HD08ComplexityMuller}).
If we want to compute the "ACD" of an "automaton" $\A$ with the "acceptance condition" given as a "Zielonka tree", we can simply compute the "Zielonka DAG" from it, apply the previous procedure to get the "ACD-DAG" and then unfold the latter to obtain the "ACD".
As a result, we can compute the "ACD" in time polynomial in $|Q| + |\zielonkaTree{\F}| + |\acd{\A}|$; we will then obtain Theorem~\ref{th-comp:compt-ACD-poly-ZT}  by showing that $|\acd{\A}| \leq |Q|\cdot |\zielonkaTree{\F}|$.
Quite surprisingly, the arguments we need to use to prove these upper bounds for $|\acdDAG{\A}|$ and $|\acd{\A}|$ are quite different.

The algorithm we propose builds the "ACD-DAG" in a top-down fashion: first, it computes the strongly connected components of $\A$ and initialises the root of each of the DAGs in $\acdDAG{\A}$. 
Then, it iteratively computes the children of the already found nodes using the sub-procedure $\computeChildrenACD$, presented in~Algorithm~\ref{algo:ChildrenNodeACD}. 
Given a node $n$ labelled with $\nuAcd(n)=\ell$ (assume that $\ell$ is an "accepting" "cycle"), $\computeChildrenACD$ goes through all "round" nodes in the "Zielonka DAG" and for each such node $m$ computes the maximal sub-cycles of $\ell$ whose set of colours is included in the one of $m$, but not in the one of any child of $m$.
The algorithm then selects maximal "cycles" among all those, add them to $\acdDAG{\A}$ (if they do not already appear in the DAG)  and sets them as children of $n$.

We use the following notations:
\begin{itemize}
	\item \AP $\intro*\SCCDec(\S)$ outputs a list of the strongly connected components of~$\S$.
	\item \AP $\intro*\pop(\mathsf{stck})$ removes an element from the stack $\mathsf{stck}$ and returns it. 
	\item \AP $\intro*\push(\mathsf{stck}, \mathsf{L})$ adds the elements of $\mathsf{L}$ to the stack $\mathsf{stck}$.
	\item \AP $\intro*\maxInclusion(\mathsf{lst})$ returns the list of the maximal subsets in $\mathsf{lst}$.
\end{itemize}
All the previous functions can be computed in polynomial time. 


\begin{algorithm}[!htbp]
	\caption{Computation of the "ACD-DAG"}
	\label{algo:computationACD}
	\begin{algorithmic}[1]
		\Statex \textbf{Input:} A "Muller" "automaton" $\A$
		\Statex \textbf{Output:} $\acdDAG{\A}$
		\State $\langle \S_{1}, \dots, \S_{r}\rangle \leftarrow \SCCDec(\A)$
		
		\State Add $\S_{1}, \dots, \S_{r}$ as the root of $r$ different DAGs  of $\acdDAG{\A}$
		\State $\nodesStack \leftarrow \langle \S_{1}, \dots, \S_{r}\rangle$ \Comment{Initialise a stack}
		\While{$\nodesStack\neq \emptyset$}
		\State $\nodeACD \leftarrow \pop(\nodesStack)$
		\State $\childrenList \leftarrow\computeChildrenACD(\nodeACD)$
		\State $\newChildren \leftarrow$ elements of $\childrenList$ that do not appear in $\acdDAG{\A}$\label{line-Cacd:newChildren}
		\State Add the nodes of $\newChildren$ to $\acdDAG{\A}$\label{line-Cacd:addChildrenNodes}
		\State Add an edge from $\nodeACD$ to each element of $\childrenList$ \label{line-Cacd:addChildrenEdges}
		\State $\push(\nodesStack,\newChildren)$
		\EndWhile
		\State \Return $\acdDAG{\A}$
	\end{algorithmic}
\end{algorithm}

We provide in Algorithm~\ref{algo:ChildrenNodeACD} a procedure to compute the children of a node of the "ACD-DAG".
We show in Lemma~\ref{lemma-comp:analysis-comptChildrenZT} that this algorithm is correct and terminates in polynomial time in $\size{\zielonkaDAG{\F}} + |Q|$.
We say that two nodes have the ""same shape"" if they are both "round" or both "square".

\AP
\begin{algorithm}[!htbp]
	\caption{ $\intro*\computeChildrenACD(n)$: Computing the children of a node $n$ of the "ACD-DAG"}
	\label{algo:ChildrenNodeACD}
	\begin{algorithmic}[1]
		\Statex \textbf{Input:} A node of $\acdDAG{\A}$ labelled by a "cycle" $C$
		\Statex \textbf{Output:} Maximal "subcycles" $\ell_1,\dots,\ell_k$ of $C$ such that $\colAut(\ell_i)\in \F \iff \colAut(C)\notin \F$.
		\State $\childrenList$ $\gets \emptyset$  
		
		\For{$m \in \zielonkaDAG{\F}$ a node of the "same shape" as $n$} \label{line-CC:for-altSets} 
		
		\State $C_m \leftarrow$ restriction of $C$ to transitions $e$ such that $\colAut(e)\in C$ 
		\State $\langle C_{m,1}, \dots, C_{m,r}\rangle \leftarrow \SCCDec(C_m)$\label{line-CC:SCC} 
		\For{$i = 1,\dots, r$}\label{line-CC:for-SCCs} 
		\If{ for all child $p$ of $m$, $\colAut(C_{m,i})\nsubseteq \colAut(\nuAcd(p))$}
		\State $\childrenList \leftarrow \childrenList \cup \{C_{m,i}\}$\label{line-CC:addChildren} 
		\EndIf			
		\EndFor
		\EndFor\label{line-CC:endFor}
		\State $\childrenList \leftarrow \maxInclusion(\childrenList)$ \label{line-CC:maxIncl} 
		\State \Return $\childrenList$ \label{line-CC:return} 
	\end{algorithmic}
\end{algorithm}

\subsection{Complexity analysis}
We now prove correctness and termination in polynomial time of the algorithms presented in the previous subsection, establishing Theorems~\ref{th-comp:compt-ACD-poly-ZT} and~\ref{th-comp:compt-ACD-DAG-poly-ZDAG}.

We first remark that Algorithm~\ref{algo:computationACD} makes at most $|\acdDAG{\A}|$ calls to the function $\computeChildrenACD$, as each node of the "ACD-DAG" is added at most once to $\nodesStack$.
Therefore, to obtain Theorem~\ref{th-comp:compt-ACD-DAG-poly-ZDAG} (computation of the "ACD-DAG") we need to show: (1) $|\acdDAG{\A}|$ is polynomial in $|Q|+|\zielonkaDAG{\F}|$, and (2) the function $\computeChildrenACD$ takes polynomial time in this measure.




We start by showing that we can compute the children of a node of $\acdDAG{\A}$ in polynomial time in $|Q|+|\zielonkaDAG{\F}|$. The obtention of the upper bounds on $\acdDAG{\A}$ will be the subject of the next subsection.

\begin{lemma}\label{lemma-comp:analysis-comptChildrenZT}
	Algorithm~\ref{algo:ChildrenNodeACD} computes the list of children of a node of $\acdDAG{\A}$ in polynomial time in $\size{\zielonkaDAG{\F}} + |Q|$.
\end{lemma}
\begin{proof}
	First let us argue that the returned list contains exactly the children of the input node in $\acdDAG{\A}$.
	
	Let $n$ be the input node, $C$ its label, and let us assume that it is "square", the other case is symmetric.
	Its children are the maximal "cycles" $\ell_1,\dots,\ell_k\in \cycles{C}$ such that $\colAut(\ell_i)\in \F$. By definition of $\zielonkaDAG{\F}$, those are the maximal "cycles" such that there exists a "round" node $m$ in $\zielonkaDAG{\F}$ such that $\colAut(\ell_i) \subseteq C$ and $\colAut(\ell_i) \nsubseteq p$ for all children $p$ of $C$.
	This is straightforwardly what Algorithm~\ref{algo:ChildrenNodeACD} computes, as the algorithm goes through all "round" nodes $C$, computes the maximal "cycles" whose set of colours are included in $C$ but not in its children in $\zielonkaDAG{\F}$ and adds them to $\childrenList$. It then outputs the maximal "cycles" in $\childrenList$.
	
	For the complexity, note that we go through the \for{} loop on line~\ref{line-CC:for-altSets} at most $\size{\zielonkaDAG{\F}}$ times, and through the \for{} loop of line~\ref{line-CC:for-SCCs} at most $|Q|$ times at each iteration. Computing $\SCCDec(\S_C)$ on line~\ref{line-CC:SCC} requires time linear in $|Q|$ by Tarjan's algorithm~\cite{Tarjan72DepthFirst}. 
	As a result, the execution time of Algorithm~\ref{algo:ChildrenNodeACD} up to line~\ref{line-CC:endFor}  and the size of $\childrenList$ after line~\ref{line-CC:endFor} are both polynomial in $|Q|+\size{\zielonkaDAG{\A}}$, hence the whole algorithm takes polynomial time in that measure.
\end{proof}


\subsection{Upper bounds on the size of the ACD and the ACD-DAG }
We now establish the desired upper bounds on the size of the "ACD" and the "ACD-DAG". We start by proving that  $|\acd{\A}|\leq |Q|\cdot|\zielonkaTree{\F}|$; the analysis of the size of $\acdDAG{\A}$ will be a refinement of this proof.

A polynomial upper bound on the size of $\acdDAG{\A}$ implies Theorem~\ref{th-comp:compt-ACD-DAG-poly-ZDAG} simply by combining the fact that computing the children of a node of $\acdDAG{\A}$ requires a time polynomial in $\size{\zielonkaDAG{\F}}+ |Q|$ (Lemma~\ref{lemma-comp:analysis-comptChildrenZT}), and the fact that we call $\computeChildrenACD$ exactly once per node of $\acdDAG{\A}$ in Algorithm~\ref{algo:computationACD} (as we never push back in $\nodesStack$ any set that was already explored, see line~\ref{line-Cacd:newChildren}-\ref{line-Cacd:addChildrenNodes}).

To establish Theorem~\ref{th-comp:compt-ACD-poly-ZT},
we remark that to compute $\acd{\A}$ from $\zielonkaTree{\F}$ and $\A$ we simply fold $\zielonkaTree{\F}$ to obtain $\zielonkaDAG{\F}$, apply Theorem~\ref{th-comp:compt-ACD-DAG-poly-ZDAG} to get $\acdDAG{\A}$, and then unfold the latter to obtain $\acd{\A}$.
The first two steps require a time polynomial in $\size{\zielonkaDAG{\F}}+|Q| \leq \size{\zielonkaTree{\F}}+|Q|$, while the third step takes a time polynomial in $\size{\acd{\A}}\leq |Q|\cdot|\zielonkaTree{\F}|$.

\paragraph*{Upper bound on the size of the ACD}



\begin{proposition}\label{prop-comp:size-ACD}
	Let $\A$ be a "Muller" automaton and $\F$ the family defining its "acceptance condition". Then,
	$  |\acd{\A}|\leq |Q|\cdot|\zielonkaTree{\F}|$.
\end{proposition}

We start by giving a technical lemma that will be useful for the subsequent analysis.

\begin{lemma}\label{lemma-comp:many-descendants}
	Let $C\subseteq \GG$ and let $n_C$ be a node in $\zielonkaTree{\F}$ 
	such that $C\subseteq \nu(n_C)$. 
	Let $D_1,\dots, D_k$ be $k$ subsets of $C$ such that, for all $i\neq j$, $C\in \F \iff D_i\notin \F \iff D_i\cup D_j\in \F$. 
	Then, there are $k$  strict descendants of $n_C$, $n_1,\dots,n_k$, such that $D_i\subseteq \nu(n_i)$, $ \nu(n_i)\in \F \iff D_i\in \F$ and such that nodes $n_i$ are pairwise incomparable for the ancestor relation.
	Moreover, these nodes can be computed in polynomial time in $|\zielonkaTree{\F}|$.
\end{lemma}
\begin{proof}
	To simplify notations we assume that $C\in \F$ and $D_i\notin \F$ (the proof is symmetric in the other case).
	For each $D_i$ we pick a node $n_i$ which is a descendant of $n_C$, such that $D_i\subseteq \nu(n_i)$ and maximal for $\ancestor$ with this property. In particular, $n_i$ is "square" and a strict descendant (Lemma~\ref{lemma-zt:accepting-set-in-ZT}). 
	We prove that, for $j\neq i$, $D_j\nsubseteq \nu(n_i)$, implying that $n_i$ and $n_j$ are incomparable for the ancestor relation.
	Suppose by contradiction that for some $j\neq i$, $D_j\subseteq \nu(n_i)$. Then, $D_j\subseteq D_i\cup D_j\subseteq \nu(n_i)$, so, by Lemma~\ref{lemma-zt:accepting-set-in-ZT}, $D_i\cup D_j\notin \F$, contradicting the hypothesis.
\end{proof}

\begin{lemma}\label{lemma-comp:size-treeVertex-noMore-ZT}
	For every state $q$, the tree $\treeVertex{q}$ has size at most $|\zielonkaTree{\F}|$.
\end{lemma}
\begin{proof}
	We define in a top-down fashion an injective function $f\colon\treeVertex{q}\to \zielonkaTree{\F}$. For the base case, we send the root of $\treeVertex{q}$ to the root of $\zielonkaTree{\F}$.
	Let $n$ be a node in $\treeVertex{q}$ such that $f(n)$ has been defined, and let $n_1,\dots, n_k$ be its children. We let $C_n = \colAut(\nuAcd(n))$ and $D_i = \colAut(\nuAcd(n_i))$ be the colours labelling the "cycles" of these nodes. These sets satisfy that for
	$i\neq j$, $C_n\in \F \iff D_i\notin \F \iff D_i\cup D_j\in \F$. Indeed, if the union of $D_i$ and $D_j$ does not change the acceptance, we could take the union of the corresponding "cycles", contradicting maximality.
	Lemma~\ref{lemma-comp:many-descendants} provides $k$ descendants of $f(n)$ such that the subtrees rooted at them are pairwise disjoint. This allows to define $f(n_i)$ for all $i$ and carry out the induction.
\end{proof}

We conclude  that the size of $\acd{\A}$ is polynomial in  $|Q|+|\zielonkaTree{\F}|$, concluding the proof of Proposition~\ref{prop-comp:size-ACD}:
\[ |\acd{\A}|\leq \sum_{q\in Q} |\treeVertex{q}| \leq |Q|\cdot|\zielonkaTree{\F}|.\]	


\paragraph*{Upper bound on the size of the ACD-DAG}

\begin{proposition}\label{prop-comp:size-ACD-DAG}
	Let $\A$ be a "Muller" automaton and $\F$ the family defining its "acceptance condition". Then,
	$  |\acdDAG{\A}|\leq |Q|\cdot|\zielonkaDAG{\F}|$.
\end{proposition}

For obtaining this result, we want to follow the same proof scheme than in Proposition~\ref{prop-comp:size-ACD}: our objective is to show that for all $q\in Q$, the "local subDAG" $\dagVertex{q}$ can be embedded in  $\zielonkaDAG{\F}$. However, we face a technical difficulty; in the case of the "ACD" we had that the subtrees rooted at $k$ incomparable nodes were disjoint, which allowed us to carry out the recursion smoothly. This property no longer holds in "DAGs". 

\begin{lemma}\label{lemma-comp:size-dagVertex-noMore-ZDAG}
	For every state $q$, the "DAG" $\dagVertex{q}$ has size at most $|\zielonkaDAG{\F}|$.
\end{lemma}
\begin{proof}
	We will define an injective function $f\colon\dagVertex{q}\to \zielonkaDAG{\F}$. For a node $n$ in $\dagVertex{q}$, we let $C_n = \colAut(\nuAcd(n))$ be the set of colours appearing in the label of $n$. 
	\AP If $n$ is not the root, we define $\intro*\pred(n)$ to be an immediate ancestor of $n$ (that is, $n$ is a child of $\pred(n)$). We let $\intro*\predBranch(n)$ be the sub-branch of nodes above $n$ obtained by  successive applications of $\pred$, that is, $\predBranch(n)= \{n'\in \dagVertex{q} \mid n' = \pred^k(n) \text{ for some } k\}$. We note that the elements of $\predBranch(n)$ are totally ordered by $\ancestorDAG$ ($n$ being the maximal node and the root the minimal one).
	
	We define $f$ recursively: For the root $n_0$ of $\dagVertex{q}$, we let $f(n_0)$ be a maximal node (for $\ancestorDAG$) in $\zielonkaDAG{\F}$ containing $C_{n_0}$ in its label. 
	For $n$ a node such that we have define $f$ for all its ancestors, we let $f(n)$ be a maximal node (for $\ancestorDAG$) in the subDAG rooted at $f(\pred(n))$ containing $C_{n}$ in its label. 
	We remark that $f(n)$ is a "round node" if and only if $n$ is a "round node@@acd" (by Lemma~\ref{lemma-zt:accepting-set-in-ZT}). Also, if $n'$ is an ancestor of $n$ in $\predBranch(n)$, then $f(n')$ is an ancestor of $f(n)$ in $\zielonkaDAG{\F}$.
	
	We prove now the injectivity of $f$. Let $n_1,n_2$ be two different nodes in $\dagVertex{q}$ (that is, $\nuAcd(n_1)\neq\nuAcd(n_2)$). First, we show that the colours appearing in their labels must differ.
	\begin{claim}
		It is satisfied that $C_{n_1}\neq C_{n_2}$.
	\end{claim}
	\begin{claimproof}
		Suppose by contradiction that $C_{n_1}= C_{n_2}$. 
		Then, any node $n$ containing $\nuAcd(n_1)$ in its label satisfies that $\nuAcd(n)$ is an "accepting@@cycle" "cycle" if and only if $\nuAcd(n)\cup \nuAcd(n_2)$ is an "accepting@@cycle" "cycle".
		Let $n$ be a node of minimal depth such that $\nuAcd(n_1)\subseteq \nuAcd(n)$ and $\nuAcd(n_2)\nsubseteq \nuAcd(n)$. 
		The label of an immediate predecessor of $n$ contains $\nuAcd(n_1)\cup\nuAcd(n_2)$ by minimality. This leads to a contradiction, as $\nuAcd(n)\subsetneq \nuAcd(n)\cup \nuAcd(n_2)$, so $\nuAcd(n)$ would not be a maximal subcycle producing an alternation in the acceptance status.
	\end{claimproof}
	
	We assume w.l.o.g. that $n_1$ is "round@@acd" (that is, $C_{n_1}\in \F$).
	Suppose by contradiction that $f(n_1)=f(n_2)$. Then, $n_2$ is also "round@@acd", and it is satisfied that $C_{n_1}\cup C_{n_2}\subseteq f(n_1)$, by definition of $f$. Again by definition of $f$, no child of $f(n_1)$ contains $C_1\cup C_2$, so, by Lemma~\ref{lemma-zt:accepting-set-in-ZT}, $C_1\cup C_2\in \F$.
	Let $n'$ be the minimal node in $\predBranch(n_1)$ such that $\nuAcd(n_2)\subseteq \nuAcd(n')$. We do the prove for the case in which $n'$ is "round@@acd", the other case is symmetric. 
	Let $\tilde{n}$ be the child of $n'$ in $\predBranch(n')$, which is a "square node@@acd". We claim that the following three properties hold:
	\begin{enumerate}[label=\roman*)]
		\item $C_{\tilde{n}}\cup C_{n_2} \in \F$,
		\item $C_{\tilde{n}}\cup C_{n_2} \subseteq f(\tilde{n})$, and
		\item no child of $f(\tilde{n})$ contains $C_{\tilde{n}}\cup C_{n_2}$.
	\end{enumerate}
	This leads to a contradiction, as the second and third items, combined with Lemma~\ref{lemma-zt:accepting-set-in-ZT} and the fact that $f(\tilde{n})$ is a "square node", imply that $C_{\tilde{n}}\cup C_{n_2} \notin \F$. We prove the three items:
	
	\begin{enumerate}[label=\roman*)]
		\item Follows from the fact that $\nuAcd(\tilde{n})$ is a maximal "rejecting@@cycle" "cycle" of  $\nuAcd(n')$, but $\nuAcd(n')$ contains $\nuAcd(\tilde{n})\cup \nuAcd(n_2)$.
		\item By definition of $f$, $C_{\tilde{n}} \subseteq f(\tilde{n})$. Also, the node $f(\tilde{n})$ is an ancestor of $f(n_2)$, so $C_{n_2}\subseteq f(n_2) \subseteq f(\tilde{n})$.
		\item By definition of $f$, no child of $f(\tilde{n})$ contains $C_{\tilde{n}}$ in its label.\qedhere
	\end{enumerate}
\end{proof}

We conclude  that the size of $\acdDAG{\A}$ is polynomial in  $|Q|+|\zielonkaDAG{\F}|$:
\[ |\acdDAG{\A}|\leq \sum_{q\in Q} |\dagVertex{q}| \leq |Q|\cdot|\zielonkaDAG{\F}|.\]

\section{Minimisation of colours and Rabin pairs}
\label{sec:minimisation-colours}

We consider the problem of minimising the representation of the "acceptance condition" of "automata".
That is, given an "automaton" $\A$ using a "Muller" (resp. "Rabin") "acceptance condition", what is the minimal number of colours (resp. "Rabin  pairs") needed to define an "equivalent acceptance condition over" $\A$?



We first study the question of the minimisation of colours for "Muller languages", without taking into account the structure of the "automaton". We show that given the "Zielonka DAG" of the condition (resp. set of "Rabin pairs"), we can minimise its number of "colours" (resp. number of "Rabin pairs") in polynomial time (Theorems~\ref{th-min:colorMinML-poly} and~\ref{th-min:RabinPairsMinML-poly}). 
An alternative point of view over the minimisation of "Rabin pairs", using "generalised Horn formulas", is presented in Appendix~\ref{sec:Gen-Horn}.
Then, we consider the question taking into account the structure of the "automaton". Surprisingly, we show that in this case both problems are $\NP$-complete, and this hardness results holds even if the "ACD" is given as input (Theorems~\ref{th-min:colorMin-Aut-NPhard} and~\ref{th-min:RabinPairsMin-Aut-NPHard}).
We highlight that the $\NP$ upper bound is also non-trivial.

\subsection{Minimisation of the representation of Muller languages in polynomial time}

\subsubsection{Minimisation of colours for Muller languages}
\AP We say that a "Muller language" $L\subseteq \SS^\oo$ is ""$k$-colour type@@lang"" if there is a set of $k$ colours~$\GG$, a "Muller language" $L'\subseteq \GG^\oo$ and a mapping $\phi\colon \SS \to \GG$ such that for all $w\in \SS^\oo$, $w\in L \iff \phi(w)\in L'$. 
\begin{remark}
	A "Muller language" $L\subseteq \SS^\oo$ is "$k$-colour type@@lang" if and only if it can be "recognised" by a "deterministic" "Muller" "automaton" with one state using $k$ "output colours".
	However, this is \emph{not} the same as having a "Muller automaton" "recognising" $\MullerC{\F}{\SS}$ using at most $k$ colours (in general, automata with more states will use fewer colours).
	
	Also, $L$ is "$k$-colour type@@lang" if and only if all automata $\A$ using $L$ as "acceptance condition" can be "relabelled" with an "equivalent@@cond" "Muller condition" using at most $k$ colours.
\end{remark}

\AP
\begin{center}
	\fbox{\begin{tabular}{rl}
			{\textbf{Problem:}} & \intro*\pbColorMinML{} \\[1mm]
			{\textbf{Input:}} & A "Muller language" $\MullerC{\F}{\SS}$ represented by the "Zielonka@@DAG"\\ 
			& "DAG@@Zielonka" $\zielonkaDAG{\F}$ and a positive integer $k$.\\[1mm]   
			{\textbf{Question:}} & Is $\MullerC{\F}{\SS}$ "$k$-colour type@@lang"? 
	\end{tabular}} 
\end{center}

We could have chosen other representations of $\MullerC{\F}{\SS}$ for the input of this problem (mainly, "colour-explicitly@@Muller" or as a "Zielonka tree"). We have chosen to specify the input as a "Zielonka DAG", as it is more succinct than the other representations (c.f. Figure~\ref{fig-prel:representationMuller} and Propositions~\ref{prop-size:explicit-to-DAG},~\ref{prop-size:size-ZT-vs-ZDAG}).  We now prove that this problem can be solved in polynomial time if $\F$ is represented as a "Zielonka DAG", which implies that it can be equally solved in polynomial time if $\F$ is represented "colour-explicitly@@Muller" or as a "Zielonka tree".

\begin{theorem}[Tractability of minimisation of colours for Muller languages]\label{th-min:colorMinML-poly}
	The problem \pbColorMinML{} can be solved in polynomial time. 
\end{theorem}
\begin{proof}
	\AP We say that two letters $a,b\in \SS$ are ""$\F$-equivalent"", written $a\sim_\F b$, if they appear in the same nodes of the "Zielonka DAG" $\zielonkaDAG{\F}$, i.e., for every node $n$ of the "Zielonka DAG" of $\F$, $a\in \nu(n) \iff b\in \nu(n)$.
	It is immediate to check that $\sim_\F$ is indeed an equivalence relation.
	We let $[a]$ denote the equivalence class of $a$ for $\sim_\F$ and $\quotient{\SS}{\sim_\F}$ the set of equivalence classes. For $C\subseteq \SS$, we write $\mathsf{sat}(C) = \bigcup_{a\in C}[a]\subseteq \SS$ and $\pi(C)=\{[a] \mid a\in C\}$.
	
	\begin{claim}
		\label{claim-comp:saturation-classes}
		For all $C\subseteq \SS$, $C\in \F \iff \mathsf{sat}(C) \in \F$.
	\end{claim}
	\begin{claimproof}
		Assume that $C\in \F$ (the case $C\notin \F$ is symmetric).
		There is a "round" node $n$ in $\zielonkaDAG{\F}$ such that $C \subseteq \nu(n)$ and for all child $n'$ of $n$, $C\nsubseteq \nu(n')$.
		By definition of $\sim_{\F}$, $\mathsf{sat}(C)\subseteq \nu(n)$ and $\mathsf{sat}(C) \nsubseteq \nu(n')$, for all child $n'$ of $n$. We conclude that $\mathsf{sat}(C) \in \F$. 
	\end{claimproof}
	
	\begin{claim}
		$\MullerC{\F}{\SS}$ is "$k$-colour type@@lang", for $k$ the number of equivalent classes of $\sim_\F$.
	\end{claim}
	
	\begin{claimproof}
		We define the mapping $\phi : \SS \to \quotient{\SS}{\sim_\F}$ by $\phi(a) = [a]$, and let $\widetilde{\F} = \{ \pi(C)\mid C\in \F\}$.
		We show that for all $C \in \powplus{\SS}$, $C \in \F \iff \phi(C) \in \widetilde{\F}$.
		The left-to-right implication is immediate by definition of  $\widetilde{\F}$.
		For the other direction, suppose that $\phi(C) \in \widetilde{\F}$. Then there exists $D \in \F$ such that $\pi(C) = \pi(D)$, which implies $\mathsf{sat}(C) = \mathsf{sat}(D)$.
		By Claim~\ref{claim-comp:saturation-classes}, $\mathsf{sat}(D) =\mathsf{sat}(C) \in \F$, so, by the same claim, $C\in \F$ as wanted. 
	\end{claimproof}
	
	For the converse implication, we first prove that non-equivalent colours can be ``separated'' by the family $\F$.
	
	\begin{claim}
		\label{claim-comp:witness-set}
		For all $a,b \in \SS$, if $a\nsim_\F b$ then there exists a set $S\subseteq \SS$ such that one of $S \cup \set{a}, S\cup \set{b}, S\cup\set{a,b}$ is in $\F$ and another is not in $\F$.
	\end{claim}
	
	\begin{claimproof}
		As $a\nsim_\F b$, there is a node of $\zielonkaTree{\F}$ whose label $\nu(n)$ contains $a$ but not $b$ or $b$ and not $a$.
		We select such a node $n$ of minimal depth. We assume that $\nu(n)$ contains $a$ and not $b$, and that it is "square" (the other cases are similar).
		
		As the root is labelled $\SS$, $n$ is not the root. Let $m$ be the parent of $n$, as we took $n$ of minimal depth, the label $\nu(m)$ contains $a$ and $b$.
		Let $S= \nu(n) \setminus \set{a}$. As $n$ is "square", $S \cup \set{a}$ is a maximal rejected subset of $\nu(m)$, hence $S \cup \set{a}$ is rejected and $S \cup \set{a,b}$ is accepted, as wanted.
	\end{claimproof}

	\begin{claim}\label{claim-comp:k-col-iff-k-classes}
		If $\MullerC{\F}{\SS}$ is "$k$-colour type@@lang", then there are at most $k$ equivalence classes for the "$\sim_\F$ relation".
	\end{claim}
	
	\begin{claimproof}
		By definition, there is an alphabet $\GG$ with $k$ colours, a family of sets $\F' \subseteq \powplus{\GG}$ and a morphism $\phi \colon \SS\to \GG$ such that for all $S \in \pow{\GG}$, $S\in \F \iff \phi(S)\in \F'$.
		
		Let $a, b \in \SS$ such that $a \nsim_\F b$. By Claim~\ref{claim-comp:witness-set}, there exists a set $S$ such that one of $S \cup \set{a}, S\cup \set{b}, S\cup\set{a,b}$ is in $\F$ and another is not in $\F$. Hence $\phi(a) \neq \phi(b)$, as otherwise the image by $\phi$ of those three sets would be the same.
		Hence two colours of different equivalence classes for $\sim_\F$ cannot be mapped by $\phi$ to the same colour. As a result, there are at most $\size{\GG} = k$ classes for $\sim_\F$.
	\end{claimproof}
	Therefore, in order to minimise the number of required colours, we need to compute the classes of the "$\F$-equivalence" relation. This can be directly done by inspecting the "Zielonka DAG".	
\end{proof}

\subsubsection{Minimisation of Rabin pairs for Rabin languages}

In this section we tackle the minimisation of the number of "Rabin pairs" to represent "Rabin languages". We provide a polynomial-time algorithm which turns a family of "Rabin pairs" into an "equivalent@@cond" one with a minimal number of pairs. The algorithm comes down to partially computing the Zielonka tree of the input "Rabin language" from top to bottom, and stopping when we have obtained a set of "Rabin pairs" "equivalent@@cond" to the input.

We present the algorithm in a different way in order to clarify the proofs, in particular the proof that the resulting number of pairs is minimal.

\AP We say that a "Rabin language" $L\subseteq \SS^\oo$ is ""$k$-Rabin-pair type@@lang"" if there is a family of $k$ "Rabin pairs" $\R$ over some set of colours $\GG$ and a mapping $\phi\colon \SS \to \GG$ such that for all $w\in \SS^\oo$, $w\in L \iff \phi(w)\in \RabinC{\R}{\GG}$. 

\begin{remark}
	A "Rabin language" $L\subseteq \SS^\oo$ is "$k$-Rabin-pair type@@lang" if and only if it can be "recognised" by a "deterministic" "Rabin" "automaton" with one state using $k$ "Rabin pairs".
\end{remark}

\begin{remark}\label{rmk-min:Rabin-type-same-colours}
	If $L\subseteq \SS^\oo$ is "$k$-Rabin-pair type@@lang", then there exists a family of $k$ "Rabin pairs" $\R'$ over the same alphabet $\SS$ such that $L=\RabinC{\R'}{\SS}$.
\end{remark}
\begin{proof}
	Let $\R=\{(\greenPair_1,\redPair_1),\dots,(\greenPair_k,\redPair_k)\}$ be a set of "Rabin pairs" over $\GG$ and let $\phi\colon \SS \to \GG$ such that for all $w\in \SS^\oo$, $w\in L \iff \phi(w)\in \RabinC{\R}{\GG}$. 
	It suffices to define $\R'=\{(\greenPair_1',\redPair_1'),\dots,(\greenPair_k',\redPair_k')\}$ with $\greenPair_i' = \inv{\phi}(\greenPair_i)$ and $\redPair_i' = \inv{\phi}(\greenPair_i)$.	
\end{proof}

\AP
\begin{center}
	\fbox{\begin{tabular}{rl}
			{\textbf{Problem:}} & \intro*\pbRabinPairMinML{} \\[1mm]
			{\textbf{Input:}} & A family of "Rabin pairs" $\R$ over $\SS$ and a positive integer $k$.\\[1mm]  
			{\textbf{Question:}} & Is $\RabinC{\R}{\SS}$ "$k$-Rabin-pair type@@lang"? 
	\end{tabular}} 
\end{center} 

Our main result of this section is the following theorem:

\begin{theorem}[Tractability of minimisation of Rabin pairs for Rabin languages]\label{th-min:RabinPairsMinML-poly}
	The problem \pbRabinPairMinML{} can be solved in polynomial time.
\end{theorem}

Given a set of "Rabin pairs" $\R = \{(\greenPair_1,\redPair_1),\dots,(\greenPair_r,\redPair_r)\}$ over $\SS$ and a set $S\subseteq \SS$, \AP we say that $S$ ""satisfies@@Rabin"" (or that is \emph{accepted by}) $\R$ if, for some $i$, $S\cap \greenPair_i\neq \emptyset$ and $S\cap \redPair_i=\emptyset$.
\AP Otherwise, we say that $S$ is rejected by $\R$.
By a small abuse of notation, we write $S\in \RabinC{\R}{\SS}$ (resp.~$S\notin \RabinC{\R}{\SS}$) if $S$ is "accepted by@@Rabin" (resp. "rejected by@@Rabin") $\R$. We define the same notions for "Streett conditions" symmetrically.

Before we present the algorithm, we observe that two critical operations can be done in polynomial-time.

\begin{lemma}
	\label{lem:Rabin-max-reject-PTIME}
	Let $\R$ be a family of "Rabin pairs" over $\SS$ and let $S \subseteq \SS$. There exists a maximum subset of $S$ "rejected by@@Rabin" $\R$, and it is computable in polynomial time.
\end{lemma}

\begin{proof}
	We describe an algorithm building a decreasing sequences of subsets of $S$.
	Initially, set $T= S$. While there exists $(\greenPair,\redPair) \in \R$ "satisfied by@@Rabin" $T$, set $T=T \setminus \greenPair$.
	This algorithm maintains the invariant that all sets "rejected by@@Rabin" $\R$ should be included in $T$.	
	Furthermore, it terminates in at most $\size{\SS}$ iterations as $T$ strictly decreases at each step. 
	In the end, we obtain a set that does not satisfy $\R$, and that is maximum by the invariant property. 
\end{proof}

\begin{lemma}
	\label{lem:Rabin-difference-PTIME}
	Given two families $\R, \R'$ of "Rabin pairs" over $\SS$, there is a polynomial-time algorithm that checks whether  $\RabinC{\R}{\SS} \nsubseteq \RabinC{\R'}{\SS}$ and returns a maximal set $S \in \RabinC{\R}{\SS} \setminus \RabinC{\R'}{\SS}$ if it is the case.
\end{lemma}

\begin{proof}
	For each $(\greenPair,\redPair) \in \R$, we apply the following procedure.

	Set $S= \SS \setminus \redPair$, and define $S_{(\greenPair, \redPair)}$ as the maximum subset of $S$ not satisfying $\R'$. We can compute $S_{(\greenPair, \redPair)}$ by Lemma~\ref{lem:Rabin-max-reject-PTIME}.
	If $S_{(\greenPair, \redPair)} \cap \greenPair = \emptyset$, then it does not satisfy $(\greenPair, \redPair)$, hence $\RabinC{(\greenPair, \redPair)}{\SS} \setminus \RabinC{\R'}{\SS} = \emptyset$. If this is the case for all $(\greenPair, \redPair)$, then we can conclude that $\RabinC{\R}{\SS} \subseteq \RabinC{\R'}{\SS}$. Otherwise, we can select a maximal set among the $S_{(\greenPair, \redPair)}$ such that $S_{(\greenPair, \redPair)} \cap \greenPair \neq \emptyset$. 
	
	This yields a maximal set in $\RabinC{\R}{\SS} \setminus \RabinC{\R'}{\SS}$.
\end{proof}

In Algorithm~\ref{algo-min-rabin} we give a procedure minimising the number of "Rabin pairs". We remark that the "Rabin condition" built by this algorithm uses the same set of colours as the input "Rabin condition". 

\begin{algorithm}
	\caption{Minimisation algorithm for Rabin conditions.}
	\label{algo-min-rabin}
	\begin{algorithmic}
		\Statex \textbf{Input:} A set of "Rabin pairs" $\R$ over $\SS$
		\State $\R_{\min} \gets \{\}$
		\While{$\RabinC{\R}{\SS} \nsubseteq \RabinC{\R_{\min}}{\SS}$}
		
		\State $S \gets $ maximal set in $\RabinC{\R}{\SS} \setminus \RabinC{\R_{\min}}{\SS}$
		
		\State $T \gets $ maximum subset of $S$ not in $\RabinC{\R}{\SS}$
		
		\State $\R_{\min} \gets \R_{\min} \cup \{(\SS \setminus T, \SS \setminus S)\}$
		\EndWhile
		\State \Return $\R_{\min}$
	\end{algorithmic}
\end{algorithm}

\begin{lemma}
	\label{lem:algo-Rabin-terminates}
	Algorithm~\ref{algo-min-rabin} terminates and $\RabinC{\R_{\min}}{\SS} = \RabinC{\R}{\SS}$.
\end{lemma}

\begin{proof}
	The algorithm clearly terminates, as $\RabinC{\R_{\min}}{\SS}$ increases at each iteration of the loop. Thus, we eventually get out of the loop, hence $\RabinC{\R}{\SS} \subseteq \RabinC{\R_{\min}}{\SS}$. Furthermore $\RabinC{\R_{\min}}{\SS} \subseteq \RabinC{\R}{\SS}$ is a loop invariant: at the start we have $\RabinC{\R_{\min}}{\SS} = \emptyset$, and at each loop iteration we add to $\R_{\min}$ a pair $(\SS \setminus T, \SS \setminus S)$ such that $S \in \RabinC{\R}{\SS}$ and $T$ is the maximum subset of $S$ rejected by $\R$. As a consequence, since $\RabinC{(\SS \setminus T, \SS \setminus S)}{\SS}$ contains only sets included in $S$ but not in $T$, $\RabinC{(\SS \setminus T, \SS \setminus S)}{\SS} \subseteq \RabinC{\R}{\SS}$ and thus the invariant is maintained.
\end{proof}

\begin{lemma}
	\label{lem:algo-Rabin-no-union}
	Let $\R_{\min}$ be the set of "Rabin pairs" obtained by applying Algorithm~\ref{algo-min-rabin} on a set of "Rabin pairs" $\R$.
	Let $(\greenPair_1, \redPair_1) \neq (\greenPair_2, \redPair_2) \in \R_{\min}$, then $\SS \setminus \redPair_1 \cup \SS\setminus \redPair_2$ is not "accepted@@Rabin" by $\R_{\min}$.
\end{lemma}

\begin{proof}
	Suppose by contradiction that $\SS \setminus \redPair_1 \cup \SS\setminus \redPair_2$ is "accepted@@Rabin" by $\R_{\min}$. 
	Then there exists $(\greenPair_3, \redPair_3) \in \R_{\min}$ accepting $\SS \setminus \redPair_1 \cup \SS \setminus \redPair_2$. 
	As a consequence, neither $\SS \setminus \redPair_1$ nor $\SS \setminus \redPair_2$ intersects $\redPair_3$, and one of them intersects $\greenPair_3$.
	Therefore, $(\greenPair_3, \redPair_3)$ necessarily accepts either $\SS \setminus \redPair_1$ or $\SS \setminus \redPair_2$. Without loss of generality, we assume that it accepts $\SS \setminus \redPair_1$. In particular, we have $\redPair_3 \subseteq \redPair_1$.
	
	During the execution of Algorithm~\ref{algo-min-rabin} on $\R$ resulting in $\R_{\min}$, $S$ must have taken the values $\SS \setminus \redPair_j$ for both $j=1$ and $j=3$, starting with $j=3$ since $S$ is always taken maximal and $\SS \setminus \redPair_1 \subseteq \SS \setminus \redPair_3$. However, after adding $(\greenPair_3, \redPair_3)$ to $\R_{\min}$, $\SS \setminus \redPair_1$ is accepted by $\R_{\min}$, contradicting the fact that $(\greenPair_1, \redPair_1)$ is in $\R_{\min}$ in the end.
\end{proof}

By Remark~\ref{rmk-min:Rabin-type-same-colours}, it suffices to check minimality of $\R_{\min}$ amongst families of "Rabin pairs" over the alphabet $\SS$.

\begin{lemma}
	\label{lem:algo-min-Rabin-minimal}
	Let $\R$ and $\Tilde{\R}$ be two families of "Rabin pairs" over $\SS$ such that $\RabinC{\R}{\SS} = \RabinC{\Tilde{\R}}{\SS}$, and let $\R_{\min}$ the family returned by Algorithm~\ref{algo-min-rabin} when applied on $\R$.
	Then, $|\R_{\min}|\leq |\Tilde{\R}|$.	
\end{lemma}

\begin{proof}
	In order to prove the lemma, we map each pair $(\greenPair, \redPair)$ of $\R_{\min}$ to a pair $(\Tilde{\greenPair}, \Tilde{\redPair})$ of $\Tilde{\R}$ and prove that the mapping is injective.
	
	By Lemma~\ref{lem:algo-Rabin-terminates}, we have $\RabinC{\R}{\SS} = \RabinC{\R_{\min}}{\SS}$, and thus $\RabinC{\Tilde{\R}}{\SS} = \RabinC{\R_{\min}}{\SS}$.
	For all $(\greenPair, \redPair) \in \R_{\min}$, as $\SS \setminus \redPair$ is "accepted by@@Rabin" $\R_{\min}$ and $\RabinC{\Tilde{\R}}{\SS} = \RabinC{\R_{\min}}{\SS} $, we can find $(\Tilde{\greenPair}, \Tilde{\redPair}) \in \Tilde{\R}$ such that $\SS \setminus \redPair$ intersects $\Tilde{\greenPair}$ but not $\Tilde{\redPair}$.
	
	Now assume that there exist $(\greenPair_1, \redPair_1) \neq (\greenPair_2, \redPair_2) \in \R_{\min}$ such that $(\Tilde{\greenPair}_1, \Tilde{\redPair}_1) = (\Tilde{\greenPair}_2, \Tilde{\redPair}_2)$.
	The pair $(\Tilde{\greenPair}_1, \Tilde{\redPair}_1)$ then accepts both $\SS \setminus \redPair_1$ and $\SS \setminus \redPair_2$. As a consequence, both $\SS \setminus \redPair_1$ and $\SS \setminus \redPair_2$ intersect $\Tilde{\greenPair}_1$ and neither intersects $\Tilde{\redPair}_1$, thus $\SS \setminus \redPair_1 \cup \SS \setminus \redPair_2$ is also accepted by $(\Tilde{\greenPair}_1, \Tilde{\redPair}_1)$. This contradicts Lemma~\ref{lem:algo-Rabin-no-union}. 
	
	As a consequence, for all $(\greenPair_1, \redPair_1), (\greenPair_2, \redPair_2) \in \R_{\min}$, $(\Tilde{\greenPair}_1, \Tilde{\redPair}_1) \neq (\Tilde{\greenPair}_2, \Tilde{\redPair}_2)$. As a result, $\Tilde{\R}$ must contain at least as many pairs than $\R_{\min}$, proving the lemma.
\end{proof}

\begin{proposition}
	\label{prop:Rabin-min-in-PTIME}
	Algorithm~\ref{algo-min-rabin} terminates in polynomial time and returns a family of "Rabin pairs" with the same "Rabin language" as the input and with a minimal number of pairs.
\end{proposition}

\begin{proof}
	By Lemmas~\ref{lem:algo-Rabin-terminates} and~\ref{lem:algo-min-Rabin-minimal}, the algorithm terminates and returns a family of "Rabin pairs" with the desired property.
	Furthermore, since a pair is added to $\R_{\min}$ at each iteration of the loop, and since the resulting family contains at most $\size{\R}$ pairs (by minimality), the algorithm terminates after at most $\size{\R}$ iterations. Finally, by Lemmas~\ref{lem:Rabin-max-reject-PTIME} and~\ref{lem:Rabin-difference-PTIME}, each iteration can  be done in polynomial time, hence the algorithm terminates in polynomial time.
\end{proof}

\begin{corollary}
	\label{cor:Streett-condition-minimisation}
	Given a set of "Rabin pairs" $\R$, one can compute a set $\R_{min}$ with the same "Streett language" and a minimal number of clauses.
\end{corollary}

\begin{proof}
	It suffices to observe that two sets of "Rabin pairs" $\R_1, \R_2$ have the same "Rabin language" if and only if they have the same "Streett language". Hence by Proposition~\ref{prop:Rabin-min-in-PTIME}, Algorithm~\ref{algo-min-rabin} also minimises the number of pairs for the "Streett language".
\end{proof}

\subsection{Minimisation of acceptance conditions on top of an automaton is $\NP$-complete}\label{subsec:min-colours-automata}

We now consider the problem of minimising the number of colours or "Rabin pairs" used by a "Muller" or "Rabin condition" over a fixed "automaton".
We could expect that it is possible to generalise the previous polynomial time algorithms by using the "ACD", instead of the "Zielonka DAG". 
Quite surprisingly, we show that these problems become $\NP$-complete when taking into account the structure of the automata.


\subsubsection{Minimisation of colours on top of a Muller automaton}

\AP We say that a "deterministic" "Muller" "automaton" $\A$ is ""$k$-colour type@@TS"" if we can "relabel" it with a  "Muller condition"  using at most $k$ "output colours" that is "equivalent over" $\A$.

\AP
\begin{center}
	\fbox{\begin{tabular}{rl}	
			{\textbf{Problem:}} & \intro*\pbColorMinAut{} \\[1mm]
			{\textbf{Input:}} & A "deterministic" "Muller" "automaton" $\A$ and a positive integer $k$.\\[1mm]   
			{\textbf{Question:}} & Is $\A$ "$k$-colour type@@TS"?  
	\end{tabular}} 
\end{center}

Muller automata are sometimes defined using several colours per edge, which may allow one to use less colours and simplify the "Muller condition"~\cite{HOAFormat2015}.  
Formally, the "output alphabet" of such automata is of the form $\pow{\GG}$, and the acceptance condition is given as a family $\F$ of "accepting sets" over $\GG$ (and not over $\pow{\GG}$). An infinite sequence of outputs $w\in (\pow{\GG})^\oo$ is accepting if 
$\{c\in \GG\mid c\in w_i \text{ for infinitely many } i\}$ belongs to $\F$.

\AP We say that a "deterministic" "Muller" "automaton" $\A$ is ""$k$-multiple-colour type@@TS"" if we can "relabel" it with an "equivalent@@cond" "Muller condition" using at most $k$ "colours", with possibly several colours per edge.
The problem \intro*\pbMultiColourMinAut{} consists in determining whether an input "deterministic" "Muller automaton" $\A$ is "$k$-multiple-colour type@@TS".	\\


We remark that we have not specified the representation of the "acceptance condition" of $\A$ in the previous problems; therefore, they admit different variants according to this representation. We will show that for the three representations we are concerned with ("colour-explicit@@Muller", "Zielonka tree" and "Zielonka DAG"), both problems \pbColorMinAut{} and \pbMultiColourMinAut{} are $\NP$-complete. 
This implies that the problem is $\NP$-hard even if the "ACD" is provided as input, by Theorem~\ref{th-comp:compt-ACD-poly-ZT}.\\

Hugenroth showed\footnote{As of today, the proof is not currently publicly available online, we got access to it by a personal communication. The statement of the theorem only express the $\NP$-hardness for the "colour-explicit@@Muller" representation, but a look into the reduction works unchanged if the condition is given as a "Zielonka tree".} that, \emph{for state-based "automata"}, the problem \pbColorMinAut{} is $\NP$-hard when the "acceptance condition" of $\A$ is represented "colour-explicitly@@Muller" or as a "Zielonka tree"~\cite{Hugenroth23ZielonkaDAG}. 
However, it is not straightforward to generalise it to \emph{transition-based "automata"}.

\begin{theorem}[$\NP$-completeness of minimisation of colours for Muller automata]\label{th-min:colorMin-Aut-NPhard}
	The problems \pbColorMinAut{} and \pbMultiColourMinAut{} are $\NP$-complete, if the "acceptance condition" $\MullerC{\F}{\GG}$ of $\A$ is represented "colour-explicitly@@Muller", as a "Zielonka tree" or as the "ACD" of $\A$.
\end{theorem}

We note that both the $\NP$-hardness and the fact that these problems lie in \NP{} are not obvious: we could be tempted to guess an acceptance condition on the same automaton structure and check equivalence of the two automata. The problem is that reducing the number of colours might blow up the size of the representation of the acceptance condition.

\paragraph*{$\NP$-upper bound}
As mentioned before, to prove the $\NP$-upper bound for \pbColorMinAut{} we cannot just guess a "Muller acceptance condition" and check equivalence of automata, as the size of the representation of the acceptance condition could blow up.
However, when the condition is given "colour-explicitly@@Muller" or by a "Zielonka tree" or "Zielonka DAG", we can circumvent this problem by not building explicitly the representation of the acceptance condition. Instead we simply guess a colouring of the edges and check that it is compatible with the "ACD-DAG".

By contrast, in the case of Emerson-Lei conditions, this method fails. In fact, it is easy to reduce from \textsc{\small{UNSAT}} to the problem of whether a one-state Emerson-Lei automaton can be recoloured with zero colours, showing that the problem is $\coNP$-hard. 

\begin{lemma}
	\label{lem-compatibility-acdDAG}
	Let $\A = (Q,q_\init,\Sigma, \DD, \GG, \colAut, \MullerC{\F}{\GG})$ be a "deterministic" "Muller automaton" with its acceptance condition given "colour-explicitly@@Muller", by a "Zielonka tree" or by a "Zielonka DAG".
	Let $\colAut' : \DD \to \GG'$ (resp.~$\colAut' : \DD \to \pow{\GG'}$) be a colouring of its transitions.
	We can check in polynomial time that there exists a "Muller condition" $\F' \subseteq \powplus{\GG'}$ such that $(\colAut, \MullerC{\F}{\GG})$ and $(\colAut', \MullerC{\F'}{\GG'})$ are "equivalent over" $\A$.  
\end{lemma}

\begin{proof}
	First, we claim that such a condition $\F'$ exists if and only if there is no pair of words $w_+ \in \L(A)$ and $w_- \notin \L(\A)$ such that the sets of colours produced infinitely often under $\colAut'$ by their "runs" are equal.
	Indeed, if such words exist, it is clear that no $\F'$ over $\GG'$ can be consistent with the language "recognised" by $\A$. Conversely, if there are no such words, we can just take $\F'$ as the family of sets of colours seen infinitely often by "accepted words".
	
	We also note that, by definition of the "Zielonka DAG", if $w_+ \in \L(A)$ and $w_- \notin \L(\A)$ then there are nodes $n_+, n_-$ of $\zielonkaDAG{\F}$ such that $n_+$ is "round" and $n_-$ is "square" and the set of colours seen infinitely often in the run of $w_+$ (resp.~$w_-$) in $\A$ is included in the label of $n_+$ (resp.~$n_-$) but not in the ones of its children.	
	
	Hence we simply need to check the existence of those two words and nodes.
	Towards this, we construct a "Streett automaton" $\S_=$ over $\SS\times \SS$ "accepting" the pairs of words $(w,w')$ whose runs produce the same sets of colours infinitely often under $\colAut'$.
	Also, for each pair of nodes $n_+, n_-$ of $\zielonkaDAG{\F}$ such that $n_+$ is "round" and $n_-$ is "square", we construct a "Streett automaton" $\S_{n_+,n_-}$ over $\SS\times \SS$ "accepting" the pairs of words $(w,w')$ such that $w$ is accepted according to the node $n_+$ and $w'$ is rejected according to $n_-$.
	Finally, we can conclude by checking the non-emptiness of the intersection of $\S_=$ with each of these automata; this can be done in polynomial time, as a Streett automaton for the intersection can be build in polynomial time~\cite{Boker18WhyTypes}, as well as checking non-emptiness~\cite{EmersonLei1987ComplexityEmptinessAut}.
	We show how to build these automata in polynomial time when the "Muller condition" is given as a "Zielonka DAG"; the other two representation can be converted to this one in polynomial time.
	
	All these automata have a similar structure; they follow simultaneously the runs of the words $w,w'$ in $\A$, producing their outputs either under $\col$ or under $\col'$.
	
	Formally, the set of states is $Q\times Q$, with $(q_\init, q_\init)$ as initial state, and has transitions
	\[ (q,q') \re{(a,a'): (c,c')} (p,p') \quad \tif \quad q\re{a:c}p \; \tand \; q'\re{a':c'} p' \; \tin  \A. \]
	The output colours of $\S_=$ are $\GG'\times \GG'$, and those of $\S_{n_+,n_-}$ are $\GG\times \GG$. The colouring is given by $\col'$ and $\col$, respectively, as expected.
	
	The "Streett condition" of $\S_=$ expresses that the colours seen infinitely often in the first component are the same as the ones seen in the second component. Formally, it can be given by the set of "Rabin pairs"
	\[\R_= = \{\big( (x,-),(-,x) \big) \mid x\in \GG' \} \; \cup \; \{\big( (-,x),(x,-) \big) \mid x\in \GG' \},\]
	where $(x,-) = \{(x,c) \mid c\in \GG'\}$, and $(-,x)$ is defined symmetrically.
	
	
	The "Streett condition" of $\S_{n_+,n_-}$ expresses that the set of colours seen infinitely often by the run on the first component is included in $\nu(n_+)$ but not in $\nu(m)$ for any child $m$ of $n_+$, and similarly for the second component and $n_-$.
	Formally, the set of "Rabin pairs" is given by $\R_{n_+} \cup \R_{n_-}$ defined as:
	\[\R_{n_+} = \{ \big((\GG\setminus \nu(n_+)) \times \GG \, , \,  \emptyset\big)\} \cup\{\big(\GG\times \GG \, , \, (\GG\setminus \nu(m)) \times \GG\big) \mid m \text{ a child of } n_+\}.\]
	The set $\R_{n_-}$ is defined analogously in the second component for $n_-$.
\end{proof}

We conclude that \pbColorMinAut{} and \pbMultiColourMinAut{} are in $\NP$: we just need to guess the colouring $\col'$ with $k$ colours and check in polynomial time that a compatible "Muller condition" $\F'$ exists using the previous lemma.

\paragraph*{$\NP$-hardness}
We prove the $\NP$-hardness of the problems \pbColorMinAut{} and \pbMultiColourMinAut{} at the same time, for the representations "colour-explicit@@Muller" and "Zielonka tree". The result for the other representations follows then from Proposition~\ref{prop-size:explicit-to-DAG} and Theorem~\ref{th-comp:compt-ACD-poly-ZT}.
In Appendix~\ref{sec-app:alternative-reduction} we give an alternative $\NP$-hardness reduction (for the case of single-coloured transitions), using an automaton with only $2$ states.\\

We reduce from the problem \pbChromNum{}, defined as follows.
\AP An ""(undirected) graph"" is a pair $G = (V,E)$ consisting of a set of vertices $V$ and a set of edges $E\subseteq {V \choose 2}$ (that is, edges are subsets of size exactly two, in particular, no self loops are allowed).
\AP A ""$k$-colouring@@graph"" of an "undirected graph" $G=(V,E)$ is a mapping $\graphcolouring:V\rightarrow \{1,\dots,k\}$ such that $\graphcolouring(v)=\graphcolouring(v')\Rightarrow \{v,v'\}\notin E$ for every pair of nodes $v, v' \in V$. 
\AP The problem \intro*\pbChromNum{} consists in, given a "graph" $G$ (that can be assumed connected) and a positive integer $k$, decide whether $G$ admits a "$k$-colouring@@graph". We write \intro*\pbThreeCol{} for this problem with fixed $k=3$. Both problems are well-known to be $\NPc$~\cite{Karp72Reducibility,Stockmeyer73Graph3colour}, and they remain $\NPc$ on graphs of degree at most $4$~\cite{GJS76SimplifiedNP}.

Let $G = (V, E)$ be a connected "graph". 
We select an arbitrary vertex $v_{\init} \in V$.
\AP A ""pseudo-path"" in $G$ is a (finite or infinite) sequence $v_0 e_0 v_1 e_1 \cdots\in (V\cup E)^\infty$ such that $v_i, v_{i+1} \in e_i$ for all $i$. Note that we allow $v_i$ and $v_{i+1}$ to be equal, hence the term "pseudo-path"; that is, we allow a "pseudo-path" to step on an edge without going through it, and come back to the previous vertex. 
\AP A "pseudo-path" is ""initial"" if $v_0 = v_{\init}$. We say that such a "pseudo-path" ""stabilises around"" $v$ if it is infinite and there exists $i$ such that for all $j>i$, $v_j = v$, i.e., the "pseudo-path" eventually stays on the same vertex and just steps on the adjacent edges.
\AP We write $\intro*\Stab(v)$ for the set of "initial" "pseudo-paths" "stabilising around" $v$.
\AP For $v\in V$, we write $\intro*\adj(v)$ for the set of edges $\set{e \in E \mid v\in e}$.

We define the "automaton" $\intro*\autChromG$ as follows:
\begin{itemize}\setlength\itemsep{0.5mm}
	\item $Q =  V \cup E \cup \{q_\init\}$, where $q_\init$ is a fresh element, which is the initial state,
	\item $\SS = \GG = V \cup E$,
	\item $\Delta = \{(q_\init,v_\init,v_\init)\} \cup \set{(v, e,  e) \in V \times E^2 \mid v \in e} \cup \set{(e, v,  v) \in E \times V^2 \mid v \in e}$,
	\item the colour of each transition is the letter it reads, $\colAut(q\re{x}q') = x$,
	\item The "acceptance condition" is the "Muller language associated to"
	\[ \F = \{ C\subseteq \{v\}\cup \adj(v) \mid v\in V \}. \]	
\end{itemize}

This automaton is "deterministic" and "recognises" the language $\bigcup_{v \in V} \Stab(v)$. Note that $\autChromG$ is not "complete": it only reads "initial" "pseudo-paths" of $G$.

The representation of the automaton $\autChromG$ is polynomial in $|V|+|E|$. 
The family $\F$ has a "Zielonka tree" and a "Zielonka DAG" of polynomial size (more precisely, they both have $|V|+1$ nodes), so by Theorem~\ref{th-comp:compt-ACD-poly-ZT}, we can provide in polynomial time $\acd{\autChromG}$.
Moreover, if $G$ has bounded degree (we can assume that it has outdegree $4$), the "colour-explicit" representation of $\F$ is also of polynomial size.

\begin{lemma}
	\label{lem-3-colouring-multicolour}
	A  connected "graph" $G$ admits a "$3$-colouring@@graph" if and only if $\autChromG$ is "$3$-colour type@@TS" if and only if $\autChromG$ is "$3$-multiple-colour type@@TS".
\end{lemma}
\begin{proof}
	We prove that if $G$ admits a "$3$-colouring@@graph", then $\autChromG$ is "$3$-colour type@@TS" (and therefore "$3$-multiple-colour type@@TS"), and that if $\autChromG$ is "$3$-multiple-colour type@@TS" then $G$ admits a "$3$-colouring@@graph".
	
	Let $\graphcolouring \colon V \to \{1,2,3\}$ be a "$3$-colouring@@graph" of $G$. We let $\GG' = \{1,2,3\}$ and define the colouring $\col' : \Delta \to \GG'$ with $\col'(e \xrightarrow{v} v) = \graphcolouring(v)$ and $\col'(v \xrightarrow{e} e) = \graphcolouring(v)$ for all $e \in E, v \in e$ (the colouring from $v_\init$ is irrelevant).
	We define the family $\F = \set{\set{1}, \set{2}, \set{3}}$, and let $\A'$ be the "automaton" obtained by setting the "acceptance condition" of $\autChromG$ to be $\MullerC{\F}{\GG}$ (that is, we accept the "runs" eventually visiting only one colour). Let us prove that $\L(\autChromG) = \L(\A')$.
	Let $w \in \L(\autChromG)$, $w$ is an "initial" "pseudo-path" and there must exist $v$ such that $w$ "stabilises around" $v$. Let $i = \graphcolouring(v)$, ultimately $w$ only visits $v \cup \adj(v)$ and thus it only produces colour $i$, so $w \in \L(\A')$.
	Let $w \in \L(\A')$. Again, $w$ is an "initial" "pseudo-path" and there must exist $i$ such that $w$ ultimately only visits $\bigcup_{v\in \graphcolouring^{-1}(i)} \set{v} \cup \adj(v)$. 
	\AP Moreover, as $\graphcolouring$ is a "$3$-colouring@@graph" of $G$, $\graphcolouring^{-1}(i)$ is an ""independent set"" (that is, no two vertices are connected by an edge), hence $w$ cannot visit infinitely often two distinct vertices from this set without visiting infinitely often an intermediate vertex of a different colour.
	As a consequence, $w$ must "stabilise around" some $v \in \graphcolouring^{-1}(i)$, thus $w \in \L(\autChromG)$. We have shown that $\L(\autChromG) = \L(\A')$.

	For the other direction of the reduction, suppose we have a multi-colouring $\col' : \Delta \to \pow{\GG'}$, with $\GG' = \{1,2,3\}$ and a family $\F\subseteq \powplus{\GG'}$ yielding an "equivalent acceptance condition over" $\autChromG$.
	Then we define a colouring $\graphcolouring : V \to \set{1,2,3}$ as follows. First, we define the function $\intro*\env: V \to \pow{\GG}$ by $\env(v) = \bigcup_{e \in E, v \in e} \col'(e \xrightarrow{v} v) \cup \col'(v \xrightarrow{e} e)$.
	
	\begin{claim}
		For all $\set{u,v} \in E$, $\env(u) \nsubseteq \env(v)$.
	\end{claim}
	
	\begin{claimproof}
		Suppose by contradiction that we have $\env(u) \subseteq \env(v)$ for two neighbours $u,v$.
		Then we can construct a run cycling through $u, v$ and all their neighbouring edges, and seeing infinitely often the set of colours $\env(u) \cup \env(v) = \env(v)$. This run must be rejected. This is a contradiction as a run cycling through only $v$ and neighbouring edges is accepted but sees the same set of colours infinitely often.
	\end{claimproof}
	
	We can then split the set $\GG = \pow{\set{1,2,3}}$ into three chains $\GG = \set{\emptyset, \set{1}, \set{1,2}, \set{1,2,3}} \sqcup \set{\set{2}, \set{2,3}} \sqcup \set{\set{3}, \set{1,3}}$. By the previous claim, the preimage by $\env$ of each of those three sets is an "independent set". We thus obtain a partition of $G$ into three "independent sets", yielding a "$3$-colouring@@graph".
\end{proof}

\subsubsection{Minimisation of Rabin pairs on top of a Rabin automaton}
Similarly, we consider the problem of minimising the number of "Rabin pairs" over a fixed "Rabin" "automaton".

\AP We say that a "deterministic" "Muller" "automaton" $\A$ is ""$k$-Rabin-pair type@@TS"" if we can "relabel" it with an "equivalent@@cond" "Rabin condition" using at most $k$ "Rabin pairs".

\AP
\begin{center}
	\fbox{\begin{tabular}{rl}
			{\textbf{Problem:}} & \intro*\pbRabinPairMinAut{} \\[1mm]
			{\textbf{Input:}} & A "deterministic" "Rabin" "automaton" $\A$ and a positive integer $k$.\\[1mm]   
			{\textbf{Question:}} & Is $\A$ "$k$-Rabin-pair type@@TS"?  
	\end{tabular}} 
\end{center}   

As before, we can consider different representations of the "acceptance condition" $\RabinC{\R}{\GG}$ of the "automaton": using "Rabin pairs", with a "colour-explicit@@Muller" "Muller condition", or by providing the "Zielonka tree", the "Zielonka DAG" or the "ACD".

\begin{theorem}[$\NP$-completeness of minimisation of Rabin pairs for Rabin automata]\label{th-min:RabinPairsMin-Aut-NPHard}
	The problem \pbRabinPairMinAut{} is $\NP$-complete for all the previous "representations" of the "acceptance condition".
\end{theorem}

\begin{lemma}
	The problem \pbRabinPairMinAut{} is in $\NP$.
\end{lemma}
\begin{proof}
	It suffices to guess a family of $k$ "Rabin pairs" over the set of "colours" $\DD$ and check if the obtained automaton "recognises" the same language as before. For the representation as "Rabin pairs", this can be done in polynomial time as the equivalence of "deterministic" "Rabin" automata can be checked in polynomial time~\cite{ClarkeDK93Unified}.
	Propositions~\ref{prop-size:explicit-to-DAG} and~\ref{prop-size:from-ZDAG-to-RabinPairs}  imply that this is also possible for the other representations.
\end{proof}

We now show that the problem \pbRabinPairMinAut{} is $\NP$-hard. We reduce from the problem \pbChromNum{}, using the same construction as in the previous subsection.
Consider the automaton $\autChromG$ defined in the proof of Theorem~\ref{th-min:colorMin-Aut-NPhard}.
It turns out that the "acceptance condition" of this automaton is a "Rabin language", indeed, we can define it as $\RabinC{\R}{V \cup E}$ by letting:
\[\R = \set{\big(V \cup E, \,(V \cup E) \setminus (\set{v} \cup \adj(v))\big) \mid v \in V}.\]
	
As before, the "Zielonka tree" of $\RabinC{\R}{V \cup E}$ has size at most $|V|+1$, and for graphs of outdegree at most $4$, a family representing $\RabinC{\R}{V \cup E}$ "colour-explicitly" is of polynomial size.

	\begin{lemma}
		A  connected "graph" $G$ admits a "$k$-colouring@@graph" if and only if $\autChromG$ is "$k$-Rabin-pair type@@TS".
	\end{lemma}
	\begin{proof}
	Let $\graphcolouring\colon V \to \{1,\dots,k\}$ be a "$k$-colouring@@graph" of $G$. For all $i \in \set{1, \ldots, k}$ we define the "Rabin pair" $R_i = (\greenPair_i,\redPair_i)$ with:
	\[\greenPair_i = V \cup E \; , \quad \redPair_i = (V \cup E)  \setminus \bigcup_{v\in \inv{\graphcolouring}(i)} \set{v} \cup \adj(v).\]
	We set $\R' = \set{R_i \mid i \in \set{1,\ldots,k}}$.
	Let $\A'$ be the "automaton" obtained by setting the "acceptance condition" of $\autChromG$ to be $\RabinC{\R'}{V\cup E}$. Let us prove that $\L(\A') = \L(\autChromG)$.
	Let $w \in \L(\autChromG)$, $w$ is an "initial" "pseudo-path" and there must exist $v$ such that $w$ "stabilises around" $v$. Let $i = \graphcolouring(v)$, ultimately $w$ only visits $v \cup \adj(v)$ and thus it satisfies $R_i$. As a result, $w \in \L(\A')$.
	Let $w \in \L(\A')$. Again, $w$ is an "initial" "pseudo-path" and there must exist $i$ such that $w$ ultimately only visits $\bigcup_{v\in \graphcolouring^{-1}(i)} \set{v} \cup \adj(v)$. 
	Moreover, as $\graphcolouring$ is a "$k$-colouring@@graph" of $G$, $\graphcolouring^{-1}(i)$ is an "independent set", hence $w$ cannot visit infinitely often two distinct vertices from this set without visiting infinitely often an intermediate vertex of a different colour.
	As a consequence, $w$ must "stabilise around" some $v \in \graphcolouring^{-1}(i)$, thus $w \in \L(\autChromG)$. We have shown that $\L(\autChromG) = \L(\A')$.
	
	For the other direction of the reduction, suppose we have a family of $k$ "Rabin pairs" $\R' = (\greenPair_i, \redPair_i)_{1\leq i \leq k}$ (over some set of colours $\GG'$) such that the automaton $\A'$ obtained by "relabelling" $\autChromG$ with the "acceptance condition" $\Rabin{\R'}$ recognises $\L(\autChromG)$. For all $v \in V$, as we assumed $G$ to be connected, there is a finite path $\pi_v = v_{\init} e_1 v_1 e_2 v_2 \cdots e_\ell v$. We also define $\rho_v$ as a finite "pseudo-path" $e_1 v e_2 \cdots v e_r v$ such that $\set{e_1, \ldots, e_r} = \adj(v)$.
	The word $w_v = \pi_v \rho_v^\oo$ is in $\L(\autChromG)$, thus there exists $c_v \in \set{1,\ldots, k}$ such that the run of $w_v$ in $\A'$ satisfies $(\greenPair_{c_v}, \redPair_{c_v})$. The set of colours it sees infinitely often is $\set{v} \cup \adj(v)$, thus $\redPair_{c_v} \cap (\set{v} \cup \adj(v)) = \emptyset$ and $\greenPair_{c_v} \cap (\set{v} \cup \adj(v)) \neq \emptyset$. 
	We thus define a function $\graphcolouring : V \to \set{1, \ldots, k}$ as $v \mapsto c_v$. 
	
	It remains to prove that $\graphcolouring$ is a valid "$k$-colouring@@graph" of $G$. 
	Suppose  there exist two neighbours $u,v$ such that $\graphcolouring(u)=\graphcolouring(v)=i$. 
	Then the runs over $w_u$ and $w_v$ both satisfy $(\greenPair_i, \redPair_i)$. 
	As a result, the set $\set{v} \cup \adj(v) \cup \set{u} \cup \adj(u)$ also satisfies $(\greenPair_i, \redPair_i)$. 
	We define $\rho = e_1 u e_2 \cdots u e_r u$ and $\rho' = e'_1 v e'_2 \cdots v e'_r v$ with $e_1 = e'_1 = \set{u,v}$.
	We can then observe that the word $\pi_u (\rho \rho')^\oo$ has an accepting run in $\A'$, as the colours it sees infinitely often are $\set{v} \cup \adj(v) \cup \set{u} \cup \adj(u)$. However, this word is not accepted by $\autChromG$, a contradiction.
	As a result, $\graphcolouring$ is a valid "$k$-colouring@@graph" of $G$. 
\end{proof}

This concludes our reduction, showing that the minimisation of "Rabin pairs" with respect to a given automaton is \NP-hard. 

%
%
%

\section{Size of different representations of acceptance conditions}
\label{sec:size}
We start Section~\ref{subsec:worst-case} by  analysing the size of the "Zielonka tree" and "ACD" in the worst case.
Using Proposition~\ref{prop-prelim:minimal-det-parity}, stating that minimal "(history-)deterministic" "parity" "automata" can be derived from the "Zielonka tree", we can directly translate the lower bounds for the size of the "Zielonka tree" into lower bounds for  "(history-)deterministic" "parity" "automata". We recover in this way some results from Löding~\cite{Loding1999Optimal} and generalise them to "history-deterministic" automata.
Then, we compare in Section~\ref{subsec:size-comparison} the size of different "representations@@Muller" of "Muller languages" and study the translations between them, with special focus on the "Zielonka tree" and the "Zielonka DAG", proving the claims from in Figure~\ref{fig-prel:representationMuller}.

\subsection{Worst case analysis of the Zielonka tree}\label{subsec:worst-case}

We study the size of the "Zielonka tree" in the worst case.
By Remark~\ref{rmk-acd:ZT-as-ACD}, the given bounds apply to the "ACD", as the "Zielonka tree" can be seen as the "ACD" of a "Muller" "automaton" with just one state.

\begin{proposition}[Size of the Zielonka tree: Worst case]\label{prop-size:worst-case-ZT}
	Let $\F\subseteq \powplus{\GG}$ be a family of subsets, and let $m=|\GG|$. It holds:
	\begin{itemize}\setlength\itemsep{1mm}
		\item $|\zielonkaTree{\F}| \leq 1+m+m(m-1)+\dots + m!$,
		\item $|\leaves(\zielonkaTree{\F})|\leq m!$, and
		\item the height of $\zielonkaTree{\F}$ is at most $m$.
	\end{itemize}
	These bounds are tight: for all $m\in \NN$, there is a family $\F_m\subseteq \powplus{\GG_m}$ over a set of $m$ colours such that the previous relations are equalities.
\end{proposition}
\begin{proof}
	We start by showing that the given bounds are tight. We suppose that $m$ is even (the construction is symmetric if $m$ is odd), and let $\GG_m=\{1,\dots,m\}$. Consider the family\footnotemark{} $\evenLetters{m}\subseteq\powplus{\GG_m}$ given by:
	\footnotetext{This family of subsets already appear in the worst-case study of parity automata recognising a "Muller language" in Mostowski's paper introducing the "parity condition"~\cite[p.161]{Mostowski1984RegularEF}.}
	\[ \intro*\evenLetters{m} = \{C\subseteq \GG_m \mid |C| \text{ is even}\}. \]
	
	First, we remark that the last inequality follows from the fact that the subsets $\GG_m, \GG_m\setminus\{1\}, \dots, \GG_m\setminus\{1,\dots, m-1\}$ form a branch of the "Zielonka tree".
	Let $n$ be a node of the "Zielonka tree" of $\evenLetters{m}$, and let $X_n = \nu(n)$ be its label. Then $n$ has a child for each subset of $X_n$ of size $|X_n|-1$. A simple induction gives that the level at depth $k$ of the "Zielonka tree" has $m(m-1)\cdots (m-(k-1))$ nodes. This establishes the two first equalities of the statement.

	We prove now the upper bounds. The last item follows from the fact that the label of a node is a set of size strictly smaller than the label of its parent. 
	Let $\F\subseteq \powplus{\GG}$, and $m=|\GG|$. We show by recurrence that the "Zielonka tree" $\zielonkaTree{\F}$ is not bigger than that of $\evenLetters{m}$. We remark that for all $X\subseteq \GG_m$, there is a node $n_X$ in the "Zielonka tree" of $\evenLetters{m}$ labelled $X$, and that the subtree rooted at $n_X$ is isomorphic to the "Zielonka tree" of $\evenLetters{|X|}$.
	Let $n_0$ be the root of $\zielonkaTree{\F}$, and let $n_1,\dots, n_k$ be its children. Then, we can find in the "Zielonka tree" of $\evenLetters{m}$ $k$ incomparable nodes having as labels $\nu(n_1),\dots, \nu(n_k)$. By induction hypothesis, the subtree rooted at each of these nodes is not smaller than the subtrees rooted at $n_1,\dots, n_k$.
	This shows the two first items, ending the proof.
\end{proof}

We recover results analogous to those of L\"oding~\cite{Loding1999Optimal}, and strengthen them as they apply to "history-deterministic" automata. These directly follow combining the previous proposition with Proposition~\ref{prop-prelim:minimal-det-parity}. 

\begin{corollary}{}
	For every "Muller language" $L\subseteq\GG^\oo$ there exists a "deterministic" "parity" "automaton" "recognising" $L$ of size at most~$|\GG|!$.
	This bound is tight: for all $n$, a minimal "history-deterministic" "parity" "automaton" "recognising" the "Muller language" associated to $\evenLetters{n}$ has $n!$ states.
\end{corollary}

Performing a slightly more careful analysis and using the characterisation of minimal "history-deterministic" "Rabin" automata by Casares, Colcombet and Lehtinen~\cite{CCL22SizeGFG}, we can obtain similar tight bounds for these automata. We refer to~\cite[Corollary~II.93]{Casares23Thesis} for details.


\subsection{Comparing the sizes of various acceptance conditions}\label{subsec:size-comparison}

\subparagraph{Colour-explicit vs Zielonka trees.} First, we remark that a "colour-explicit@@Muller" representation of a family $\F$ can be arbitrary larger than a representation of $\zielonkaTree{\F}$.

\begin{proposition}\label{prop-size:expl-larger-ZT}
For all $n\in \NN$, there is a family of subsets $\F_n\subseteq \powplus{\GG_n}$ over $\GG_n = \{1,\dots,n\}$ such that $|\F_n| = 2^n-1$ and $|\zielonkaTree{\F}|=1$.
\end{proposition}
\begin{proof}
It suffices to take $\F=\powplus{\GG_n}$.
\end{proof}

Even if the family $\F$ is represented "explicitly@@Muller" as a list of subsets, we cannot compute its "Zielonka tree" in polynomial time, as $\zielonkaTree{\F}$ can be exponentially larger than $|\F|$.


\begin{proposition}
	For all $n\in \NN$, there is a family of subsets $\F_n\subseteq \powplus{\GG_{n}}$ over $\GG_{n} = \{1,\dots, 2n\}$ such that:
	\[ |\zielonkaTree{\F_n}| \geq 3^{|\F_n|}. \]
	More precisely, there is a family $\F_n$ such that $|\zielonkaTree{\F_n}|=4\cdot 3^{n-1} -1$ and $|\F_n|=n$.
\end{proposition}

\begin{proof}
	Consider the family $\F_n = \set{\set{1,2}, \set{1,2,3,4}, \ldots, \set{1,2,\ldots,2n-1, 2n}}$.	
	The Zielonka tree for $\F_n$ has a "round" root labelled $\GG_n$. It has $2n$ children, which are "square" nodes labelled by $\set{1,2,\ldots,2n-1, 2n} \setminus \set{i}$ for $1 \leq i \leq 2n$. For $i \geq 3$, each of those nodes has a single child labelled $\set{1,2,\ldots,2j-1, 2j}$ with $j = \lfloor \frac{i+1}{2}\rfloor$. The subtree rooted at this child is the "Zielonka tree" of $\F_{j}$.
	
	As a result we get 
	\[\size{\zielonkaTree{\F_n}} = 
	1 + 2n + \sum_{j=1}^{n-1} 2\size{\zielonkaTree{\F_j}} =
	2+ 1+2(n-1) + 2\size{\zielonkaTree{\F_j}} + \sum_{j=1}^{n-2} 2\size{\zielonkaTree{\F_j}} = 
	2+ 3 \size{\zielonkaTree{\F_{n-1}}}.\]
	Hence $\size{\zielonkaTree{\F_n}} +1 = 3(\zielonkaTree{\F_{n-1}}+1)$. As $\size{\zielonkaTree{\F_1}}= 3$, we get $\size{\zielonkaTree{\F_n}} = 4\cdot 3^{n-1} -1$
\end{proof}

The previous bound is almost optimal, as shown next.

\begin{proposition}
	For every family of subsets $\F\subseteq \powplus{\GG}$ over an alphabet $\GG$ (for $|\F|$ and $|\GG|$ large enough), we have:
	\[ |\zielonkaTree{\F}| \leq \alpha^{\size{\F}} \beta^{\size{\GG}},\]
	for all $\alpha > 3^{1/3} \simeq 1.44$ and $\beta > 4$.
\end{proposition}
\begin{proof}	
	We will over-approximate the number of branches of the "Zielonka tree" by the number of \emph{well-shaped chains} of subsets of $\GG$. A chain $A_0 \supset \cdots \supset A_k$ is well-shaped if for all $i$, if $A_i \notin \F$ (1) $A_{i-1}$ and $A_{i+1}$ are in $\F$ (if they exist) and (2) $\size{A_{i-1}} = \size{A_i} +1$ (if $A_{i-1}$ exists) or $A_{i} =\GG$.
	We note that this definition is asymmetric in the following sense: a well-shape sense may contain several "accepting@@Muller" subsets in a row, but no two rejecting consecutive ones. Also, when changing from accepting to rejecting only one element is added (but this is not necessarily true for the other direction).
	
	Let $R_0 \supset S_1 \supset R_1 \supset \cdots \supset S_k \supset R_k$ be the sequence of labels of a branch, with $(R_i)$ the labels of "round" nodes and $(S_i)$ the ones of "square" nodes (we assume the first and last sets are in $\F$, other cases are similar).
	For each $i$ such that $|R_{i-1} \setminus S_i| >1$, we pick an element $x \in R_{i-1} \setminus S_i$ and add the set $S_{i} \cup \set{x}$ in the chain between $R_{i-1}$ and $S_i$. Note that by definition of the tree $S_i$ is a maximal subset of  $R_{i-1}$ not in $\F$ and thus $S_i \cup \set{x}$ is in $\F$.
	Clearly we can recover the initial branch from the resulting chain. 
	Furthermore the resulting chain is well-shaped. Thus the number of branches is bounded by the number of well-shaped chains.
	
	Let us count the number of well-shaped chains. They can all be obtained by picking a chain $A_0 \supset \cdots \supset A_k$ of sets in $\F$ and, for each $i$, either adding $A_i \setminus \set{x}$ between $A_{i-1}$ and $A_{i}$ for some $x \in A_{i-1} \setminus A_{i}$ or not adding anything.
	The number of well-shaped chains that can be obtained this way from a fixed chain $A_0 \supset \cdots \supset A_k$ is at most $\prod_{i=1}^{k} (\size{A_{i-1}} - \size{A_{i}} + 1)$.
	An easy induction on the size of $A_0$ shows that this number is bounded by $2^{\size{\GG}}$.
	It is a classical exercise that a partially ordered set of size $m$ has at most $3^{(m+1)/3}$ chains of maximal length (it suffices to see the poset as a DAG and then proceed by induction on its height).
	Maximal chains of $\F$ have at most $|\Gamma|$ elements, hence $\F$ contains at most $3^{(\size{\F}+1)/3} 2^{\size{\GG}}$ chains.

	We thus obtain that $\zielonkaTree{\F}$ contains at most $3^{(\size{\F}+1)/3} 2^{2\size{\GG}} = 3^{(\size{\F}+1)/3} 4^{\size{\GG}}$ branches, and thus at most $3^{(\size{\F}+1)/3} 4^{\size{\GG}} \size{\GG}$ nodes, as the tree is of height at most $\size{\GG}$.
\end{proof}

The bounds on $\alpha$ and $\beta$ in the previous proposition could be improved with a finer analysis of the proof.

\subparagraph{Colour-explicit vs Zielonka DAGs.}
Hunter and Dawar showed that we can compute the "Zielonka DAG" of a family $\F$ in polynomial time if $\F$ is given as a list of subsets~\cite[Theorem~3.17]{HD08ComplexityMuller}. 
\begin{proposition}[{\cite[Theorem~3.17]{HD08ComplexityMuller}}]\label{prop-size:explicit-to-DAG}
Given a family of subsets $\F\subseteq \powplus{\GG}$, we can compute the "Zielonka DAG" of $\F$ in polynomial time in $|\F|+|\GG|$.
In particular, $\zielonkaDAG{\F}$ has polynomial size in $|\F|+|\GG|$.
\end{proposition}

However the reverse transformation cannot be done in polynomial time as Proposition~\ref{prop-size:expl-larger-ZT} also applies to the "Zielonka DAG".

\subparagraph{Zielonka trees vs Zielonka DAGs.} It is clear that, given a "Zielonka tree" $\zielonkaTree{\F}$, we can compute the corresponding "Zielonka DAG" $\zielonkaDAG{\F}$ in polynomial time. The converse is not possible.
We note that this statement follows from complexity considerations: solving "Muller" games with the winning condition represented as a "Zielonka DAG" is $\PSPACE$-complete~\cite{HD08ComplexityMuller}, while solving those games with the condition represented as a "Zielonka tree" is equivalent to solving "parity" games~\cite{DJW1997memory}, which can be done in quasi-polynomial time~\cite{CJKLS22}.
However, to the best of the author's knowledge, no explicit family witnessing an exponential gap between the two representations appears in the literature.

\begin{proposition}\label{prop-size:size-ZT-vs-ZDAG}
For all $n\in \NN$, there is a family of subsets $\F_n\subseteq \powplus{\GG_n}$ over $\GG_n = \{1,\dots,n\}$ such that:
\begin{itemize}
	\item the size of the "Zielonka DAG" of $\F_n$ is at most $2n$,
	\item the size of the "Zielonka tree" of $\F_n$ is at least $2^{\lfloor n/2\rfloor}$.
\end{itemize} 
\end{proposition}
\begin{proof}
Consider the family defined as follows:
\[ \intro*\smallDAGCondition{n} = \{ C = \{c_1<c_2<\dots<c_k\}\subseteq \GG_n \mid c_1 \text{ is odd and } c_2=c_1+1\}. \]

Equivalently, we can describe this family as \[\bigcup_{\substack{i=1, \\ i\text{ odd}}}^n X_i,  \text{ where } X_i = \{C\subseteq \GG_n \mid i\in C \tand i+1\in C \tand c>i \text{ for all } c\in C\}.\]

\begin{figure}[ht]
	\hspace{0mm}
	\begin{minipage}[c]{0.37\textwidth} 
		\includegraphics[scale=0.95]{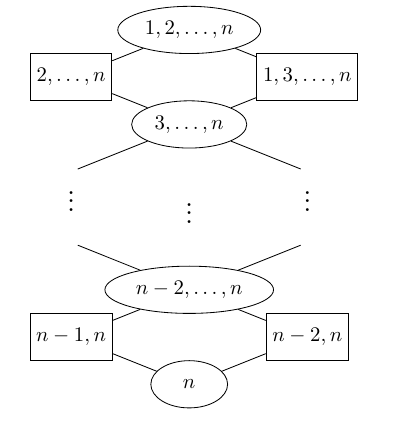}
	\end{minipage} 
	\hspace{8mm}
	\begin{minipage}[c]{0.55\textwidth} 
		\includegraphics[scale=0.95]{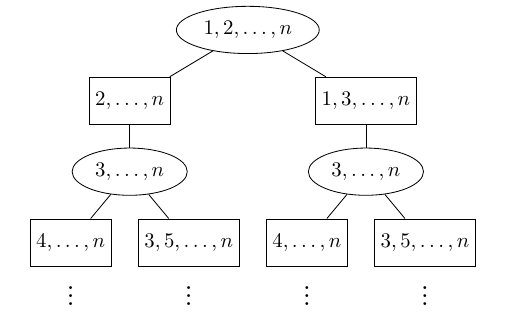}
	\end{minipage} 
	\caption{On the left, the "Zielonka DAG" of the condition $\smallDAGCondition{n}$ (for $n$ odd), of size $\O(n)$. On the right, its "Zielonka tree", of exponential size.}
	\label{fig-size:ZT-vs-ZDAG}
\end{figure}	

We show the "Zielonka DAG" and the "Zielonka tree" of $\smallDAGCondition{n}$ (for $n$ odd) in Figure~\ref{fig-size:ZT-vs-ZDAG}.
We observe that the "Zielonka DAG" has height $n$; even levels consist in a single node, and odd levels have two nodes. Therefore, its size is $\lceil n/2 \rceil + n$. On the other hand, the "Zielonka tree" (with height also $n$), has $2^{\lfloor k/2\rfloor}$ nodes at the level of depth $k$.
\end{proof}


\subparagraph{Rabin vs Zielonka trees and Zielonka DAGs.}
If the "Muller language" associated to a family~$\F$ is a "Rabin language", then we can compute a family of "Rabin pairs" $\R$ such that $\Rabin{\R} = \Muller{\F}$ in polynomial time. The converse is not possible, we cannot compute the "Zielonka DAG" in polynomial time, since it can be of exponential size in the number of "Rabin pairs".

\begin{proposition}\label{prop-size:from-ZDAG-to-RabinPairs}
Let $\F\subseteq \powplus{\GG}$ be a family of subsets, and assume that $\MullerC{\F}{\GG}$ is a "Rabin language" (that is, it admits a representation with "Rabin pairs").
Then, given the "Zielonka DAG" $\zielonkaDAG{\F}$ we can compute in polynomial time a family of "Rabin pairs" $\R$ over $\GG$ such that $\RabinC{\R}{\GG} = \MullerC{\F}{\GG}$. 
\end{proposition}
\begin{proof}	
Let $N = \roundnodes \disjUnion \squarenodes$ be the nodes of the "Zielonka DAG", partitioned into "round" and "square" nodes.
By Proposition~6.2 from~\cite{CCFL24FromMtoP}, all "round nodes" of $\zielonkaDAG{\F}$  have at most one child.	
We define a "Rabin pair" for each "round" node of $\zielonkaDAG{\F}$, $\R = \{(\greenPair_n,\redPair_n)\}_{n\in\roundnodes}$, where $\greenPair_n$ and $\redPair_n$ are defined as follows: 
\begin{equation*}
	\begin{cases}
		\greenPair_n = \GG \setminus \nu(n),\\
		\redPair_n = \nu(n)\setminus \nu(n'), \text{ for } n'  \text{ the only child of } n\text{, if it exists.}\\ 
		\redPair_n=\nu(n) \text{ if $n$ has no children.}
	\end{cases}
\end{equation*} 

That is, the pair $(\greenPair_n,\redPair_n)$ "accepts@@RabinPair" the sets of colours $A\subseteq \GG$ that contain some of the colours that disappear in the child of $n$ and none of the colours appearing above $n$ in the "Zielonka DAG". We show that $\Rabin{\R}= \Muller{\F}$. Let $A$ be a set of colours. If $A\in \F$, let $n$ be a maximal node (for $\ancestor$) containing $A$. It is a "round node" and there is some colour $c\in A$ not appearing in the only child of $n$. Therefore, $c\in \greenPair_n$ and $A \cap \redPair_n=\emptyset$. Conversely, if $A\notin \F$, then for every "round node" $n$ with a child $n'$, either $A\subseteq \nu(n')$ (and therefore $A\cap \greenPair_n = \emptyset$) or $A \nsubseteq \nu(n)$ (and in that case $A\cap \redPair_n \neq \emptyset$).
\end{proof}

\begin{proposition}\label{prop-app-size:ZT-exponential-in-RabinPairs}
For all $m\in \NN$, there is a family $\R$ of $m$ "Rabin pairs" over a set of colours~$\GG$ of size $2m$, such that $|\zielonkaTree{\F_{\R}}| \geq m!$ and $|\zielonkaDAG{\F_{\R}}|\geq 2^{m}$, where $\F_{\R}\subseteq \powplus{\GG}$ is the (only) family such that $\Muller{\F_{\R}}=\Rabin{\R}$.
\end{proposition}
\begin{proof}
Let $\GG = \{g_1,r_1,g_2,r_2,\dots,g_m,r_m\}$ and define the "Rabin pairs" of $\R$ as $\greenPair_i = \{g_i\}$ and $\redPair_i = \{r_i\}$. We depict the "Zielonka tree" of the corresponding family of subsets in Figure~\ref{fig-app-size:Rabin-worst}.

\begin{figure}[ht]
	\centering
	\includegraphics[]{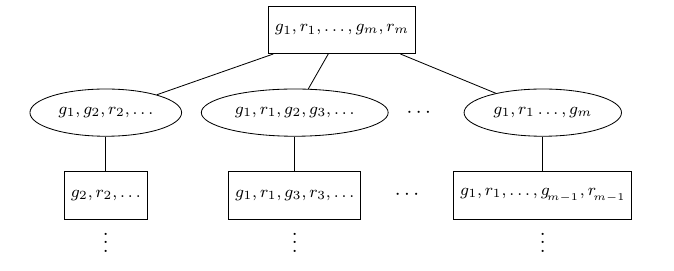}
	\caption{The "Zielonka tree" $\zielonkaTree{\F_{\R}}$ of the "Rabin language" from the proof of Proposition~\ref{prop-app-size:ZT-exponential-in-RabinPairs}.}
	\label{fig-app-size:Rabin-worst}
\end{figure}

The "Zielonka tree" $\zielonkaTree{\F_{\R}}$ satisfies that the levels at depth $k$ and $k+1$ have $m(m-1)\dots k$ nodes, which shows that $|\leaves(\zielonkaTree{\F_{\R}})| = m!$. 
For the bound on the size of the "Zielonka DAG", we observe that for each subset $X\subseteq \{1,\dots, m\}$ there is at least one subset appearing as the label of some nodes of the "Zielonka tree", namely, $\{g_i,r_i \mid i\in X\}$.
\end{proof}

\section{Conclusion}
\label{sec:conclusion}
In this work we obtained several positive results concerning the complexity of simplifying the acceptance condition of an $\omega$-"automaton".

Our first technical result is that the computation of the "ACD" (resp. "ACD-DAG") of a "Muller" automaton is not harder than the computation of the "Zielonka tree" (resp. "Zielonka DAG") of its "acceptance condition" (Theorems~\ref{th-comp:compt-ACD-poly-ZT} and~\ref{th-comp:compt-ACD-DAG-poly-ZDAG}). This provides support for the assertion that the optimal transformation into "parity" automata based on the "ACD" is applicable in practical scenarios, backing the experimental evidence provided by the implementations of the "ACD-transform"~\cite{CDMRS22Tacas}.

Furthermore, this result has several implications for our simplification purpose: 
\begin{itemize}
	\item We can "decide the typeness" of "Muller" "automata" in polynomial time (Corollary~\ref{cor-comp:decision-typ-poly}). 
	\item We can compute the "parity index" of a language "recognised" by a "deterministic" "Muller" "automaton" in polynomial time (Corollary~\ref{cor-comp:decision-parity-index-poly}).
\end{itemize}

In addition, we showed that we can minimise in polynomial time the "colours" and "Rabin pairs" necessary to represent a "Muller language". 
However, these problems become $\NP$-hard when taking into account the structure of a particular automaton using this "acceptance condition", even if the "ACD" of the automaton is provided as input.
Nevertheless, we believe that the methods for the minimisation of colours in the case of "Muller languages" could be combined with the structure of the ACD to obtain heuristics reducing the number of colours used by "Muller automata", which might lead to substantial (although not optimal) reductions in the number of colours.

In sum, our results help to clarify the potential of the "alternating cycle decomposition" and complete the picture of our understanding about the possibility of simplifying the "acceptance conditions" of $\oo$-automata.

\bibliographystyle{plainurl}
\bibliography{references}

\newpage
\appendix

\section{Generalised Horn formulas}
\label{sec:Gen-Horn}

"Horn formulas" are a popular fragment of propositional logic, as they enjoy some convenient complexity properties. It is well-known that the satisfiability problem for those formulas can be solved in linear time~\cite{Dowling84linearHorn}.

In this appendix, we study a succinct representation of "Horn formulas", called "Generalised Horn formula". They allow one to merge several Horn clauses with the same premises, e.g. $(x_1 \land x_2 \implies y_1)$ and $(x_1 \land x_2 \implies y_2)$, into a single clause $(x_1 \land x_2 \implies y_1 \land y_2)$. 
We can apply the classical linear-time algorithm for satisfiability on this generalised form, however, note that it is not linear in the size of the generalised formula, but in the size of the implicit Horn formula represented.

We will prove that we can minimise the number of clauses in a "GH formula" in polynomial-time, using our algorithm for minimising the number of pairs in a "Rabin condition" as a black box.

This result contrasts nicely with the \NP-completeness of minimising the number of clauses in a Horn formula~\cite{Boros1994complexity} (see also~\cite{Chang2006HornMin}).
On the other hand, minimising the number of literals in a "GH formula" remains \NP-complete, just like in the case of Horn formulas~\cite{Hammer1993optimal}. This can be showed by a slight adaptation of the reduction from~\cite{Chang2006HornMin} to "GH formulas".

Our technique for clause minimisation may thus be of interest for the study of Horn formulas. 

On the other hand, "generalised Horn formulas" are likely not a suitable representation for "acceptance conditions" on "automata", as they yield an $\NP$-complete emptiness problem (Proposition~\ref{prop-Horn:emptiness}). This is an interesting example of a family of "acceptance conditions" whose satisfiability problem is in \PTimeFull{} but which yields an \NP-complete emptiness problem on automata.\\

\begin{definition}
	\AP A ""Horn clause"" is a disjunction of literals with at most one \emph{positive} literal, that is, a literal with no negation.
	Equivalently, it is a Boolean formula of the form either $(x_1 \land \cdots \land x_n) \implies y$ or $(x_1 \land \cdots \land x_n) \implies \bot$.
	A ""Horn formula"" is a conjunction of "Horn clauses".
	
	\AP A ""generalised Horn clause"" (or "GH clause") is a Boolean formula of the form either $(x_1 \land \cdots \land x_n) \implies (y_1 \land \cdots \land y_m)$ or $(x_1 \land \cdots \land x_n) \implies \bot$ (in the latter case, the clause is called ""negative@@Horn"").
	A ""generalised Horn formula"" (or "GH formula") is a conjunction of "GH clauses". It is ""simple"" if none of its "GH clauses" are "negative@@Horn".  
\end{definition}

We will now use our \PTimeFull{} algorithm for minimising the number of pairs in a "Rabin condition" to minimise the number of clauses in a "GH formula" (Proposition~\ref{prop-Horn:minimisation-Horn}). 
We start by applying it to minimise the number of clauses of "simple" "GH formulas".

In all that follows we will not distinguish valuations $\nu : \mathrm{Var} \to \set{\top, \bot}$ from the corresponding subsets of variables $\set{v \in \mathrm{Var} \mid \nu(v) = \top}$.

\begin{lemma}
	\label{lem:min-simple-horn}
	There is a polynomial-time algorithm that minimises the number of clauses of a "simple" "GH formula".
\end{lemma}

\begin{proof}
	It suffices to observe that there is a correspondence between "simple" "GH formulas" and "Streett conditions".	
	Define the function $\alpha$ that turns a "GH clause" $(x_1 \land \cdots \land x_n) \implies (y_1 \land \cdots \land y_m)$ into the "Rabin pair" $(\set{y_1, \ldots, y_m}, \set{x_1, \ldots, x_n})$.
	We extend it into a function turning "simple" "GH formulas" into families of "Rabin pairs" by defining $\alpha(\bigwedge_{i=1}^k \mathrm{GH}_i) = (\alpha(\mathrm{GH}_i))_{i=1}^k$, with its associated "Streett language".
	We can then observe that $\alpha$ is a bijection (we consider Boolean formulas up to commutation of the terms, for instance we consider that $\varphi \lor \psi$ and $\psi \lor \varphi$ are the same formula).
	We also note that the number of clauses of a "simple" "GH formula" is the number of pairs of its image by $\alpha$.

	Finally, note that for all "simple" "GH formula" $\varphi$, the set of sets "accepted@@Streett" by the "Streett condition" $\alpha(\varphi)$ is $\set{\nu^{-1}(\bot) \mid \nu \text{ satisfies } \varphi}$.
	As a result, two "simple" "GH formula" are equivalent if and only if their images by $\alpha$ define the same "Streett language".

	In conclusion, in order to minimise the number of clauses of a "simple" "GH formula", one can simply apply $\alpha$ to it, minimise the number of pairs in the resulting "Streett condition", and then apply $\alpha^{-1}$.
\end{proof}

The extension to all "Generalised Horn formulas" is essentially a technicality, due to the fact that "negative@@Horn" clauses cannot be directly translated into "Rabin pairs" as in the previous proof. We circumvent this problem by replacing them with some non-"negative@@Horn" clauses and proving that minimising the initial "Horn formula" comes down to minimising the resulting "simple" one.

\begin{proposition}\label{prop-Horn:minimisation-Horn}
	There is a polynomial-time algorithm to minimise the number of clauses of a  "GH formula".
\end{proposition}

\begin{proof}
	Let $\varphi$ be a "GH formula", $V$ the set of variables appearing in it. If $\varphi$ does not contain any "negative@@Horn" clause, then it is satisfied by the valuation mapping every variable to $\top$ and thus can only be equivalent to "simple" "GH formulas". We can thus apply Lemma~\ref{lem:min-simple-horn} directly.

	Let $\psi$ be a "simple" "GH formula" and $N_1, \ldots, N_k$ "negative@@Horn" "Horn clauses", with $k>0$, such that $\varphi = \psi \land \neg N_1 \land \cdots \land \neg N_k$. We add a fresh variable $x_\bot$ that will play the role of $\bot$. For all $i \in [1,k]$, let $x^i_1, \ldots, x^i_{p(i)}$ be such that $N_i = (x^i_1 \land \cdots \land x^i_{p(i)} \implies \bot)$ and let $C_i = (x^i_1 \land \cdots \land x^i_{p(i)} \implies x_\bot)$.
	Define $\Tilde{\varphi} = \psi \land C_1 \land \cdots \land C_k \land (x_\bot \Rightarrow \bigwedge_{y \in V} y)$. 
	Note that the valuations satisfying $\Tilde{\varphi}$ are exactly the ones mapping $x_{\bot}$ to $\bot$ and whose projection on the other variables satisfies $\varphi$, plus the one mapping every variable to $\top$.
	
	As $\Tilde{\varphi}$ is "simple", we can apply Lemma~\ref{lem:min-simple-horn} to obtain an equivalent "simple" "GH formula" $\Tilde{\varphi}_{\min}$ with a minimal number of clauses.
	We define $\varphi_{\min}$ as this formula where every clause with $x_{\bot}$ on the left side has been removed and every clause of the form $(x_1 \land \cdots \land x_n) \implies (y_1 \land \cdots \land y_m)$ where one of the $y_i$ is $x_\bot$ has been replaced by $(x_1 \land \cdots \land x_n \implies \bot)$.
	
	As $\Tilde{\varphi}$ is not satisfied by the valuation mapping $x_\bot$ to $\top$ and all other variables to $\bot$, at least one clause in $\Tilde{\varphi}_{\min}$ has an $x_\bot$ on the left, hence $\varphi_{\min}$ has less clauses than $\Tilde{\varphi}_{\min}$.
	
	We have to argue that $\varphi$ and $\varphi_{\min}$ are equivalent, and that $\varphi_{\min}$ is minimal with respect to the number of clauses.
	First let us show that $\varphi$ and $\varphi_{\min}$ are equivalent. Let $\nu$ be a valuation, we write $\nu_{\bot}$ for the valuation mapping $x_\bot$ to $\bot$ and matching $\nu$ on $V$. We have 
	\[\nu \text{ satisfies } \varphi_{\min} \;\; \iff \;\; \nu_\bot \text{ satisfies } \Tilde{\varphi}_{\min} \;\;
	\iff \;\; \nu_\bot \text{ satisfies } \Tilde{\varphi}  \;\; \iff \;\; \nu \text{ satisfies } \varphi.\]
	

	Then let us prove that $\varphi_{\min}$ is minimal with respect to the number of clauses. Assume by contradiction that we have a "GH formula" $\varphi'$ equivalent to $\varphi$ and with less clauses than $\varphi_{\min}$.
	Then we can replace every "negative@@Horn" clause $(\neg x_1 \lor \cdots \lor \neg x_n)$ in $\varphi'$ by a clause $(x_1 \land \cdots \land x_n) \implies x_{\bot}$ and add a clause $(x_{\bot} \implies \bigwedge_{y \in V} y)$ to get a "simple" "GH formula" $\varphi''$ equivalent to $\Tilde{\varphi}$ and with less clauses than $\Tilde{\varphi}_{\min}$. This contradicts the minimality of $\Tilde{\varphi}_{\min}$.
	
	Hence $\varphi_{\min}$ has a minimal number of clauses.
\end{proof}

\AP Given a "GH formula" $\varphi$ using variables in $\GG$, its ""GH language"" is
\[\intro*\GHC{\varphi}{\GG} = \set{w\in \GG^\oo \mid \minf(w)\models \varphi}.\]


\begin{proposition}\label{prop-Horn:emptiness}
	Checking emptiness of an automaton with an "acceptance condition" represented by a "GH formula" is \NP-complete.
\end{proposition} 

\begin{proof}
	
	The \NP~upper bound follows from the one on Emerson-Lei conditions.
	
	For the hardness, we reduce from the Hamiltonian cycle problem.
	Let $G = (V,E)$ be a directed graph. For all edge $e \in E$ we write $\mathrm{src}(e)$ for its first vertex and $\mathrm{tgt}(e)$ for the second one. We define the "automaton" $\mathcal{A} = (Q,  q_{\init}, \SS, \DD, \GG, \col, W)$ as follows:
	\begin{itemize}
		\item $Q = \set{v^-, v^+ \mid v \in V}$, and we pick an arbitrary $v \in V$ and set $q_{\init} = v^-$.
		
		\item $\SS = \GG = \set{l_v \mid v \in V} \cup \set{l_e \mid e\in E} \cup \set{l_\bot}$, every transition is coloured with the letter it reads.
		
		\item $\Delta = \set{(v^-, l_v, v^+) \mid v \in V} \cup \set{(\mathrm{src}(e)^+, l_e, \mathrm{tgt}(e)^-) \mid e \in E}$.
		
		\item $W = \GHC{\varphi}{\GG}$ with 
		\[\varphi =  \Big[ \bigwedge_{\substack{e\neq e' \in E \\ \mathrm{src}(e) = \mathrm{src}(e')}} (l_e \land l_{e'} \implies l_{\bot})\Big] \land \Big[\bigwedge_{v \in V} (l_v \implies \bigwedge_{v' \in V} l_{v'})\Big]. \]

	\end{itemize}  
	
	A run of $\mathcal{A}$ is a sequence $v^-_0 \xrightarrow{v_0} v^+_0 \xrightarrow{(v_0, v_1)} v^-_1 \xrightarrow{v_1} v^+_1 \xrightarrow{(v_1, v_2)} \cdots$. It is accepted if and only if all vertices are visited infinitely often and the run ultimately always selects the same edge from every vertex.
	The existence of such a run is equivalent to the existence of a Hamiltonian cycle.
\end{proof}

\section{An alternative reduction for Theorem~\ref{th-min:colorMin-Aut-NPhard}}
\label{sec-app:alternative-reduction}
We provide an alternative $\NP$-harness reduction for the problem \pbColorMinAut{}, obtaining another proof for Theorem~\ref{th-min:colorMin-Aut-NPhard}.
The interest of this reduction is that it uses an automaton with only $2$ states, and it brings to light the difficulty to combine the structure of the "local subtrees" of the "ACD" to minimise the number of colours. 

We reduce from the problem \pbMaxClique{} defined as follows.
\AP A ""clique"" of $G$ is a subset $V'\subseteq V$ such that  $\{v',u'\}\in E$ for every $v'\neq u'\in V'$.
\AP The problem \intro*\pbMaxClique{} consists in, given a "graph" $G$ (that can be assumed connected) and a positive integer $k$, decide whether $G$ contains a "clique" of size $k$.
The problem \pbMaxClique{} is well-known to be $\NPc$~\cite{Karp72Reducibility}.

Let $G=(V,E)$ be a "simple", connected "undirected" graph and $k\in \NN$. We consider the "automaton" $\A_{G,k}$ defined as:
\begin{itemize}
	\item It has two states $q_{\mathsf{vert}}$ (which is initial) and $q_{k}$.
	\item The "input alphabet" is $\SS= V\cup A_k \cup \{x\}$, where $A_k$ is a set of size $k$ disjoint from $V$ and $x$ is a fresh letter.
	\item The set of "output colours" is $\GG = V\cup A_k$.
	\item The transitions of $\A_{G,k}$ are given by:
	\begin{itemize}
		\item $q_{\mathsf{vert}} \re{v:v} q_{\mathsf{vert}}$ for every $v\in V$,
		\item $q_{\mathsf{vert}} \re{x:y} q_{k}$ (where $y\in \Gamma$ is irrelevant), and
		\item $q_k\re{a:a}q_k$ for every $a\in A_k$.
	\end{itemize}
	
	\item Its "acceptance condition" is the "Muller language" associated to the family:
	\[ \F = E \cup \{\{a,a'\} \mid a,a'\in A_k, \, a\neq a'\}. \]
\end{itemize} 

The "representation@@Muller" of this "automaton" is polynomial in $|G|+k$, since $|\F| = \O(|E| + k^2)$. We also note that the "Zielonka tree" of $\F$ has size $\O(|E| + k^2)$.

We will use the following property satisfied by $\Lang{\A_{G,k}}$:
\begin{itemize}
	\item For all $\aa\in \SS$, words ending by $\aa^\oo$ are not in $\Lang{\A_{G,k}}$ ("cycles" consisting in a single self loop are "rejecting@@cycle").
	\item For all $a,b\in A_k$, $a\neq b$, $x(ab)^\oo \in \Lang{\A_{G,k}}$.
\end{itemize}

\begin{lemma}
	$G$ admits a "clique" of size $k$ if and only if $\A_{G,k}$ is "$|V|$-colour type@@TS". 
\end{lemma}
\begin{proof}
	Assume that $V' = \{v_1',\dots, v_k'\}$ is a "clique" of size $k$ of $G$, and let $A_k = \{a_1, \dots, a_k\}$.
	We consider the "Muller condition" using as set of colours $\GG' = V$ and given by $\F' = E$.
	The new "acceptance condition" over $\A_{G,k}$ is obtained by using the same "colouring@@TS" for the self loops over $q_{\mathsf{vert}}$, and recolouring self loops $q_k\re{a_i:a_i}q_k$ with $q_k\re{a_i:v_i'}q_k$.
	It is immediate that the obtained "acceptance condition" is "equivalent to@@acc" the original one of $\A_{G,k}$.
	
	For the converse, assume that $\A_{G,k}$ is "$|V|$-colour type@@TS". Then there is a set $\GG'$ of $|V|$ colours and a "colouring function" $\colAut'\colon \DD \to \GG'$ yielding an "equivalent@@cond" condition over $\A_{G,k}$. 
	
	First, we show that for two different self loops $e_1=q_{\mathsf{vert}}\re{v_1:c_1} q_{\mathsf{vert}}$ and $e_2 = q_{\mathsf{vert}}\re{v_2:c_2}q_{\mathsf{vert}}$,  we have $c_1\neq c_2$ (where $c_1=\colAut'(e_1)$ and $c_2=\colAut'(e_2)$).
	If $\{v_1,v_2\}\in E$, this is clear, as $\{e_1\}$ is a "rejecting cycle", but $\{e_1,e_2\}$ is "accepting@@cycle". 
	Suppose that  $\{v_1,v_2\}\notin E$, and let $u\in V$ such that $\{v_1,u\}\in E$ (which exists as $G$ is connected). Then, the "cycle" $\{e_1, q_{\mathsf{vert}}\re{u} q_{\mathsf{vert}} \}$ is "accepting@@cycle" while $\{e_1,e_2,q_{\mathsf{vert}}\re{u} q_{\mathsf{vert}}\}$ is "rejecting@@cycle", so they cannot be coloured equally.
	Therefore, for each colour $c\in \GG'$ there is one self loop $v$ such that $\colAut'(v)= c$.
	
	Secondly, we remark that for two different self loops $e_1=q_{k}\re{a_1:c_1} q_{k}$ and $e_2 = q_{k}\re{a_2:c_2} q_k$ over $q_k$ it is also satisfied that $c_1=\colAut'(e_1) \neq c_2 = \colAut'(e_2)$, as $xa_1^\oo\notin \Lang{\A_{G,k}}$, but $x(a_1a_2)^\oo\in \Lang{\A_{G,k}}$. 
	Let $\{c_1,\dots,c_k\}$ be the $k$ different colours labelling the self loops over $q_k$.
	We obtain that the subset $\{v_1,\dots, v_k\}\subseteq V$ of vertices such $\colAut'(v_i)=c_i$ form a "clique" of size $k$ in $G$.
\end{proof}

\end{document}